%% file: main.tex
\documentclass[letterpaper,11pt]{article}
\usepackage[utf8]{inputenc}
\usepackage[english]{babel}

\usepackage{booktabs}
\usepackage{multirow}

\usepackage{xspace}
\usepackage{amsmath,amsfonts,amsthm,amssymb}
\usepackage{mathtools}
\usepackage{thmtools}
\usepackage{thm-restate}
\usepackage{bm}
\usepackage{verbatim}
\usepackage{algorithm}
\usepackage{algorithmicx}
\usepackage{subcaption}

\usepackage[margin=1in]{geometry}

\usepackage[dvipsnames,usenames]{xcolor}
\usepackage[colorlinks=true,pdfpagemode=UseNone,urlcolor=RoyalBlue,linkcolor=RoyalBlue,citecolor=OliveGreen,pdfstartview=FitH]{hyperref}

\usepackage{cleveref}

\usepackage{tcolorbox}

\usepackage{todonotes} 

\usepackage{mleftright}

\makeatletter
\newenvironment{breakalgo}{%
  \def\@fs@cfont{\bfseries}%
  \let\@fs@capt\relax%
  \par\noindent%
  \medskip%
  \rule{\linewidth}{.4pt}%
  \vspace{-3pt}%
  \vspace{-1.3\baselineskip}%
}{%
  \vspace{-.75\baselineskip}%
  \rule{\linewidth}{.4pt}%
  \medskip%
}
\makeatother

\declaretheorem{definition}

\declaretheorem{theorem}
\declaretheorem[sibling=theorem]{lemma}
\declaretheorem[sibling=theorem]{claim}
\declaretheorem[sibling=theorem]{corollary}
\declaretheorem[sibling=theorem]{proposition}

\declaretheorem{remark}

\input{macro}

\title{Explicit Directional Affine Extractors and Improved Hardness for Linear Branching Programs}

\author{}

\author{
Xin Li \thanks{Department of Computer Science, Johns Hopkins University, \texttt{lixints@cs.jhu.edu}. Supported by NSF CAREER Award CCF-1845349 and NSF Award CCF-2127575.}
\and
Yan Zhong \thanks{Department of Computer Science, Johns Hopkins University, \texttt{yzhong36@jhu.edu}. Supported by NSF CAREER Award CCF-1845349.}
}

\begin{document}

\begin{titlepage}
    \thispagestyle{empty}
    \maketitle
    \begin{abstract}
        \thispagestyle{empty}
        \input{abstract}

    \end{abstract}
\end{titlepage}

\tableofcontents
\thispagestyle{empty}
\newpage
\addtocontents{toc}{\protect\thispagestyle{empty}} 
\setcounter{page}{1}

\input{introduction}
\input{prelim}

\input{Lcond}

\input{daext2}

\input{ac0}

\input{open}
\input{acknowledgement}

\bibliographystyle{alpha}
\bibliography{daExt}

\appendix

\input{app-ac0-xor-srolbp}

\input{app-missing-proofs}

\end{document}

%% file: macro.tex
\newcommand{\bin}{\{0,1\}}
\newcommand{\abs}[1]{\left| #1 \right|}
\newcommand{\pbra}[1]{\left( #1 \right)}
\newcommand{\cbra}[1]{\left\{ #1 \right\}}
\newcommand{\sbra}[1]{\left[ #1 \right]}

\newcommand{\DADisp}{\mathsf{DADisp}}
\newcommand{\DAExt}{\mathsf{DAExt}}

\newcommand{\Cor}{\mathsf{Cor}}
\newcommand{\LSExt}{\mathsf{LSExt}}
\newcommand{\AffineSRExt}{\mathsf{AffineSRExt}}
\newcommand{\Enc}{\mathsf{Enc}}

\newcommand{\snmExt}{\mathsf{snmExt}}

\newcommand{\acExt}{\mathsf{AC^0}\text{-}\mathsf{Ext}}
\newcommand{\acbfExt}{\mathsf{AC^0}\text{-}\mathsf{BFExt}}

\newcommand{\LExt}{\mathsf{LExt}}
\newcommand{\acLExt}{\mathsf{AC}^0\text{-}\mathsf{LExt}}

\newcommand{\ldACB}{\mathsf{ldACB}}
\newcommand{\acCB}{\mathsf{AC}^0\text{-}\mathsf{CB}}
\newcommand{\acaffCB}{\mathsf{AC^0}\text{-}\mathsf{AffCB}}
\newcommand{\FFAssign}{\mathsf{FFAssign}}

\newcommand{\scond}{\mathsf{SCond}}
\newcommand{\bcond}{\mathsf{BasicCond}}
\newcommand{\bgcond}{\mathsf{BasicGCond}}
\newcommand{\sgcond}{\mathsf{SGCond}}
\newcommand{\nipm}{\mathsf{NIPM}}
\newcommand{\Slice}{\mathsf{Slice}}
\newcommand{\flip}{\mathsf{flip}\text{-}\mathsf{flop}}
\newcommand{\laExt}{\mathsf{laExt}}
\newcommand{\Ext}{\mathsf{Ext}}
\newcommand{\IP}{\mathsf{IP}}
\newcommand{\cp}{\mathsf{cp}}
\newcommand{\WROLBP}{\mathsf{WROLBP}}
\newcommand{\SROLBP}{\mathsf{SROLBP}}
\newcommand{\ROBP}{\mathsf{ROBP}}

\newcommand{\A}{\mathcal{A}}
\newcommand{\B}{\mathcal{B}}
\newcommand{\G}{\mathsf{G}}
\newcommand{\F}{\mathsf{F}}
\newcommand{\N}{\mathbb{N}}
\newcommand{\D}{\cal D}
\newcommand{\X}{{\cal X}}
\newcommand{\poly}{\mathsf{poly}}
\DeclareMathOperator{\spn}{span}
\newcommand{\eps}{\varepsilon}
\newcommand{\ac}{\mathsf{AC}}
\newcommand{\E}{\mathbf{E}}
\renewcommand{\Pr}{\mathbf{Pr}}

\newcommand{\AC}{\mathsf{AC}}
\newcommand{\avgH}{\widetilde{H}_\infty}

\newcommand{\Supp}{\mathsf{Supp}}

\newcommand{\post}{\mathsf{Post}}
\newcommand{\pre}{\mathsf{Pre}}

\newcommand{\Ker}{\mathsf{Ker}}
\newcommand{\Span}{\mathsf{Span}}

\newcommand{\zo}{\{0, 1\}}
\newcommand{\bits}{\{0, 1\}}

\newcommand{\BI}{\begin{itemize}}
\newcommand{\EI}{\end{itemize}}
\newcommand{\BE}{\begin{enumerate}}
\newcommand{\EE}{\end{enumerate}}
\newcommand{\eqdef}{\vcentcolon=}

\newcommand{\BT}{\begin{theorem}}   \newcommand{\ET}{\end{theorem}}
\newcommand{\BD}{\begin{definition}}   \newcommand{\ED}{\end{definition}}
\newcommand{\BCR}{\begin{corollary}} \newcommand{\ECR}{\end{corollary}}
\newcommand{\BCT}{\begin{constr}} \newcommand{\ECT}{\end{constr}}
\newcommand{\BL}{\begin{lemma}}   \newcommand{\EL}{\end{lemma}}
\newcommand{\BCM}{\begin{claim}}   \newcommand{\ECM}{\end{claim}}

%% file: abstract.tex
Affine extractors give some of the best-known lower bounds for various computational models, such as $\AC^0$ circuits, parity decision trees, and general Boolean circuits. However, they are not known to give strong lower bounds for read-once branching programs ($\ROBP$s). In a recent work, Gryaznov, Pudl\'{a}k, and Talebanfard (CCC' 22) introduced a stronger version of affine extractors known as directional affine extractors, together with a generalization of $\ROBP$s where each node can make linear queries, and showed that the former implies strong lower bound for a certain type of the latter known as strongly read-once linear branching programs ($\SROLBP$s). Their main result gives explicit constructions of directional affine extractors for entropy $k > 2n/3$, which implies average-case complexity $2^{n/3-o(n)}$ against $\SROLBP$s with exponentially small correlation. A follow-up work by Chattopadhyay and Liao (CCC' 23) improves the hardness to $2^{n-o(n)}$ at the price of increasing the correlation to polynomially large, via a new connection to sumset extractors introduced by Chattopadhyay and Li (STOC' 16) and explicit constructions of such extractors by Chattopadhyay and Liao (STOC' 22). Both works left open the questions of better constructions of directional affine extractors and improved average-case complexity against $\SROLBP$s in the regime of small correlation.

This paper provides a much more in-depth study of directional affine extractors, $\SROLBP$s, and $\ROBP$s. Our main results include:
\begin{itemize}    
    \item An explicit construction of directional affine extractors with $k=o(n)$ and exponentially small error, which gives average-case complexity $2^{n-o(n)}$ against $\SROLBP$s with exponentially small correlation, thus answering the two open questions raised in previous works.
    \item An explicit function in $\AC^0$ that gives average-case complexity $2^{(1-\delta)n}$ against $\ROBP$s with negligible correlation, for any constant $\delta>0$. Previously, no such average-case hardness is known, and the best size lower bound for any function in $\AC^0$ against $\ROBP$s is  $2^{\Omega(n)}$.
\end{itemize}

One of the key ingredients in our constructions is a new linear somewhere condenser for affine sources, which is based on dimension expanders. The condenser also leads to an unconditional improvement of the entropy requirement of explicit affine extractors with negligible error. We further show that the condenser also works for general weak random sources, under the Polynomial Freiman-Ruzsa Theorem in $\F_2^n$, recently proved by Gowers, Green, Manners, and Tao (arXiv' 23).

%% file: introduction.tex
\section{Introduction}
\label{sec:intro}
Randomness extractors are functions that extract almost uniform random bits from weak random sources that have poor quality. Although the original motivation of randomness extractors comes from bridging the gap between the quality of randomness required in typical applications and that available in practice, as pseudorandom objects, they turn out to have broad applications in computer science.\ For example, the kind of extractors known as \emph{affine extractors} are shown to be closely connected to complexity theory. Indeed, they give strong size lower bounds for $\AC^0$ circuits (constant depth circuits with NOT gates and unbounded fan-in AND, OR gates) by the standard switching lemma \cite{Hastad86}, and are shown to give exponential size lower bounds for DNF circuits with a bottom layer of parity gates, together with strong average-case hardness for parity decision trees \cite{CohenS16}. Via sophisticated gate elimination techniques, they also give the best-known size lower bounds for general Boolean circuits \cite{DK11, FGHK15, LY22}. We define affine extractors below.

\begin{definition}[Affine extractor] An $(n, k)$ affine source is the uniform distribution over some affine subspace with dimension $k$, of the vector space $\F^n_2$.\footnote{More generally, affine sources and affine extractors can be defined over any finite field, but in this paper we focus on the binary field $\F_2$.} A function $\Ext : \bin^n \to \bin^m$ is an affine extractor for entropy $k$ with error $\eps$ if for every $(n, k)$ affine source $X$, we have
\[ \Ext(X) \approx_{\eps} U_m,\]
where $U_m$ stands for the uniform distribution over $\bin^m$, and $\approx_{\eps}$ means $\eps$ close in statistical distance. We say $\Ext$ is explicit if it is computable by a polynomial-time algorithm. 
\end{definition}

However, affine extractors are not known to imply strong lower bounds for computational models that measure space complexity.\ For example, a natural model in this context is a branching program, which is a directed acyclic graph with one source and two sinks, and each non-sink node has out-degree $2$. To define the computation of the branching program, one marks each non-sink node with the index of an input bit, and labels the two outgoing edges by $0$ and $1$, respectively. Furthermore, one sink is labeled by $1$ and the other is labeled by $0$. The program now computes any input by following the natural path from the source to one sink, while reading the corresponding input bits and going through the corresponding edges. The program accepts the input if and only if the path ends in the sink with label $1$, and the size of the branching program is defined as the number of its nodes, which roughly corresponds to $2^{O(s)}$ where $s$ is the space complexity of the computation.

Proving non-trivial lower bounds of an explicit function for general branching programs turns out to be a challenging problem. The best known bound is $\Omega(\frac{n^2}{\log^2 n})$ \cite{Nechiporuk66} after decades of effort, which is not enough to separate $\mathsf{P}$ from $\mathsf{LOGSPACE}$. Thus, most research on lower bounds for branching programs has focused on restricted models, and the most well-studied is the model of \emph{read-once branching program}, where on any computational path, any input bit is read at most once. Exponential lower bounds are known in this model \cite{Wegener88, Zak84, DBLP:conf/fct/Dunne85, DBLP:journals/tcs/Jukna88, DBLP:journals/tcs/KrauseMW91, Simon1992ANL, doi:10.1137/S0097539795290349, DBLP:journals/ipl/Gal97, DBLP:journals/ipl/BolligW98, AndreevBCR99, DBLP:journals/tcs/Kabanets03}, however, it is not clear if affine extractors imply strong lower bounds here. For example, the inner product is a good affine extractor for any entropy $k> n/2$, but it can be computed by a read-once branching program of size $O(n)$.

In a recent work \cite{GryaznovPT:CCC:2022}, Gryaznov, Pudl\'{a}k, and Talebanfard introduced a generalization of affine extractors called \emph{directional affine extractors} and a generalization of standard read-once branching programs called \emph{read-once linear branching programs}, and show that explicit constructions of the former imply strong lower bounds for certain cases of the latter.\ We define the two generalizations below.

\begin{definition}[Directional affine extractor]\label{def:daext} A function $\DAExt : \bin^n \to \bin^m$ is a directional affine extractor for entropy $k$ with error $\eps$ if for every $(n, k)$ affine source $X$ and every non-zero vector $a \in \F^n_2$, we have
\[ (\DAExt(X), \DAExt(X+a))  \approx_{\eps} (U_m, \DAExt(X+a)).\]
We say the function is a (zero-error) directional affine disperser if there exists some $b \in \bin^m$ such that $$\Big|\Supp\pbra{\DAExt(X) \mid \DAExt(X+a)=b}\Big|=2^m$$.
\end{definition}

\begin{remark}
Our definition is slightly more general than the definition in \cite{GryaznovPT:CCC:2022}, since we allow the extractor to output more than one bits. In the special case of $m=1$, our definition implies that in \cite{GryaznovPT:CCC:2022}, the reverse is also true up to a small loss in parameters as shown in \cite{ChattopadhyayL:ccc:2023}.
\end{remark}

\begin{definition}[Linear branching program \cite{GryaznovPT:CCC:2022}] A linear branching program on $\F^n_2$ is a directed acyclic graph $P$ with the following properties:
\begin{itemize}
    \item There is only one source $s$ in $P$.
    \item There are two sinks in $P$, labeled with $0$ and $1$ respectively.
    \item Every non-sink node $v$ is labeled with a linear function $\ell_v : \F^n_2 \to \F_2$. Moreover, there are exactly two outgoing edges from $v$, one is labeled with $1$ and the other is labeled with $0$.
\end{itemize}
The size of $P$ is the number of non-sink nodes in $P$. $P$ computes a Boolean function $f : \bin^n \to \bin$ in the following way. For every input $x \in \F^n_2$, $P$ follows the computation path by starting from $s$, and when on a non-sink node $v$, moves to the next node following the edge with label $\ell_v(x) \in \bin$. The computation ends when the path ends at a sink, and $f(x)$ is defined to be the label on this sink.
\end{definition}

\cite{GryaznovPT:CCC:2022} defines two kinds of read-once linear branching programs ($\mathsf{ROLBP}$ for short). Specifically, given any linear branching program $P$ and any node $v$ in $P$, let $\pre_v$ denote the span of all linear queries that appear on any path from the source to $v$, excluding the query $\ell_v$. Let $\post_v$ denote the span of all linear queries in the subprogram starting at $v$.

\begin{definition}[Weakly read-once linear branching program] A linear branching program $P$ is weakly read-once if for every inner node $v$ of $P$, it holds that $\ell_v \notin \pre_v$.
\end{definition}

\begin{definition}[Strongly read-once linear branching program] A linear branching program $P$ is strongly read-once if for every inner node $v$ of $P$, it holds that $\pre_v \cap \post_v=\{0\}$.
\end{definition}

In this paper, we will focus on strongly read-once linear branching programs, and use $\SROLBP$ as a shorthand. As observed in \cite{GryaznovPT:CCC:2022} and \cite{ChattopadhyayL:ccc:2023}, even the more restricted $\SROLBP$s generalize several important and well-studied computational models, for example, decision trees, parity decision trees, and standard read-once branching programs. These models have applications in diverse areas, such as learning theory, streaming algorithms, communication complexity and query complexity. Thus, just as the natural generalizations from $\AC^0$ circuits to $\AC^0[\oplus]$ circuits ($\AC^0$ with parity gates), and from decision trees to parity decision trees, studying the generalization from $\ROBP$s to $\mathsf{ROLBP}$s is also a natural direction. In addition, as observed in \cite{GryaznovPT:CCC:2022}, parity decision trees are the only case in $\AC^0[\oplus]$ for which we have strong average-case lower bounds, and they are closely related to tree-like resolution refutation proof systems. Thus studying $\mathsf{ROLBP}$s as a generalization of parity decision trees is of particular interest (in fact, this is the original motivation in \cite{GryaznovPT:CCC:2022}). We now define two complexity measures of  $\SROLBP$s below.

\begin{definition}
    For a Boolean function $f: \bin^n \to \bin$, let $\SROLBP(f)$ denote the smallest possible size of a strongly read-once linear branching program that computes $f$, and $\SROLBP_{\eps}(f)$ denote the smallest possible size of a strongly read-once linear branching program $P$ such that 
    $$\Pr_{x \leftarrow_U \F^n_2}[P(x)=f(X)] \geq \frac{1}{2}+\eps.$$ The definition can be adapted to $\ROBP$s naturally.
\end{definition}

The main contribution of \cite{GryaznovPT:CCC:2022} is to show that directional affine extractors give strong average-case hardness for $\SROLBP$s. Specifically, they show that for any directional affine extractor $\DAExt$ for entropy $k$ with error $\eps$, we have $\SROLBP_{\sqrt{\eps/2}}(\DAExt) \geq \eps 2^{n-k-1}$. In addition, they give an explicit construction of directional affine extractor for $k \geq \frac{2n}{3}+c$ with $\eps \leq 2^{-c}$, which also implies exponential average-case hardness for $\SROLBP$s of size up to $2^{\frac{n}{3}-o(n)}$.\ Thus, directional affine extractors are indeed stronger than standard affine extractors and give strong lower bounds in more computational models. \cite{GryaznovPT:CCC:2022} left open the question of explicit constructions of directional affine extractors for $k=o(n)$.

In a follow-up work, Chattopadhyay and Liao \cite{ChattopadhyayL:ccc:2023} showed that another kind of extractors, known as \emph{sumset extractors}, also give strong average-case hardness for $\SROLBP$s. These extractors were introduced by Chattopadhyay and Li  \cite{ChattopadhyayL16}, which are extractors that work for the sum of two (or more) independent weak random sources. By using existing constructions of such extractors in \cite{ChattopadhyayL22}, they give an explicit function $\Ext$ such that $\SROLBP_{n^{-\Omega(1)}} (\Ext) \geq 2^{n-\log^{O(1)} n}$, i.e., the branching program size lower bound becomes close to optimal, but the correlation increases from exponentially small to polynomially large. Similarly, \cite{ChattopadhyayL:ccc:2023} left open the question of obtaining improved average-case hardness against $\SROLBP$s in the small correlation regime.

We remark that directional affine extractors are a special case of \emph{affine non-malleable extractors}, which are defined by Chattopadhyay and Li \cite{ChattopadhyayL:stoc:2017}. Roughly, an affine non-malleable extractor is an affine extractor such that the output is still close to uniform, even conditioned on the output of the extractor where the input affine source is modified by any affine function with no fixed points.
In this context, directional affine extractors just correspond to the case where the tampering function adds a non-zero affine shift to the source. Previously, the best affine non-malleable extractor due to Li \cite{Li:focs:2023} works for entropy $k \geq (1-\gamma)n$ for some small constant $\gamma < 1/3$ with error $2^{-\Omega(n)}$. Thus this does not give a better construction of directional affine extractors.\ However, \cite{Li:focs:2023} does give an improved sumset extractor, which yields an explicit function $\Ext$ such that $\SROLBP_{\eps} (\Ext) \geq 2^{n-O (\log n)}$ for any constant $\eps>0$, i.e., the branching program size lower bound becomes optimal up to the constant in $O(.)$, but the correlation increases to any constant.

\subsection{Our Results}
In this paper, we present a much more in-depth study of directional affine extractors, affine non-malleable extractors, $\SROLBP$s, and standard $\ROBP$s. To begin with, we observe that it is not a priori clear that $\SROLBP$s are more powerful than standard $\ROBP$s. Indeed, it is easy to see that $\AC^0[\oplus]$ and parity decision trees are exponentially more powerful than $\AC^0$ circuits and standard decision trees, respectively, since parity requires exponential size $\AC^0$ circuits and decision trees. However, any parity function can be computed by an $\ROBP$ of size $O(n)$. Nevertheless, there are previous works \cite{Okolnishnikova93, Jukna95ITA, glinskih_et_al:LIPIcs.MFCS.2017.26} which showed that computing explicit characteristic functions of certain affine subspaces require $\ROBP$s of size $2^{\Omega(n)}$ (e.g., the satisfiable Tseitin formulas in \cite{glinskih_et_al:LIPIcs.MFCS.2017.26}). Since such functions are easily computable by an $\SROLBP$ of size $O(n)$, this provides a separation between $\SROLBP$ and $\ROBP$ and shows that indeed $\SROLBP$s are exponentially more powerful than $\ROBP$s.




In turn, this further demonstrates that directional affine extractors have stronger properties than standard affine extractors, as they imply strong lower bounds for $\SROLBP$s. Next, we give explicit constructions of directional affine extractors with much better parameters than that in \cite{GryaznovPT:CCC:2022}. Our construction works for any linear entropy with exponentially small error.

\begin{theorem}
    For any constant $0<\delta \le 1$, there exists a family of explicit directional affine extractors $\DAExt:\bin^n \to \bin^m$ for entropy $k \geq \delta n$ with error $\eps=2^{-\Omega(n)}$ and output length $m=\Omega(n)$.
\end{theorem}

In fact, our construction can work for slightly sub-linear entropy.

\begin{theorem} \label{thm:daextmain}
There exists a constant $c>1$ and an explicit family of directional affine extractors $\DAExt:\bin^n\to\bin^m$ for entropy $k \geq cn(\log\log\log n)^2/\log\log n$ with error $\eps=2^{-n^{\Omega(1)}}$ and output length $m=n^{\Omega(1)}$, as well as an explicit family of directional affine dispersers for entropy $k \geq cn(\log\log n)^2/\log n$ with $m=n^{\Omega(1)}$. 
\end{theorem}

This theorem immediately gives much improved average-case hardness for $\SROLBP$s.

\begin{theorem} \label{thm:rolbpmain}
There is an explicit function $\DAExt$ such that $\SROLBP_{2^{-n^{\Omega(1)}}}(\DAExt) \geq 2^{n-\widetilde{O}(\frac{n}{\log \log n})}$, where $\Tilde{O}(.)$ hides $(\log\log\log n)^2$ factors.
\end{theorem}

In particular, we can achieve exponentially small correlation while obtaining a $2^{n-o(n)}$ size lower bound for $\SROLBP$s, which is almost optimal. This significantly improves the $2^{n/3-o(n)}$ size lower bound in \cite{GryaznovPT:CCC:2022} and the polynomially large correlation in \cite{ChattopadhyayL:ccc:2023}. Thus, Theorem~\ref{thm:daextmain} and \ref{thm:rolbpmain} provide positive answers to the two open questions in \cite{GryaznovPT:CCC:2022} and \cite{ChattopadhyayL:ccc:2023} mentioned before.

We remark that under our new definition, a directional affine extractor is strictly stronger than a standard affine extractor. Thus Theorem~\ref{thm:daextmain} also improves the entropy requirement of negligible error affine extractors, from the previously best-known result of $\frac{n}{\sqrt{\log \log n}}$ \cite{Yehudayoff11, Li:CCC:2011} to $\frac{cn(\log\log\log n)^2}{\log\log n}$.




We also revisit the hardness results for standard $\ROBP$s. As mentioned before, exponential and even close to optimal size lower bounds are known for explicit functions in this model, where the current best result is an explicit function that requires $\ROBP$s (in fact, $\SROLBP$s) of size $2^{n-O(\log n)}$ \cite{Li:focs:2023}. However, there has also been a lot of interest in finding functions in lower complexity classes that give strong lower bounds for $\ROBP$s. It is clear that the class $\mathsf{NC}^0$ is not sufficient. Thus the next possible class is $\AC^0$. Indeed there are previous works giving explicit $\AC^0$ functions that require $\ROBP$s of size $2^{\Omega(\sqrt{n})}$\cite{DBLP:journals/tcs/Jukna88, DBLP:journals/tcs/KrauseMW91, DBLP:journals/ipl/Gal97, DBLP:journals/ipl/BolligW98} and even $2^{\Omega(n)}$ \cite{glinskih_et_al:LIPIcs.MFCS.2017.26}, yet there is no average-case hardness as far as we know. Here, we improve both the size lower bound and the average-case hardness by giving an explicit $\AC^0$ function that has negligible correlation with $\ROBP$s of size $2^{(1-\delta)n}$ for any constant $\delta>0$.

\begin{theorem}
For any constant $\delta>0$ there is an explicit function $\acExt$ in $\AC^0$ such that $\ROBP_{2^{-\poly \log n}}(\acExt) \geq 2^{(1-\delta)n}$.
\end{theorem}

One of the key ingredients in our constructions is a new linear somewhere condenser for affine sources. Specifically, we have

\begin{definition}
For any $0 < \delta < \gamma < 1$, a function $\scond: \F_2^n \to (\F_2^m)^{\ell}$ is a $(\delta, \gamma)$ affine somewhere condenser, if it satisfies the following property: for any affine source $X$ over $\F_2^n$ with entropy $\delta n$, let $(Y_1, \cdots, Y_{\ell}) =\scond(X) \in (\F_2^m)^{\ell}$, then there exists at least one $i \in [\ell]$ such that $Y_i$ is an affine source over $\F_2^m$ with entropy at least $\gamma m$.
\end{definition}

\begin{theorem}
There exists a constant $\beta > 0$ such that for any $0< \delta \leq 1/2$, there is an explicit $(\delta, 1/2+\beta)$ affine somewhere condenser $\scond: \F_2^n \to (\F_2^m)^t$, where $t=\poly(1/\delta)$ and $m=n/\poly(1/\delta)$. Moreover, $\scond$ is a linear function. 
\end{theorem}

We further show that (a slight modification of) this condenser works for general weak random sources, under the well-known Polynomial Freiman-Ruzsa Theorem in $\F^n_2$, once one of the most important conjectures in additive combinatorics and very recently proved by Gowers, Green, Manners, and Tao~\cite{gowers2023conjecture}. See section~\ref{sec:cond-general} for details.

Previously, all condensers of this kind are based on sum-product theorems, and the function is a polynomial with degree $\poly(1/\delta)$~\cite{BarakKSSW05, Raz05, Zuckerman:toc:2007}. In contrast, there exist constructions of linear \emph{seeded} extractors, where if one lists the outputs of the extractor for all possible seeds, then we get a somewhere random source such that at least one output is close to uniform, and the function is a linear function. However, in many applications such as ours, one needs to use a somewhere condenser instead of simply listing all outputs of an extractor, since the former only gives a small number (e.g., a constant) of outputs as opposed to $\poly(n)$ outputs from the extractor. Hence, our linear somewhere condenser complements the existing sum-product theorem based somewhere condensers. Moreover, our construction of the condenser is based on \emph{dimension expanders}, which are algebraic pseudorandom objects previously studied based on their own interests, with no clear applications in computer science as far as we know. Thus, our construction can be viewed as one of the first applications of dimension expanders in computer science.

Finally, we study the question of whether directional affine extractors can give strong lower bounds for the class of $\AC^0[\oplus]$ in a black box way. Cohen and Tal \cite{CohenT:random:2015} showed via probablistic methods that standard affine extractors do not suffice since depth-$3$ $\AC^0[\oplus]$ circuits can compute optimal affine extractors. Using a slightly modified argument as that in \cite{CohenT:random:2015}, we show that even the stronger version of directional affine extractors does not suffice. Specifically, depth-$3$ $\AC^0[\oplus]$ circuits can also compute optimal directional affine extractors. This in turn provides a strong separation of $\AC^0[\oplus]$ from $\SROLBP$.

\begin{theorem}
    There exists a function $f: \bin^n \to \bin$ which is a directional affine extractor for entropy $k$ with error $\eps$, where $k=\log \frac{n}{\eps^2}+\log \log \frac{n}{\eps^2}+O(1)$ such that the following properties hold.
    \begin{enumerate}
        \item $f$ is a polynomial of degree $\log \frac{n}{\eps^2}+\log \log \frac{n}{\eps^2}+O(1)$.
        \item   $f$ can be realized by a $\mathsf{XOR}\text{-}\mathsf{AND}\text{-}\mathsf{XOR}$ circuit of size $O((n/\eps)^2\cdot \log^3(n/\eps))$.
        \item $f$ can be realized by a De Morgan formula of size $O((n^5/\eps^2)\cdot \log^3(n/\eps))$.
    \end{enumerate}
    \end{theorem}

\subsection{Overview of the Techniques}
Here we give a sketch of the main ideas used in this paper. For clarity, we shall be informal at places and ignore some technical details.

\paragraph{Directional affine extractors.} Our starting point is the construction of affine extractors by Li \cite{Li:CCC:2011}, which works for sub-linear entropy with exponentially small error. We first briefly recall the construction there. Divide an affine source $X$ of entropy rate $\delta$ into $O(1/\delta)$ blocks. By choosing the size of the blocks appropriately, one can show that there exists a ``good" block $X_g$ of entropy rate $\Omega(\delta)$, and the source $X$ still has a lot of entropy conditioned on $X_g$ (i.e., we get an affine \emph{block source}). If we know the position of $X_g$, randomness extraction is easy:  we apply a somewhere condenser (e.g., those in \cite{BarakKSSW05, Raz05, Zuckerman:toc:2007}) to condense $X_g$ into a matrix with a constant number of rows, such that at least one row has entropy rate $1-\delta/2$. At this point, we can apply a linear two-source extractor (e.g., the inner product function) to each row of the matrix and the source $X$ to get an affine \emph{somewhere random} source, conditioned on the fixing of $X_g$. This is another matrix with a constant number of rows, such that at least one row is uniform, and one can apply existing techniques to deterministically extract random bits from this source~\cite{Rao:ccc:09}. 

However, when $\delta$ is small, we don't know which block $X_g$ is good. Thus in ~\cite{Li:CCC:2011}, the construction tries all blocks, and then combines them together. To make this process work, the construction crucially maintains the following property: (*) for each block $X_i$, the output bits produced from this block are constant degree polynomials of the input bits, and the degrees decrease geometrically from the first block to the last block. With this property, the analysis goes by focusing on the first good block $X_g$. Notice that we can fix all the outputs produced from blocks before $X_g$, while all outputs produced from blocks after $X_g$ have degrees less than those from $X_g$. Thus if we take the XOR of all these outputs, an XOR lemma of polynomials \cite{v004a007, BhattacharyyaKSSZ10} guarantees the final output is still close to uniform. We note that the XOR lemma of polynomials only works for degree up to $\log n$. Hence it is important to keep the degree $c$ of the outputs from each block to be as small as possible. Roughly, we will need $c^{O(1/\delta)} < \log n$.

Our strategy now is to adapt this construction to directional affine extractors.\ Towards this, we use techniques from constructions of non-malleable extractors since, as we remark before, directional affine extractors are a special case of affine non-malleable extractors. Recent constructions of non-malleable extractors usually consist of two steps: first, generate a small advice that is different from the tampered version with high probability, and then use the advice together with other tools (e.g., correlation breakers) to achieve non-malleability. Thus, our goal is to adapt these two steps to directional affine extractors while, at the same time, still maintaining property (*), which is crucial to achieving any linear entropy or slightly sub-linear entropy. We now explain both steps.

As before, for each block $X_i$ we will get an output $U_i$, which is close to uniform if $X_i$ is a good block. Divide $U_i$ into two parts $U_i=U_{i1} \circ U_{i2}$. We will use $U_{i1}$ to generate the advice and $U_{i2}$ for the rest of the construction. Notice that from the tampered input $X'=X+a$ we also have a tampered version $U_i'=U_{i1}' \circ U_{i2}'$. In the following, we will always use letters with prime to denote the corresponding random variables produced from the tampered input. If $U_{i1} \neq U_{i1}'$ then we are done, otherwise we use $U_{i1}=U_{i1}'$ to sample some $\Omega(\delta^2 n)$ bits $H_i$ from an encoding of $X$, using an asymptotically good binary linear code. Since $X'=X+a$, we have that $H_i+H'_i$ basically corresponds to the sampled bits from the encoding of $a$. Thus $H_i \neq H'_i$ with high probability by the distance of the linear code. However, we cannot just do sampling naively since we need to keep the degree to be a constant. Therefore, we also divide both $U_{i1}$ and the encoding of $X$ into $\Omega(\delta^2 n)$ blocks where each block contains a constant number of bits, and use each block of $U_{i1}$ to sample one bit from the corresponding block of the encoding of $X$. By the distance property of the code, there are $\Omega(\delta^2 n)$ blocks of the encoding of $X$ and $X'$ that are different. Thus we still have $H_i \neq H'_i$ with high probability, and now each bit of $H_i$ is a constant degree polynomial of the bits of $U_{i1}$ and $X$. The advice string is now $U_{i1} \circ H_i$.

Once we have the advice, we can append it to another string extracted from $X$ by using a linear seeded extractor and $U_{i2}$ as the seed. Now notice that the string produced from $X$ is different from the string produced from $X'$ with high probability, and they are linearly correlated conditioned on the fixing of $(U_i, U_i')$. Thus we can apply, for example, a known affine non-malleable extractor (the state-of-the-art affine non-malleable extractor with negligible error only works for high entropy). However, the known construction of affine non-malleable extractor in \cite{ChattopadhyayL:stoc:2017} has super constant degree. Indeed, even one application of this extractor results in a polynomial of degree larger than $\log n$, which already defeats our purpose to get a directional affine extractor (we can still get a directional affine disperser, though). 


To solve this problem, we develop new ideas that make use of the special structure of $X'=X+a$. Recall that in our construction, for every block $X_i$ we get a $U_{i2}$, which is close to uniform if $X_i$ is good, and $X$ still has enough entropy conditioned on $X_i$. Our idea now is to use a \emph{seeded non-malleable extractor} $\snmExt$ instead, which is an extractor with a uniform random seed, such that if an adversary tampers with the seed but not the source, then the output of the extractor on the original inputs is close to uniform given the output on the tampered inputs. By appending the advice string to $U_{i2}$ and getting $\Tilde{U}_i=U_i \circ H_i$, we have $\Tilde{U}_i \neq \Tilde{U}'_i$ with high probability, and the seed $\Tilde{U}_i$ has high entropy if $H_i$ has small size, which suffices for the seeded non-malleable extractor as long as the extractor is strong. Now, if the seeded non-malleable extractor is also \emph{linear} conditioned on any fixing of the seed, then we have $\snmExt(X', \Tilde{U}'_i)=\snmExt(X, \Tilde{U}'_i)+\snmExt(a, \Tilde{U}'_i)$. Since $\snmExt(X, \Tilde{U}_i)$ is close to uniform given $\snmExt(X, \Tilde{U}'_i)$, and the extractor is strong (we can fix the seeds $(\Tilde{U}_i, \Tilde{U}'_i)$), this implies that $\snmExt(X, \Tilde{U}_i)$ is close to uniform given $\snmExt(X', \Tilde{U}'_i)$.  \footnote{The actual analysis involves more details since here $X$ is not independent of $(\Tilde{U}_i, \Tilde{U}'_i)$, but the property still holds due to the affine structure. We omit the details here.}

Luckily, there are previous constructions of linear seeded non-malleable extractors due to Li \cite{Li:focs:2012}, which are based on the inner product function. Moreover, this extractor also has the property that each output bit is a constant degree polynomial of the input bits. Thus everything seems to work out, except for one problem: the non-malleable extractor in \cite{Li:focs:2012} only works when the source has entropy rate $> 1/2$, but here our goal is to work for any linear (or slightly sub-linear) entropy. A natural idea would be to use the somewhere condenser (e.g., in \cite{BarakKSSW05, Raz05, Zuckerman:toc:2007}) to boost the entropy rate of $X$. However, all known condensers of this kind are based on sum-product theorems, which are non-linear functions, and applying them changes the structure of $X'=X+a$, which is important for our construction. Another idea is to apply a linear seeded extractor to $X$ and try all possible seeds. This indeed keeps the structure of $X'=X+a$, but will result in a $\poly(n)$ number of outputs, and combining them together will result in a polynomial of large, super constant degree. 

This motivates another key ingredient in our construction, a new linear somewhere condenser for affine sources. In short, we construct a linear function which, given any affine source on $n$ bits with entropy rate $0< \delta \leq 1/2$, outputs $\poly(1/\delta)$ rows such that each row has $n/\poly(1/\delta)$ bits, and at least one row has entropy rate $1/2+\beta$ for some absolute constant $\beta>0$. This complements the sum-product based somewhere condensers, and can be viewed as a separate contribution of our work. We will explain the construction of this condenser later, but finish the description of our directional affine extractor here, assuming that we have the linear somewhere condenser. 

The rest of the construction roughly goes as follows. We apply the linear somewhere condenser to the source $X$ to get a constant number of rows, then apply $\snmExt$ to each row using $\Tilde{U}_i$ as the seed. Thus we get a constant number of outputs such that at least one of them is close to uniform conditioned on the corresponding tampered output. Now we apply an \emph{affine correlation breaker} such as those in \cite{Li:stoc:17,ChattopadhyayGL:focs:2021,ChattopadhyayL22} to further break the correlations between different outputs, and combine these outputs together by taking the XOR. The correlation breaker guarantees that the final output is close to uniform conditioned on the tampered output. To keep the degree small, we need to replace all seeded extractors used in the correlation breaker with a constant degree linear seeded extractor in \cite{Li:CCC:2011}. This keeps the output bits to be constant degree polynomials of the input bits, and the remaining construction is essentially the same as that in \cite{Li:CCC:2011}.

\paragraph{Linear somewhere condenser.} We now describe our construction of the linear somewhere condenser. This is based on another pseudorandom object known as \emph{dimension expander}. Informally, a dimension expander is a set of linear mappings from a vector space $\F^n$ to itself, such that for any linear subspace $V \subset F^n$ with small dimension $k \leq n/2$, the span of the union of all the images of $V$ under the set of linear mappings has dimension at least $(1+\alpha)k$ for some absolute constant $\alpha>0$. Readers familiar with expander graphs can see that this is a linear algebraic analog of expander graphs. Thus, it is desirable to give explicit constructions of the set of linear mappings which has as few number of mappings as possible, where this number $d$ is called the degree. Dimension expanders were first introduced by Barak, Impagliazzo, Shpilka, and Wigderson \cite{BISW04}, who also showed the existence of such objects. Later, Bourgain and Yehudayoff \cite{Bourgain09, BourgainY13} gave explicit constructions of dimension expanders with degree $d=O(1)$ over any field. Interestingly, as far as we know, there are no previous applications of dimension expanders in computer science, and they are mainly studied based on their own interests and connections to other algebraic pseudorandom objects. Thus our construction can be viewed as one of the first applications of dimension expanders in computer science. 

Given an explicit dimension expander $\{T_i\}_{i \in [d]}$ where each $T_i$ is a linear mapping, and any affine source $X$ with entropy rate $\delta \leq 1/2$, we first construct a basic somewhere condenser as follows. Divide $X$ equally into $X=X_1 \circ X_2$, and our condenser produces $2d+2$ outputs: $(X_1, X_2, \{X_1+T_i(X_2)\}_{i \in [d]}, \{T_i(X_1)+X_2\}_{i \in [d]})$. We show that at least one output has entropy rate $(1+\gamma)\delta$ for some constant $\gamma>0$, and we give some intuition below. By the structure of affine sources, one can show that there exists another affine source $X_3$ independent of $X_1$ such that $X_2=X_3+L(X_1)$ for some linear function $L$. Let $H(X_1)=s$, $H(X_3)=r$ and $H(L(X_1))=t$, then we have $s+r=\delta n$. If either $s$ or $r$ is small, e.g., $s \ll \delta n/2$, then we must have $r \gg \delta n/2$ and thus $H(X_2)=r+t \geq (1+\gamma)\delta n/2$. Therefore the entropy rate of $X_2$ is at least $(1+\gamma)\delta$. The case of  $r \ll \delta n/2$ is similar. Hence, we only need to consider the case where $s \approx \delta n/2$ and $r \approx \delta n/2$, and notice that we must have either $s \leq \delta n/2$ or $r \leq \delta n/2$. Furthermore, in this case, $t$ must be small, since otherwise, we would again have $H(X_2)=r+t \geq (1+\gamma)\delta n/2$.

For simplicity, assume that $s =r= \delta n/2$, and $t=0$. Hence both $X_1$ and $X_2$ have entropy rate $\delta \leq 1/2$, and they are independent. Without loss of generality, assume the supports of both $X_1$ and $X_2$ are linear subspaces. By the property of the dimension expander, $\Span(\cup_{i \in [d]} T_i(X_1))$ has dimension at least $(1+\alpha) \delta n/2$. We now argue that there exists an $i \in [d]$ such that the support of $T_i(X_1)+X_2$ has dimension at least $(1+\alpha / d)\delta n/2$, which implies that $T_i(X_1)+X_2$ has entropy rate at least $(1+\alpha / d)\delta$. To see this, assume otherwise, then for any $i \in [d]$, any vector in the support of $T_i(X_1)+X_2$ can be expressed as a linear combination of the $r= \delta n/2$ basis vectors in the support of $X_2$ and $< (\alpha / d)\delta n/2$ other vectors. This implies that $\Span(\cup_{i \in [d]} T_i(X_1))$ has dimension $< \delta n/2 + d \cdot (\alpha / d)\delta n/2 = (1+\alpha) \delta n/2$, since any vector in $\Span(\cup_{i \in [d]} T_i(X_1))$ can be expressed as a linear combination of the $r= \delta n/2$ basis vectors in the support of $X_2$ and $< d \cdot (\alpha / d)\delta n/2$ other vectors. This contradicts the property of the dimension expander.

Thus, in all cases, we get the desired entropy rate boost. Our final somewhere condenser involves repeated uses of the basic condenser, as in previous works. It is easy to see that the entropy rate of at least one output will increase to $1/2+\beta$ for some absolute constant $\beta>0$ after $O(\log (1/\delta))$ uses of the basic condenser. The number of outputs is, therefore, $\poly(1/\delta)$ and each output has $n/\poly(1/\delta)$ bits. Finally, it is clear that the condenser is a linear function.

Once we have this linear condenser, we can even replace the somewhere condensers used in \cite{Li:CCC:2011} by the new condenser. This further reduces the degree of the polynomials of the output bits (since previous somewhere condensers are polynomials instead of linear functions). Therefore we can push the entropy requirement of our directional affine extractor to be even better than that in \cite{Li:CCC:2011}, from $\frac{n}{\sqrt{\log \log n}}$ to $\frac{cn(\log\log\log n)^2}{\log\log n}$.


We show that a slight modification of our linear condenser also works for general weak random sources, under the Polynomial Freiman-Ruzsa Theorem. Roughly, the idea is to use a careful analysis of subsources and collision probability. Specifically, it is known that if the collision probability of a distribution is small, then the distribution is close to having high min-entropy. On the other hand, if the collision probability is large, then (without loss of generality) assuming the distribution is the uniform distribution over some unknown subset, existing results in additive combinatorics imply that there is a large subset $A$ in the support of the distribution such that the size of $A+A$ is not much larger than $A$. The Polynomial Freiman-Ruzsa Theorem then implies that there is another large subset $A' \subset A$ which is ``close" to an affine subspace, which roughly reduces the analysis to the case of affine sources. See section~\ref{sec:cond-general} for the details.

\paragraph{$\AC^0$ average-case hardness for $\ROBP$s.} To show $\AC^0$ average-case hardness for $\ROBP$s, we use a standard observation that if one conditions on an inner node, then the input bits prior to this node and the input bits after this node are still independent. We then construct an appropriate extractor in $\AC^0$, which we call $\acExt$, for sources with such a structure. Specifically, given any $\ROBP$ of size $s$ and any constant $\delta>0$, we can find a cut or anti-chain (a maximal subset of vertices such that none of which is an ancestor of any other vertex) of size $O(s)$ at roughly depth $\delta n$ above the sinks, so that conditioned on the fixing of any vertex in the cut, the input uniform random string $X$ now becomes two independent weak sources $A$ and $B$, where $A$ corresponds to the first part of the program and $B$ corresponds to the second part. Since we don't know the order of bits queried by the $\ROBP$, the bits of the two sources are interleaved, and we view $X=A+B$. Using a standard averaging argument, one can show that with high probability, the following properties are satisfied: (1) $A$ and $B$ are supported on disjoint subsets of input bits; (2) $A$ has min-entropy roughly $(1-\delta)n-\log s$ and $B$ has min-entropy $\delta n$; and (3) $B$ is an oblivious bit-fixing source, which is obtained by fixing some unknown bits in a uniform random string. If $s \leq 2^{(1-2\delta)n}$ then both $A$ and $B$ have entropy rate roughly $\delta$. Now, our goal is to construct an extractor in $\ac^0$ for sources with this structure, that is also \emph{strong} in $B$. This means that even if we condition on the fixing of the vertex in the cut and $B$, the output of the extractor is still close to uniform. On the other hand, the output of the $\ROBP$ is completely determined by the vertex and $B$. Thus our extractor is average-case hard for $\ROBP$s of size up to $2^{(1-2\delta)n}$.

As usual, the function $\acExt$ will be compositions of different, more basic extractors as building blocks. Thus we need all these building blocks to be computable in $\ac^0$. Here, we leverage the constructions from two previous works on extractors in $\AC^0$: (1) the $\ac^0$-computable extractors $\acbfExt$ for bit-fixing source by Cheng and Li~\cite{ChengLi18}, and (2) the $\ac^0$-computable strong linear seeded extractors $\acLExt$ by Papakonstantinou, Woodruff, and Yang~\cite{papakonstantinou2016true}.

Now we can describe our main idea of construction. Divide $X$ into $t=O(1/\delta)$ blocks, and by an averaging argument, there exists a block $B_g$ of $B$ with entropy rate $\Omega(\delta)$. Now for the block $X_g=A_g+B_g$, we can fix $A_g$ so that $X_g$ is an oblivious bit-fixing source of entropy rate $\Omega(\delta)$ and is a deterministic function of $B$. We next fix the bits from $B$ outside of the $g$-th block so that the source $X$ outside of $X_g$ is a deterministic function of $A$ and thus independent of $X_g$. Moreover, $A$ and $X$ still have enough entropy left. 

Applying the above-mentioned extractor $\acbfExt$ for bit-fixing sources to each block $X_i$, we convert $X$ into a \emph{somewhere random source} $Y=Y_1\circ \cdots \circ Y_t$ where the row $Y_g$ is a deterministic function of $B_g$ and close to uniform, while all the other rows are deterministic functions of $A$. At this point, we can simply take the XOR of the $Y_i$'s to obtain a close-to-uniform output. However, as mentioned before, we need the extractor to be strong in $B$ and this simple approach is not sufficient.\ Instead, we fix all the outputs produced by $\acbfExt$ for $X_i$ where $i \neq g$. Note that these are all deterministic functions of $A$.\ Thus conditioned on this fixing, $Y$ becomes a deterministic function of $B$, which is independent of $A$.\ Moreover, as long as the output size of $\acbfExt$ is not too large, $A$ still has enough entropy left.\ Since $X=A+B$, we can now apply a \emph{strong} $t$-affine correlation breaker as in~\cite{Li:stoc:17,ChattopadhyayL22} with each $Y_i$ as the seed to extract from $X$ a random string, and take the XOR of them.\ The property of the correlation breaker guarantees that the string produced from $Y_g$ and $X$ is close to uniform conditioned on all the other outputs and $Y$. Hence the XOR is also close to uniform conditioned on $B$.\ To ensure the correlation breaker is computable in $\ac^0$, we replace all the strong (linear) seeded extractors in the known constructions of $t$-affine correlation breakers with the above-mentioned $\acLExt$. Since $t=O(1/\delta)$ is a constant, the correlation breaker involves a constant number of compositions of $\acLExt$, which is still in $\ac^0$. 

\subsection{Organization of the Paper}
The rest of the paper is organized as follows. In Section~\ref{sec:prelim} we give some preliminary knowledge and some primitives from prior works. 
In Section~\ref{sec:cond-affine} we describe our construction of linear somewhere random condenser for affine sources. Section~\ref{sec:cond-general} generalizes the construction to general weak sources under the Polynomial Freiman-Ruzsa Theorem. We give our construction of directional affine extractors in Section~\ref{sec:daext2}, and an $\ac^0$ computable extractor against $\ROBP$ in Section~\ref{sec:ac0}. We present some open problems in Section~\ref{sec:open}. In the appendix we show that depth-$3$ $\AC^0[\oplus]$ circuits can compute optimal directional affine extractors, and give some omitted proofs.

%% file: prelim.tex
\section{Preliminaries} \label{sec:prelim}
We often use capital letters for random variables and corresponding small letters for their instantiations. Let $s,t$ be two integers, $\{V_1^1,V_1^2,\cdots,V_1^t,V_2^1,V_2^2,\cdots, V_2^t,\cdots,V_s^1,V_s^2,\cdots,V_s^t\}$ be a set of random variables. We use $V_i^{[t]}$ to denote the subset $\{V_i^1,\cdots,V_i^t\}$ and $V_{[s]}^{j}$ to denote the subset $\{V_1^j,\cdots,V_s^j\}$. We use $V_{[s]}^{[t]}$ as a shorthand for the whole set of random variables. We also use $i_{[t]}$ to denote the set of indices $\{i_1,i_2,\cdots,i_t\}$.
Let $|S|$ denote the cardinality of the set~$S$. For $\ell$ a positive integer,
$U_\ell$ denotes the uniform distribution on $\zo^\ell$. When used as a component in a vector, each $U_\ell$ is assumed independent of the other components.

Let $\F_q$ denote the finite field of size $q$.
All logarithms are to the base 2.

\subsection{Probability Distributions and Entropy}
\begin{definition} [Statistical distance]Let $W$ and $Z$ be two distributions on
a set $S$. Their \emph{statistical distance} (variation distance) is
\[\Delta(W,Z) \eqdef \max_{T \subseteq S}(|W(T) - Z(T)|) = \frac{1}{2}
\sum_{s \in S}|W(s)-Z(s)|.
\]
\end{definition}

We say $W$ is $\eps$-close to $Z$, denoted $W \approx_\eps Z$, if $\Delta(W,Z) \leq \eps$. Let $V$ also be a distribution on the set $S$. We sometimes use $W \approx_\eps Z \mid V$ as a shorthand for $(W,V)\approx_\eps(Z,V)$. We will use this two notations interchangeably throughout the paper.
For a distribution $D$ on a set $S$ and a function $h:S \to T$, let $h(D)$ denote the distribution on $T$ induced by choosing $x$ according to $D$ and outputting $h(x)$.

\BL \label{lem:sdis}
For any function $\alpha$ and two random variables $A, B$, we have $\Delta(\alpha(A), \alpha(B)) \leq \Delta(A, B)$.
\EL

\begin{definition}[Min-entropy]
    The \emph{min-entropy} of a random variable $X$ is defined as 
    \[
        H_\infty(X) = \min_{x\in\Supp(X)}\cbra{-\log\Pr[X=x]}.
    \]
\end{definition}
For a random variable $X\in \bin^n$, we say it is an $(n,k)$-source if $H_\infty(X) \ge k$. The entropy rate of $X$ is defined as $H_\infty(X)/n$.

\subsection{Somewhere Random Sources and Extractors}

\begin{definition} [Somewhere random sources] \label{def:SR} A source $X=(X_1, \cdots, X_t)$ is $(t \times r)$
  \emph{somewhere-random} (SR-source for short) if each $X_i$ takes values in $\bits^r$ and there is an $i$ such that $X_i$ is uniformly distributed.
\end{definition}

\BD
An elementary somewhere-k-source is a 	vector of sources $(X_1, \cdots, X_t)$, such that some $X_i$ is a $k$-source. A somewhere $k$-source is a convex combination of elementary somewhere-k-sources.
\ED

\BD
A function $C: \bits^n \times \bits^d \to \bits^m$ is a $(k \to \ell, \eps)$-condenser if for every $k$-source $X$, $C(X, U_d)$ is $\eps$-close to some $\ell$-source. When convenient, we call $C$ a rate-$(k/n \to \ell/m, \eps)$-condenser.   
\ED

\BD
A function $C: \bits^n \times \bits^d \to \bits^m$ is a $(k \to \ell, \eps)$-somewhere-condenser if for every $k$-source $X$, the vector $(C(X, y)_{y \in \bits^d})$ is $\eps$-close to a somewhere-$\ell$-source. When convenient, we call $C$ a rate-$(k/n \to \ell/m, \eps)$-somewhere-condenser.   
\ED

\begin{definition}[Seeded extractor]\label{def:strongext}
A function $\Ext : \bits^n \times \bits^d \rightarrow \bits^m$ is  a \emph{strong $(k,\eps)$-extractor} if for every source $X$ with min-entropy $k$
and independent $Y$ which is uniform on $\zo^d$,
\[ (\Ext(X, Y), Y) \approx_\eps (U_m, Y).\]
\end{definition}

\subsection{The Structure of Affine Sources}

In this paper, affine sources encompass uniform distributions over linear subspaces and by affine functions we sometimes mean affine-linear functions.

\begin{definition}[Affine source]
    Let $\F_q$ be the finite field with $q$ elements. Denote by $\F_q^n$ the $n$-dimensional vector space over $\F_q$. A distribution $X$ over $\F_q^n$ is an $(n,k)_q$ affine source if there exist linearly independent vectors $a_1,\cdots,a_k \in \F_q^n$ and another vector $b\in\F_1^n$ s.t. $X$ is sampled by choosing $x_1,\cdots,x_k\in\F$ uniformly and independently and computing
    \begin{align*}
        X = \sum_{i=1}^k x_ia_i + b.
    \end{align*}
\end{definition}

The min-entropy of affine source coincides with its standard Shannon entropy, we simply use $H(X)$ to stand for the entropy of an affine source $X$.

The following lemma is a slight generalization of its version in~\cite{Li:CCC:2011}, where we show that $L$ can be an affine function instead of just a linear function. We also prove that the entropy of $X$ is constant conditioned on any fixing of $L(X)$. The readers are referred to Appendix~\ref{app:missing-proofs} for a proof. 

\begin{restatable}[Affine conditioning~\cite{Li:CCC:2011}]{lemma}{affineconditioning}
    \label{lemma:affine conditioning}
    
    Let $X$ be any affine source on $\{0,1\}^n$. Let $L:\{0,1\}^n \to \{0,1\}^m$ be any affine function. Then there exist independent affine sources $A,B$ such that:
    \begin{itemize}
        \item $X = A + B$
        \item There exists $c\in \{0,1\}^m$, such that for every $b\in \Supp(B)$, it holds that $L(b) = c$.
        \item $H(A) = H(L(A))$ and there exists an affine function $L^{-1}: \{0,1\}^m \to \{0,1\}^n$ such that $A = L^{-1}(L(A))$.
        \item $H(X \mid_{L(X)=\ell}) = H(B)$ for all $\ell\in \Supp(L(X))$.
    \end{itemize}
\end{restatable}

The following definition is a specialization of conditional min-entropy for affine sources. It is well-defined by Lemma~\ref{lemma:affine conditioning}.
\begin{definition}[Conditional min-entropy for affine sources]
Let $W$ and $Z$ be two affine sources. Define 
\begin{align*}
    H(W\mid Z) = H(W\mid_{Z=z}),\;\forall z\in \Supp(Z).
\end{align*}
\end{definition}

\begin{lemma}\label{lemma:affine bound}
Let $X, Y, Z$ be affine sources. Then
$H(X \mid (Y, Z)) \ge H(X \mid Z) - \log(\Supp(Y))$.
\end{lemma}

We will also need the following lemma from~\cite{Li:CCC:2011} when we do sequential conditioning on blocks of an affine source or argue about the total entropy of blocks of an affine source.

\begin{lemma}[Affine entropy argument~\cite{Li:CCC:2011}]
    \label{lemma:affine entropy}

    Let $X$ be any affine source on $\bin^n$. Divide $X$ into $t$ arbitrary blocks $X=X_1\circ X_2 \circ \cdots \circ X_t$. Then there exists positive integers $k_1,\cdots,k_t$ such that,
    \begin{itemize}
        \item $\forall j,1\le j\le t$ and
        $\forall (x_1,\cdots,x_{j-1})\in \Supp(X_{1}, \cdots,X_{j-1})$, $H(X_{j}\mid_{X_{1}=x_{1},\cdots,X_{j-1}=x_{j-1}} )=k_{j}$;
        \item $\sum_{i=1}^t k_i = H(X)$.
    \end{itemize}
\end{lemma}

\subsection{Average Conditional Min-Entropy and Average-Case Seeded Extractors}

\begin{definition}[Average conditional min-entropy]
The \emph{average conditional min-entropy} is defined as 
\begin{align*}
    \widetilde{H}_\infty(X\mid W) &= -\log\pbra{\E_{w\leftarrow W}\sbra{\max_x \Pr[X=x\mid W=w]}} \\
    & = -\log\pbra{\E_{w\leftarrow W}\sbra{2^{-H_\infty(X \mid W=w)}}}.
\end{align*}
\end{definition}

\begin{lemma}[\cite{DodisORS08}]
For any $s>0$, $\Pr_{w\leftarrow W}[H_\infty(X\mid W=w)\ge \widetilde{H}_\infty(X\mid W)-s]\ge 1-2^{-s}$.
\end{lemma}

\begin{lemma}[\cite{DodisORS08}]
If a random variable $B$ has at most $2^\ell$ possible values, then $\widetilde{H}_\infty(A\mid B)\ge H_\infty(A)-\ell$.
\end{lemma}

\begin{lemma}[\cite{DodisORS08}]\label{lemma:avg-ext}
For any $\delta>0$, if $\Ext$ is a $(k,\eps)$ extractor, then it is also a $(k+\log(1/\delta),\eps+\delta)$ average case extractor.
\end{lemma}

\subsection{Alternating Extraction and Independence Merging}

The following techniques underpin the construction of correlation breakers.

\begin{definition}[$L$-alternating extraction]\label{def:l-look-ahead} 
    Let $W$ be an $(n_w,k_w)$-source and $(Q_1,\cdots,Q_L)$ be $L$ $(n_q, k_q)$-sources.
    Let $\Ext_q, \Ext_w$ be strong seeded extractors
    that extract $s$ bits from sources with min-entropy $k$ with error $\eps$ and seed length $s$. Let $S_1=\Slice(Q_1,d)$ for some appropriate length $d$, $R_1=\Ext_w(W,S_1),S_2 = \Ext_q(Q_2,R_1),\cdots,R_{L-1}=\Ext_w(W,S_{\ell-1}),S_L = \Ext_q(Q_L,R_{L-1})$, then $L$-alternating extraction$(Q_1,\cdots,Q_L,W)=S_L$.
\end{definition}

\begin{lemma}[Look-ahead extractor~\cite{CGL:stoc:16}]\label{lemma:look-ahead}
    Let $W$ be an $(n_w,k_w)$-source and $W'$ be a random variable on $\bin^{n_w}$ that is arbitrarily correlated with $W$. Let $Y = (Q, S_1)$ such that $Q$ is a $(n_q, k_q)$-source, $S_1$ is a uniform string on $s$ bits, and $Y'=(Q',S_1')$ be a random variable arbitrarily correlated with $Y$, where $Q'$ and $S_1'$
    are random variables on $n_q$ bits and $s$ bits respectively. Let $\Ext_q, \Ext_w$ be strong seeded extractors
    that extract $s$ bits from sources with min-entropy $k$ with error $\eps$ and seed length $s$. Suppose $(Y,Y')$ is independent of $(W, W')$, and $k_w, k_q \ge k + 2\ell s+ 2 \log(1/\eps)$. Let $\laExt$ be the $\ell$ round look-ahead extractor using $\Ext_q, \Ext_w$, and $(R_1, \cdots, R_\ell) = \laExt_\ell(W,Y)$, $(R_1',\cdots, R'_\ell) = \laExt_\ell(W',Y')$. Then for any $0 \le j \le \ell-1$, we have
    \begin{align*}
        R_{j+1} \approx_{O(\ell \eps)} U_s \mid (Y,Y',R_0,R_0',\cdots,R_j,R_j').
    \end{align*}
\end{lemma}

The following lemma captures an essential argument for the $\flip$ and $\nipm$ constructions, which are components of correlation breakers.

\begin{lemma}[Independence-merging lemma~\cite{ChattopadhyayGL:focs:2021}]
\label{lemma:ind-merging}
    Let $\Ext:\bin^n \times \bin^d \to \bin^m$ be any $(k,\eps)$-strong seeded extractor, $X,X^{[t]} \in \bin^n$, $Y,Y^{[t]}\in\bin^d$ such that $X,X^{[t]}$ are independent with $Y,Y^{[t]}\in\bin^d$,   $W=\Ext(X,Y)$ and $W^j=\Ext(X^j,Y^j)$ for every $j\in[t]$.
    Suppose there exists $S,T \subseteq [t]$ such that 
    \begin{itemize}
        \item $(Y,Y^S) \approx_{\delta} (U_d,Y^S)$;
        \item $\avgH(X \mid X^T,Z) \ge k+tm+\log(1/\eps)$.
    \end{itemize}
    Then 
    \begin{align*}
        W\approx_{2\eps+\delta} U_m \mid ( W^{S\cup T}, Y, Y^{[t]}).
    \end{align*}
\end{lemma}

\subsection{\texorpdfstring{$\eps$}{Lg}-Biased Space and XOR Lemmas}

The tools in this subsection are utilized in~\cite{Li:CCC:2011} for their affine disperser and extractor constructions. We also adopt these techniques in our constructions of directional affine dispersers and extractors.

\begin{definition}[$\eps$-biased space]
    A random variable $Z$ over $\bin$ is $\eps$-biased if $|\Pr[Z=0]-$ $\Pr[Z=1]|\le \eps$. A sequence of $0-1$ random variables $Z_1,\cdots, Z_m$ is $\eps$-biased for linear tests if for any nonempty set $S\subset [m]$, the random variable $Z_S = \bigoplus_{i\in S} Z_i$ is $\eps$-biased.
\end{definition}

\begin{lemma}[\cite{Vazirani86}]\label{lemma:eps biased}
    Let $Z_1,\cdots, Z_m$ be $0-1$ random variables that are $\eps$-biased for linear tests. Then the distribution of $(Z_1, \cdots, Z_m)$ is $\eps\cdot 2^{m/2}$-close to uniform.
\end{lemma}

\begin{definition}
    For two functions $f, p : \bin^n \to \bin$, their correlation over the uniform distribution is defined as
    \[  
    \Cor(f, p) = \Big|\Pr_x [f(x) = p(x)] - \Pr_x[f(x) \neq p(x)]\Big|,
    \]
where the probability is over the uniform distribution. For a class $C$ of functions, we denote by $\Cor(f, C)$ the maximum of $\Cor(f, p)$ over all functions $p \in C$ whose domain is $\bin^n$.
\end{definition}

\begin{theorem}[XOR lemma for polynomials over $\F_2$~\cite{v004a007, BhattacharyyaKSSZ10}]\label{thm:xor lemma}
Let $P_d$ stand for the class of all polynomials of degree at most $d$ over $\F_2$. Let $f : \bin^n \to \bin$ be a function such that $\Cor(f, P_d) \le 1-2^{-d}$ and $f^{\oplus m}$ be the XOR of the value of $f$ on $m$ independent inputs. Then
$$\Cor(f^{\oplus m}, P_d) \le \mathsf{exp}(-\Omega(m/(4^d \cdot d))).$$
\end{theorem}

%% file: Lcond.tex
\section{Linear Somewhere Condenser for Affine Sources}\label{sec:cond-affine}
In this section we provide an explicit construction of a linear somewhere condenser for affine sources, or more conveniently, an affine somewhere condenser where each output is a linear function of the input. We begin with the definition.

\begin{definition}
For any $0 < \delta < \gamma < 1$, a function $\scond: \F_2^n \to (\F_2^m)^{\ell}$ is a $(\delta, \gamma)$ affine somewhere condenser, if it satisfies the following property: for any affine source $X$ over $\F_2^n$ with entropy $\delta n$, let $(Y_1, \cdots, Y_{\ell}) =\scond(X) \in (\F_2^m)^{\ell}$, then there exists at least one $i \in [\ell]$ such that $Y_i$ is an affine source over $\F_2^m$ with entropy at least $\gamma m$.
\end{definition}

We will prove the following theorem.

\begin{theorem}\label{thm:Lcondmain}
There exists a constant $\beta > 0$ such that for any $0< \delta \leq 1/2$, there is an explicit $(\delta, 1/2+\beta)$ affine somewhere condenser $\scond: \F_2^n \to (\F_2^m)^t$, where $t=\poly(1/\delta)$ and $m=n/\poly(1/\delta)$. Moreover, $\scond$ is a linear function. 
\end{theorem}

To prove the theorem we will use the following object known as a \emph{dimension expander}.

\begin{definition}[Dimension expander~\cite{BISW04, DvirS11}]
  Let $\F$ be a field and let $T_1, \cdots , T_d : \F^n \to \F^n$ be linear mappings. The set $T = \{T_i\}_{i=1}^{d}$ is an $\alpha$-dimension expander with degree $d$, if for every subspace $V \subset \F^n$ of dimension at most $n/2$ we have

\[\dim \left (\sum_{i=1}^{d} T_i(V) \right ) \geq (1+\alpha) \dim(V).\]
We say that $T$ is explicit if there exists a $\poly(n)$-time algorithm that, on input $n$, outputs $T$.
\end{definition}

\begin{theorem}[\cite{Bourgain09, BourgainY13}]\label{thm:dimexpand}
There exist absolute constants $d \in \N$ and $0< \alpha < 1$ such that over any field $\F$, there exists an explicit family of $\alpha$-dimension expanders with degree $d$.
\end{theorem}

Given the above theorem we first provide a basic affine condenser:

\begin{algorithm}[H]
    \caption{$\bcond(x)$}
    \label{alg:bcond}
    \begin{algorithmic}
        \medskip
        \State \textbf{Input:} $x \in \F_2^n$ --- an $n$ bit string.
        \State \textbf{Output:} $z \in (\F_2^m)^{2d+2}$ --- an array of $2d+2$ bit strings with length $m$, where $m = n/2$ and $d$ is the constant in Theorem~\ref{thm:dimexpand}.
        \\\hrulefill 
        \State \textbf{Sub-Routines and Parameters: } \\
        Let $T = \{T_i\}_{i=1}^{d}$ be the $\alpha$-dimension expander given by Theorem~\ref{thm:dimexpand}. \\
        \\\hrulefill \\
        
        Divide $x$ into $2$ blocks $x = x_1\circ x_2$ where each block has $n/2$ bits. \\
        Let $z=z_1 \circ z_2 \circ \cdots \circ z_{2d+2}$, where $z_1=x_1$, $z_2=x_2$, and $z_{2i+1}=x_1+T_i(x_2)$, $z_{2i+2}=x_2+T_i(x_1)$, for any $i \in [d]$.

    \end{algorithmic}
\end{algorithm}

We will prove the following lemma.

\begin{lemma}\label{lem:bcond}
For any $0 < \delta \leq 1/2$, $\bcond$ is a $(\delta, (1+\frac{\alpha}{4d})\delta)$ affine somewhere condenser, where $\alpha, d$ are the constants in Theorem~\ref{thm:dimexpand}.
\end{lemma}

\begin{proof}
Let $X$ be any affine source over $\F_2^n$ with entropy $k=\delta n$. Without loss of generality, assume the support of $X$ is a linear subspace $V$ (if not, we can do the analysis for the corresponding linear subspace, and then add the affine shift, since we are always dealing with linear functions here). We start by giving a set of $k$ base vectors for $V$. For this, consider the linear subspace $W \subseteq V$ s.t. the first $n/2$ bits of $W$ are $0$. Assume $\dim(W)=r$ and let $b_1, \cdots, b_r$ be a basis for $W$. Next, we extend these vectors to $b_1, \cdots, b_r, c_1, \cdots, c_s$ which form a complete basis for $V$, such that $s+r=k$. 

Note that the vectors formed by the first $n/2$ bits of $\{c_i\}_{i=1}^{s}$ are also linearly independent, otherwise some linear combination of them will be in $W$. Let $\{\bar{c}_i\}_{i=1}^{s}$ be the first $n/2$ bits of $\{c_i\}_{i=1}^{s}$, and $\{\tilde{c}_i\}_{i=1}^{s}$ be the second $n/2$ bits of $\{c_i\}_{i=1}^{s}$. Similarly, let $\{\tilde{b}_i\}_{i=1}^{r}$ be the second $n/2$ bits of $\{b_i\}_{i=1}^{r}$ (recall the first $n/2$ bits are 0). 

Now, let ${\cal Q} \subseteq [s]$ be such that $(\{\tilde{b}_i\}_{i=1}^{r}, \{\tilde{c}_i\}_{i \in \cal Q})$ form a basis of the supporting linear subspace of $X_2$. Let $C=\spn(\{\tilde{c}_i\}_{i \in \cal Q})$. The source $X$ is sampled by picking a uniform random vector $Y=(Y_1, \cdots, Y_k) \in \F_2^k$ and computing 

\[\sum_{i=1}^{s} Y_i c_i +\sum_{j=1}^{r} Y_{s+j} b_j=\sum_{i=1}^{s} Y_i (\bar{c}_i, \tilde{c}_i) +\sum_{j=1}^{r} Y_{s+j} (0, \tilde{b}_j).\]
Thus the first $n/2$ bits are given by $\sum_{i=1}^{s} Y_i \bar{c}_i$, while the second $n/2$ bits are given by

\[\sum_{i=1}^{s} Y_i \tilde{c}_i+ \sum_{j=1}^{r} Y_{s+j} \tilde{b}_j=\sum_{i \in \cal Q} Y_i \tilde{c}_i+ \sum_{i \in [s] \setminus \cal Q} Y_i \tilde{c}_i+\sum_{j=1}^{r} Y_{s+j} \tilde{b}_j.\]

Note that for any $i \in [s] \setminus \cal Q$, $\tilde{c}_i$ can be expressed as a linear combination of $(\{\tilde{b}_i\}_{i=1}^{r}, \{\tilde{c}_i\}_{i \in \cal Q})$. Let $\overline{Y}=(\{Y_i\}_{i \in [s] \setminus \cal Q})$, then the above can be written as

\[\sum_{i \in \cal Q} (Y_i+L_i(\overline{Y})) \tilde{c}_i+ \sum_{j=1}^{r} (Y_{s+j}+L_j(\overline{Y})) \tilde{b}_j,\]

where each $L_i$ or $L_j$ is a linear function from $\F_2^{s-|\cal Q|}$ to $\F_2$.

It is easy to see that the $k$ random bits $(\{Y_i\}_{i=1}^{s}, \{Y_{s+j}+L_j(\overline{Y})\}_{j=1}^{r})$ are independent and uniform (in particular, any non-trivial parity of these bits is a uniform random bit). Similarly, the random bits $(\{Y_i+L_i(\overline{Y})\}_{i \in \cal Q}, \{Y_{s+j}+L_j(\overline{Y})\}_{j=1}^{r})$ are also independent and uniform. Let $A=\spn(\{\bar{c}_i\}_{i=1}^{s})$, $B=\spn(\{\tilde{b}_i\}_{i=1}^{r})$, and $C=\spn(\{\tilde{c}_i\}_{i \in \cal Q)}$. So $\dim(A)=s$, $\dim(B)=r$, and let $\dim(C)=|{\cal Q}|=t$. By the above calculation, we have $H(X_1)=\dim(A)=s$, $H(X_2)=\dim(B)+\dim(C)=r+t$. Furthermore, let $X_3=\sum_{j=1}^{r} (Y_{s+j}+L_j(\overline{Y})) \tilde{b}_j$ and $X_4=\sum_{i \in \cal Q} (Y_i+L_i(\overline{Y})) \tilde{c}_i$, then $X_3$ is the uniform distribution over $B$ and $X_4$ is the uniform distribution over $C$. Thus $H(X_3)=r$ and $H(X_4)=t$. We know $X_1=\sum_{i=1}^{s} Y_i \bar{c}_i$. Thus $X_1$ and $X_3$ are independent, while $X_4$ is a deterministic function of $X_1$ (hence also independent of $X_3$). Note that $X_2=X_3+X_4$, and $X=(X_1, X_2)=(X_1, X_3+X_4)$.

Note that $s+r=k=\delta n$. If $s \geq (\frac{1}{2}+\frac{\alpha}{8d})k$, then $H(X_1) = s \geq (\frac{1}{2}+\frac{\alpha}{8d})k=(1+\frac{\alpha}{4d})\delta (n/2)$. Similarly, if $r \geq (\frac{1}{2}+\frac{\alpha}{8d})k$, then $H(X_2) = r+t \geq r \geq (1+\frac{\alpha}{4d})\delta (n/2)$. In either case, we are done. Otherwise, we must have $s < (\frac{1}{2}+\frac{\alpha}{8d})k$ and $r < (\frac{1}{2}+\frac{\alpha}{8d})k$, which in turn implies that  $s > (\frac{1}{2}-\frac{\alpha}{8d})k$ and $r > (\frac{1}{2}-\frac{\alpha}{8d})k$. Now if $t \geq \frac{\alpha}{4d} k$, then $H(X_2) = r+t > (\frac{1}{2}+\frac{\alpha}{8d})k=(1+\frac{\alpha}{4d})\delta (n/2)$, and again we are done.

The only case left is when $(\frac{1}{2}-\frac{\alpha}{8d})k < s, r < (\frac{1}{2}+\frac{\alpha}{8d})k$ and $t < \frac{\alpha}{4d} k$. Since $s+r =k$, one of them must be at most $k/2=\delta n/2$. We have two cases.

\begin{description}
\item [Case 1.] $(\frac{1}{2}-\frac{\alpha}{8d})k < s \leq k/2$. In this case, $\dim(A)=s \leq \delta(n/2) \leq (1/2) \cdot (n/2)$. Consider the $d$ linear mappings $\{T_i\}_{i=1}^{d}$ given by the dimension expander of Theorem~\ref{thm:dimexpand}. Note that $(1+\alpha)s > (1+\alpha)(\frac{1}{2}-\frac{\alpha}{8d})k > (\frac{1}{2}+\frac{\alpha}{8d})k > r$. We have the following claim.

\begin{claim}
There exists an $i \in [d]$ such that $\dim(T_i(A)+B) \geq r+\frac{(1+\alpha)s-r}{d}$. 
\end{claim}

To see this, suppose for the sake of contradiction that for all $i \in [d]$, we have $\dim(T_i(A)+B) < r+\frac{(1+\alpha)s-r}{d}$. Then 
\[\dim\left( \sum_{i=1}^{d} T_i(A) \right ) < r+d \cdot \frac{(1+\alpha)s-r}{d} =(1+\alpha)s,\]
since any vector in $\sum_{i=1}^{d} T_i(A)$ can be expressed by a linear combination of the $r$ basis vectors in $B$, and another $ < d \cdot \frac{(1+\alpha)s-r}{d}$ vectors, where each $T_i(A)$ contributes $< \frac{(1+\alpha)s-r}{d}$ vectors.

Now for this particular $i \in [d]$, since $X_1$ and $X_3$ are independent, we must have 
\[H(T_i(X_1)+X_3) \geq r+\frac{(1+\alpha)s-r}{d} =\frac{1+\alpha}{d}k +\frac{d-2-\alpha}{d}r \geq \left (\frac{1}{2}+\frac{\alpha}{2d} \right )k,\]
as long as $d \geq 3$.

Note that $T_i(X_1)+X_2=T_i(X_1)+X_3+X_4$, and $X_4$ is a deterministic function of $X_1$. Since $H(X_4)=t < \frac{\alpha}{4d} k$, we can fix $X_4$ and conditioned on any such fixing, 

\[H(T_i(X_1)+X_2) \geq \left (\frac{1}{2}+\frac{\alpha}{2d} \right )k-t > \left (\frac{1}{2}+\frac{\alpha}{2d}\right )k-\frac{\alpha}{4d} k=\left (\frac{1}{2}+\frac{\alpha}{4d} \right )k.\]

Therefore, in the end we still have $H(T_i(X_1)+X_2) > (\frac{1}{2}+\frac{\alpha}{4d} )k=(1+\frac{\alpha}{2d})\delta (n/2)$.

\item [Case 2.] $(\frac{1}{2}-\frac{\alpha}{8d})k < r \leq k/2$. The proof of this case is similar, with a slight modification. Specifically, we have $\dim(B)=r \leq \delta(n/2) \leq (1/2) \cdot (n/2)$. Consider the $d$ linear mappings $\{T_i\}_{i=1}^{d}$ given by the dimension expander of Theorem~\ref{thm:dimexpand}. By exactly the same argument as before, we have the following claim.

\begin{claim}
There exists an $i \in [d]$ such that $\dim(A+T_i(B)) \geq s+\frac{(1+\alpha)r-s}{d}$. 
\end{claim}

Now again, since $X_1$ and $X_3$ are independent, we must have 
\[H(X_1+T_i(X_3)) \geq s+\frac{(1+\alpha)r-s}{d} =\frac{1+\alpha}{d}k +\frac{d-2-\alpha}{d}s \geq \left (\frac{1}{2}+\frac{\alpha}{2d} \right )k,\]
as long as $d \geq 3$.

Note that $X_1+T_i(X_2)=X_1+T_i(X_3)+T_i(X_4)$, and $X_4$ is a deterministic function of $X_1$. Since $H(X_4)=t < \frac{\alpha}{4d} k$, we can fix $X_4$ and conditioned on any such fixing, 

\[H(X_1+T_i(X_2)) \geq \left (\frac{1}{2}+\frac{\alpha}{2d} \right )k-t > \left (\frac{1}{2}+\frac{\alpha}{2d}\right )k-\frac{\alpha}{4d} k=\left (\frac{1}{2}+\frac{\alpha}{4d} \right )k.\]
Therefore, in the end we still have $H(X_1+T_i(X_2)) > (\frac{1}{2}+\frac{\alpha}{4d} )k=(1+\frac{\alpha}{2d})\delta (n/2)$.
\end{description}
\end{proof}

We can now give our main condenser, which involves repeated use of the basic condenser.

\begin{algorithm}[H]
    \caption{$\scond(x)$}
    \label{alg:scond}
    \begin{algorithmic}
        \medskip
        \State \textbf{Input:} $x \in \F_2^n$ --- an $n$ bit string; $0< \delta \leq 1/2$, a given parameter.
        \State \textbf{Output:} $z \in (\F_2^m)^{\ell}$ --- a matrix of $\ell$ bit strings with length $m$, where $m = n/\poly(1/\delta)$ and $\ell=\poly(1/\delta)$.
        \\\hrulefill 
        \State \textbf{Sub-Routines and Parameters: } \\
        Let $\bcond$ be the basic condenser given by Algorithm~\ref{alg:bcond}. \\
        \\\hrulefill \\
        
        Set $x^0=x$ and let $i=0$. Initially $x^i$ has only $n_0 = 1$ row. 
         \begin{enumerate}
        \item Repeat the following step for some $h=O(\log(1/\delta))$ steps:
           For each $j$ and the $j$'th row $x^i_j$ in $x^i$, apply $\bcond(x^i_j)$ to get $2d+2$ rows. Concatenate them to get $x^{i+1}$ with $n_{i+1}=n_i \cdot (2d+2)$ rows. Set $i \leftarrow i+1$. 
        \item Let $z=x^h$.
        \end{enumerate}

    \end{algorithmic}
\end{algorithm}

We can now prove our main theorem.

\begin{proof}[Proof of Theorem~\ref{thm:Lcondmain}.] We show that Algorithm~\ref{alg:scond} gives such an affine somewhere condenser. By Lemma~\ref{lem:bcond}, for any affine source $X$ with $H(x)=\delta n$ for some $0< \delta \leq 1/2$, after some $h'=O(\log(1/\delta))$ steps at least one of the rows $x^{h'}$ has entropy at least $n'/2$. Without loss of generality assume this row has entropy exactly $n'/2$ (otherwise we can first fix some basis vectors in the support linear subspace and thus reduce the row to a convex combination of affine sources with entropy exactly $n'/2$). Then after another step of applying $\bcond$, one of the output rows will have entropy rate at least $1/2(1+\frac{\alpha}{4d})=\frac{1}{2}+\frac{\alpha}{8d}$. 

It's easy to see that $\scond$ is a linear function, and thus each row in the final output is an affine source. Furthermore, since we divide each row into $2$ equal blocks in every step and obtain $2d+2$ new rows from them, the final length of each row is $m=n/\poly(1/\delta)$ and we have altogether $\ell=\poly(1/\delta)$ rows.
\end{proof}

\section{Linear Somewhere Condenser for General Weak Sources}\label{sec:cond-general}
We next show that our linear somewhere condenser also works for general weak random sources.

\subsection{Some Useful Results}

\begin{definition}
    The collision probability of a distribution $\D$ is defined as $\cp({\D})=\Pr_{x, y \leftarrow_R {\D}}[x=y]$.
\end{definition}

\begin{definition}
    We say a distribution $\cal X$ is a convex combination of distributions ${\cal X}_1, \cdots , {\cal X}_m$ if there exist numbers $p_1, \cdots , p_m \in [0, 1]$ such that $\sum_i p_i=1$ and the random variable $\cal X$ is equal to $\sum_i p_i {\cal X}_i$.
\end{definition}


\begin{lemma}[\cite{BISW04}]\label{lem:collision2}
    Let $\X$ be a distribution such that $\cp(\X) \leq \frac{1}{KL}$. Then $\X$ is of statistical distance $\frac{1}{\sqrt{L}}$ from having min-entropy at least $\log K$.
\end{lemma}

We need the following results from additive combinatorics.

\begin{lemma}[Pl\H{u}nnecke-Ruzsa~\cite{tao_vu_2006}]\label{lem:add1} 
Let $A, B$ be finite subsets in an additive group $G$. Then 
\[|A+A| \leq \frac{|A+B|^4}{ |A||B|^2}.\]
\end{lemma}

\begin{lemma}[Balog-Szemeredi-Gowers~\cite{Balog1994AST, Gowers1998ANP}]\label{lem:add2} 
Let $A, B$ be finite subsets of an additive group $G$ and let $|A|^{1-\rho_1} \leq |B| \leq |A|^{1+\rho_1}$ . If $\cp(A + B) \geq |A|^{-(1+\rho_2-\rho_1)}$, then there exist subsets $A' \subseteq A, B' \subseteq B$ such that $|A'| \geq |A|^{1-10 \rho_2} , |B'| \geq |B|^{1-10 \rho_2}$ , and $|A' + B'| \leq |A|^{1+\rho_1+10\rho_2}$ .
\end{lemma}

\begin{theorem}[Polynomial Freiman-Ruzsa Theorem in $\F^n_2$~\cite{gowers2023conjecture}] \label{conj:poly} Let $A \subset \F^n_2$ be a set such that $|A + A| \leq M|A|$. Then there exists a subset $A' \subset A$ of size $|A'| \geq M^{-c}|A|$ such that $|\Span(A')| \leq M^c|A|$, where $c \geq 0$ is an absolute constant.
\end{theorem}


\subsection{The Construction}
We generalize our affine somewhere condenser as follows.

\begin{algorithm}[H]
    \caption{$\bgcond(x)$}
    \label{alg:bgcond}
    \begin{algorithmic}
        \medskip
        \State \textbf{Input:} $x \in \zo^n$ --- an $n$ bit string.
        \State \textbf{Output:} $z \in (\zo^m)^{2d+3}$ --- an array of $2d+3$ bit strings with length $m$, where $m = n/2$ and $d$ is the constant in Theorem~\ref{thm:dimexpand}.
        \\\hrulefill 
        \State \textbf{Sub-Routines and Parameters: } \\
        Let $T = \{T_i\}_{i=1}^{d}$ be the $\alpha$-dimension expander given by Theorem~\ref{thm:dimexpand}. \\
        \\\hrulefill \\
        
        Divide $x$ into $2$ blocks $x = x_1\circ x_2$ where each block has $n/2$ bits. \\
        Let $z=z_1 \circ z_2 \circ \cdots \circ z_{2d+3}$, where $z_1=x_1$, $z_2=x_2$, and $z_{2i+1}=x_1+T_i(x_2)$, $z_{2i+2}=x_2+T_i(x_1)$, for any $i \in [d]$. Finally let $z_{2d+3}=x_1+x_2$. Here all additions are viewing the inputs as elements in the field $\F^m_2$.

    \end{algorithmic}
\end{algorithm}

We have the following lemma.

\begin{lemma}\label{lem:bgcond}
For any $0 < \delta \leq 1/2$, $\bgcond$ is a rate $(\delta \to (1+\Omega(\frac{\alpha}{d}))\delta, 2^{-\Omega(\delta n)})$ somewhere condenser, where $\alpha, d$ are the constants in Theorem~\ref{thm:dimexpand}.
\end{lemma}

To prove the lemma we first prove the following lemmas.

\BL \label{lem:entropy1}
For any constant $c>0$ there exists a constant $\eps=\Omega(\frac{\alpha}{d})$ such that the following holds. Let $A, B$ be finite subsets of $\F^n_2$. For any $K \leq 2^{n/4}$, assume $K^{1-c\eps} \leq |A|, |B| \leq K^{1+c\eps}$. If $\cp(A + B) \geq K^{-(1+2\eps)}$, then there exist subsets $\tilde{A} \subseteq A, \tilde{B} \subseteq B$ such that $|\tilde{A}| \geq K^{1-O(\eps)} , |\tilde{B}| \geq K^{1- O(\eps)}$, $|\Span(\tilde{A})| \leq K^{1+O(\eps)} , |\Span(\tilde{B})| \leq K^{1+O(\eps)}$, and at least one row in the output of $\bcond({\A} \circ {\B})$ has min-entropy $(1+\Omega(\eps)) \log K $, where $\A, \B$ are the uniform and independent distributions over $\tilde{A}, \tilde{B}$ respectively.
\EL

\begin{proof}
    If $\cp(A + B) \geq K^{-(1+2\eps)}$, by Lemma~\ref{lem:add2} there exist subsets $A' \subseteq A, B' \subseteq B$ such that $|A'| \geq |A|^{1-O(\eps)} , |B'| \geq |B|^{1-O(\eps)}$, and $|A' + B'| \leq |A|^{1+O(\eps)}=K^{1+O(\eps)}$. Then, by Lemma~\ref{lem:add1}, we have $|A'+A'| \leq \frac{|A'+B'|^4}{ |A'||B'|^2} \leq K^{1+O(\eps)}$. Similarly we also have $|B'+B'| \leq K^{1+O(\eps)}$.

    Next, by Theorem~\ref{conj:poly}, there exists a subset $\tilde{A} \subset A'$ of size $|\tilde{A}| \geq K^{-O(\eps)}|A'| = K^{1-O(\eps)}$ such that $|\Span(\tilde{A})| \leq K^{O(\eps)}|A'| = K^{1+O(\eps)}$. Similarly there also exists such a subset $\tilde{B} \subset B'$ with the same property. Now let $\A', \B'$ be the uniform and independent distributions over $\Span(\tilde{A}), \Span(\tilde{B})$ respectively. Note that $\A', \B'$ are both affine sources and hence $\A' \circ \B'$ is also an affine source with entropy $\geq \log |\tilde{A}|+\log |\tilde{B}|=(1-O(\eps))2 \log K$. Thus, without loss of generality we can view it as an affine source with entropy exactly $(1-O(\eps))2 \log K < n/2$. Now by Lemma~\ref{lem:bcond}, at least one row in the output of $\bcond({\A'} \circ {\B'})$ has entropy $(1+\frac{\alpha}{4d})(1-O(\eps)) \log K$. Note that $|\tilde{A}||\tilde{B}| \geq (|\Span(\tilde{A})||\Span(\tilde{B})|)^{(1-O(\eps))}$. Thus the same row in the output of $\bcond({\A} \circ {\B})$ has min-entropy at least $(1+\frac{\alpha}{4d})(1-O(\eps)) \log K-O(\eps)\log K=(1+\Omega(\eps)) \log K$, as long as $\eps=\gamma \frac{\alpha}{4d}$ for a sufficiently small constant $\gamma>0$.
\end{proof}

\BL \label{lem:entropy2}
For any constant $c>0$ there exists a constant $\eps=\Omega(\frac{\alpha}{d})$ such that the following holds. Let $A, B$ be finite subsets of $\F^n_2$. For any $K \leq 2^{n/4}$, assume $K^{1-c\eps} \leq |A|, |B| \leq K^{1+c\eps}$ and $|\Span(B)| \leq K^{1+O(\eps)}$. If $\cp(A + B) \geq K^{-(1+2\eps)}$, then there exists a subset $\tilde{A} \subseteq A$ such that $|\tilde{A}| \geq K^{1-O(\eps)}$, $|\Span(\tilde{A})| \leq K^{1+O(\eps)}$, and at least one row in the output of $\bcond({\A} \circ {\B})$ has min-entropy $(1+\Omega(\eps)) \log K $, where $\A, \B$ are the uniform and independent distributions over $\tilde{A}, B$ respectively.
\EL

\begin{proof}
    The proof is exactly the same as the previous lemma, except in the second paragraph we can replace the set $\Tilde{B}$ with $B$ directly. 
\end{proof}

We can now prove the following lemma.

\BL \label{lem:entropy3}
There exists a constant $\eps=\Omega(\frac{\alpha}{d})$ such that the following holds. Let $A, B$ be finite subsets of $\F^n_2$. For any $K \leq 2^{n/4}$, assume $K^{1-\eps} \leq |A|, |B| \leq K^{1+\eps}$. Let $X, Y$ be the uniform and independent distributions over $A, B$ respectively. Then $\bgcond(X \circ Y)$ is $K^{-\Omega(\eps)}$-close to a somewhere-$(1+\Omega(\eps)) \log K $ source. In particular, $X \circ Y$ can be divided into disjoint subsources, such that for each subsource, either (1) the probability mass is at most $2K^{-\eps}$, or (2) the probability mass is at least $K^{-O(\eps)}$, and the output of $\bgcond$ on the subsource is $K^{-\eps}$-close to being an elementary somewhere $(1+\Omega(\eps)) \log K $ source.
\EL

\begin{proof}
We repeatedly apply Lemma~\ref{lem:entropy1} and Lemma~\ref{lem:entropy2}, and dividing $A \times B$ into disjoint subsets as follows. First note that if $\cp(A+B) \leq K^{-(1+2\eps)}$, then by Lemma~\ref{lem:collision2}, $X+Y$ is $K^{-\eps/2}$-close to having min-entropy $(1+\eps)\log K$. Otherwise, by Lemma~\ref{lem:entropy1}, there exist subsets $\tilde{A} \subseteq A, \tilde{B} \subseteq B$ such that $|\tilde{A}| \geq K^{1-O(\eps)} , |\tilde{B}| \geq K^{1- O(\eps)}$, $|\Span(\tilde{A})| \leq K^{1+O(\eps)} , |\Span(\tilde{B})| \leq K^{1+O(\eps)}$, and at least one row in the output of $\bcond({\A} \circ {\B})$ has min-entropy $(1+\Omega(\eps)) \log K $, where $\A, \B$ are the uniform and independent distributions over $\tilde{A}, \tilde{B}$ respectively.

Now consider the set $A^1 = A \setminus \tilde{A}$ and $B^1 = B \setminus \tilde{B}$. If $|A^1| \leq K^{1-2\eps}$ and $|B^1| \leq K^{1-2\eps}$, then the total probability mass in $X \circ Y$ corresponding to elements in $(A \times B) \setminus (\tilde{A} \times \tilde{B})$ is at most $\frac{|A^1||B|+|A||B^1|}{|A||B|} \leq 2K^{-\eps}$, and we are done.

Otherwise, consider the following three sets: $\tilde{A} \times B^1$, $A^1 \times \tilde{B}$, and $A^1 \times B^1$. Note that these are disjoint subsets whose union equals $(A \times B) \setminus (\tilde{A} \times \tilde{B})$. We have several cases. 

\begin{description}
    \item[Case 1.] Only one of $|A^1|$ and $|B^1|$ has size larger than $K^{1-2\eps}$. Without loss of generality assume $|B^1| \leq K^{1-2\eps}$. Note that in this case the total probability mass in $X \circ Y$ corresponding to elements in $(\tilde{A} \times B^1) \cup (A^1 \times B^1)= A \times B^1$ is at most $\frac{|A||B^1|}{|A||B|} \leq K^{-\eps}$.

    For $A \times \tilde{B}$, we repeatedly apply Lemma~\ref{lem:entropy2}. Initially let $A^*=A$. As long as $|A^*| \geq K^{1-2\eps}$, if $\cp(A^*+\tilde{B}) \leq K^{-(1+2\eps)}$, then again by Lemma~\ref{lem:collision2}, the output of the sum of the random variables corresponding to the uniform and independent distributions over $A^*$ and $\tilde{B}$ will be $K^{-\eps/2}$-close to having min-entropy $(1+\eps)\log K$, and we stop here. Otherwise we use Lemma~\ref{lem:entropy2} to find a subset $\tilde{A} \subseteq A^*$ such that $|\tilde{A}| \geq K^{1-O(\eps)}$, $|\Span(\tilde{A})| \leq K^{1+O(\eps)}$, and at least one row in the output of $\bcond({\A} \circ {\B})$ has min-entropy $(1+\Omega(\eps)) \log K $,  where $\A, \B$ are the uniform and independent distributions over $\tilde{A}, \tilde{B}$ respectively. We then remove $\tilde{A}$ from $A^*$ and repeat. The process ends when $|A^*| < K^{1-2\eps}$.

    Thus, altogether, we have divided $A \times B$ into disjoint subsets, or equivalently, $X \circ Y$ into disjoint subsources, such that for each subsource, either the probability mass is at most $K^{-\eps}$, or the output of $\bgcond$ on the subsource is $K^{-\eps}$-close to being a somewhere $(1+\Omega(\eps)) \log K $ source.

    \item[Case 2.] $|A^1| > K^{1-2\eps}$ and $|B^1| > K^{1-2\eps}$. We first apply the argument in Case 1 to $\tilde{A} \times B^1$ and $A^1 \times \tilde{B}$. Then we consider $A^1 \times B^1$. This is the same situation as when we start. Namely, if $\cp(A^1+B^1) \leq K^{-(1+2\eps)}$, then by Lemma~\ref{lem:collision2} we are done. Otherwise by Lemma~\ref{lem:entropy1}, there exist subsets $\tilde{A^1} \subseteq A^1, \tilde{B^1} \subseteq B^1$ such that $|\tilde{A^1}| \geq K^{1-O(\eps)} , |\tilde{B^1}| \geq K^{1- O(\eps)}$, $|\Span(\tilde{A^1})| \leq K^{1+O(\eps)} , |\Span(\tilde{B})| \leq K^{1+O(\eps)}$, and at least one row in the output of $\bcond({\A^1} \circ {\B^1})$ has min-entropy $(1+\Omega(\eps)) \log K $, where $\A^1, \B^1$ are the uniform and independent distributions over $\tilde{A^1}, \tilde{B^1}$ respectively. We can therefore continue the analysis as before. 
\end{description}

Combining the two cases, eventually we have divided $X \circ Y$ into disjoint subsources, such that for each subsource, either (1) the probability mass is at most $2K^{-\eps}$, or (2) the output of $\bgcond$ on the subsource is $K^{-\eps}$-close to being a somewhere $(1+\Omega(\eps)) \log K $ source.

Notice that when a subsource satisfies (2), its probability mass is always at least $K^{1-O(\eps)} \cdot K^{1-O(\eps)}/K^{2(1+\eps)}=K^{-O(\eps)}$. 
\end{proof}

We can now prove Lemma~\ref{lem:bgcond}.

\begin{proof}[Proof of Lemma~\ref{lem:bgcond}.]
    Given an $(n, \delta n)$ source $X$ with $0 < \delta \leq 1/2$, and $X=X_1 \circ X_2$, without loss of generality we can assume that $X$ is the uniform distribution over a set $S \subseteq \zo^n$ with $|S|=2^{\delta n}$. We first pick a constant parameter $\lambda>0$ to be chosen later. For $i \in [2]$ define $H_i = \{y \in \zo^m: \Pr[X_i = y] \geq 2^{-(1+\lambda)\delta m} \}$, which corresponds to the heavy elements in $X_i$. Notice that this implies for every $i$, $|H_i| \leq 2^{(1+\lambda)\delta m}$. Let $\tau=2^{-\beta \delta m}$ for some constant $\beta>0$ to be chosen later. We define the following sets.

    \begin{enumerate}
        \item $S' = \{x \in S : \exists i, x_i \notin H_i\}$. 
        \item For any $x \in S'$, define $I(x)$ to be the smallest $i$ such that $x_i \notin H_i$, and $T_i=\{x \in S', I(x)=i\}$. Let $B=\{i \in [2]: |T_i| < 2^{(1-\beta) \delta m} \}$, and define $\Tilde{S}=S' \setminus (\cup_{i \in B} T_i )$. Note that $|\cup_{i \in B} T_i| \leq 2\tau |S|$.
        \item $S'' = \{x \in S : \forall i, x_i \in H_i\}= S \setminus S'$.
    \end{enumerate}
    
    Note that for any $x \in \Tilde{S}$, we have $I(x) \notin B$. Let $\Tilde{X}$ be the uniform distribution over $\Tilde{S}$. For any $i \in [2] \setminus B$, and any $y \in \zo^m$, conditioned on $I(\Tilde{X})=i$, we have $\Pr[\Tilde{X}_i=y] \leq \frac{\Pr[X_i=y]}{2^{-\beta \delta m}} \leq 2^{-(1+\lambda-\beta)\delta m}.$ Thus as long as $\beta \leq \lambda/2$, $\Tilde{X}_i$ has min-entropy at least $(1+\lambda/2)\delta m$. Hence $\Tilde{X}$ is an elementary somewhere-$(1+\lambda/2)\delta m$ source.
    
    We now have two cases.

    \paragraph{Case 1.} $\Pr[X \in S'] \geq 1-\tau$. In this case, notice that $\Tilde{X}$ is $2\tau+\tau=3\tau$-close to $X$, thus we are done. 

\paragraph{Case 2.} $\Pr[X \in S''] \geq \tau$. In this case, notice that $|S''| \geq \tau |S| = 2^{(2-\beta) \delta m}$. Also, $S''$ is a subset of $H_1 \times H_2$, so 

\[|H_1 \times H_2| \geq |S''| \geq 2^{(2-\beta) \delta m}.\]

However, for each $i \in [2]$ we have $|H_i| \leq 2^{(1+\lambda)\delta m}$, and thus for each $i \in [2]$ we also have 

\[|H_i| \geq 2^{(2-\beta) \delta m}/2^{(1+\lambda)\delta m}=2^{(1-\beta-\lambda)\delta m}.\]

We now consider the source $(Y_1, Y_2)$ where each $Y_i$ is the independent uniform distribution over $H_i$. We will apply Lemma~\ref{lem:entropy3} by setting $K=2^{\delta m} \leq 2^{n/4}$, and $\eps \geq 2\lambda$. Notice that for any $i \in [2]$, we have $K^{1-\eps} \leq |H_i| \leq K^{1+\eps}$ since we have chosen $\beta \leq \lambda/2$.

Thus by Lemma~\ref{lem:entropy3}, there exits a constant $c>0$ such that $Y_1 \circ Y_2$ can be divided into disjoint subsources, such that for each subsource, either (1) the probability mass is at most $2K^{-\eps}$, or (2) the probability mass is at least $K^{-c\eps}$, and the output of $\bgcond$ on the subsource is $K^{-\eps}$-close to being an elementary somewhere $(1+\Omega(\eps)) \log K $ source.

For each subsource $Y^j$ in (2), we consider the intersection of its support with $S''$. If the intersection has probability mass at most $K^{-c\eps-4\lambda}$, then we say it is a \emph{bad} intersection, otherwise we say it is a \emph{good} intersection. Notice that for a good intersection, the output of $\bgcond$ on the subsource defined as the uniform distribution over the intersection is $K^{-\eps+4\lambda}$-close to being an elementary somewhere $(1+\Omega(\eps)-4 \lambda) \log K $ source. On the other hand, the total probability mass of the bad intersections is at most $K^{-4\lambda}$. 

Notice that the probability mass of $S''$ in $(Y_1, Y_2)$ is at least $2^{(2-\beta) \delta m}/(2^{(2+2\lambda)\delta m})=2^{-(\beta+2\lambda)\delta m}$. Hence if we define $X''$ as the uniform distribution over $S''$, then $\bgcond(X'')$ is $(2K^{-\eps}+K^{-4\lambda})/(2^{-(\beta+2\lambda)\delta m})+ K^{-\eps+4\lambda} \leq 2^{-\lambda \delta m}$-close to a somewhere $(1+\Omega(\eps)-4 \lambda) \log K =(1+\lambda) \delta m$ source, as long as we take $\lambda=\gamma \eps$ for a sufficiently small constant $\gamma>0$.

Now define $X'$ to be the uniform distribution over $S'' \cup \Tilde{S}$. Then $\bgcond(X')$ is $2^{-\lambda \delta m}$-close to a somewhere $(1+\lambda/2) \delta m$ source. Notice that $X$ is $2\tau$-close to $X'$. Thus $\bgcond(X)$ is $2\tau+2^{-\lambda \delta m}$-close to a somewhere $(1+\lambda/2) \delta m$ source.

Setting $\beta=\lambda/2$, we have that in both cases, $\bgcond(X)$ is $2^{-\Omega(\delta n)}$-close to a somewhere $(1+\Omega(\frac{\alpha}{d}))\delta m$ source.
\end{proof}

Our main condenser now involves repeated uses of the basic condenser.

\begin{algorithm}[H]
    \caption{$\sgcond(x)$}
    \label{alg:sgcond}
    \begin{algorithmic}
        \medskip
        \State \textbf{Input:} $x \in \F_2^n$ --- an $n$ bit string; $0< \delta \leq 1/2$, a given parameter.
        \State \textbf{Output:} $z \in (\F_2^m)^{\ell}$ --- a matrix of $\ell$ bit strings with length $m$, where $m = n/\poly(1/\delta)$ and $\ell=\poly(1/\delta)$.
        \\\hrulefill 
        \State \textbf{Sub-Routines and Parameters: } \\
        Let $\bgcond$ be the basic condenser given by Algorithm~\ref{alg:bgcond}. \\
        \\\hrulefill \\
        
        Set $x^0=x$ and let $i=0$. Initially $x^i$ has only $n_0 = 1$ row. 
         \begin{enumerate}
        \item Repeat the following step for some $h=O(\log(1/\delta))$ steps:
           For each $j$ and the $j$'th row $x^i_j$ in $x^i$, apply $\bgcond(x^i_j)$ to get $2d+3$ rows. Concatenate them to get $x^{i+1}$ with $n_{i+1}=n_i \cdot (2d+3)$ rows. Set $i \leftarrow i+1$. 
        \item Let $z=x^h$.
        \end{enumerate}

    \end{algorithmic}
\end{algorithm}

By a similar argument as in the proof of Theorem~\ref{thm:Lcondmain}, we can prove the following theorem.

\begin{theorem}\label{thm:LGcondmain}
There exists a constant $\beta > 0$ such that for any $0< \delta \leq 1/2$, there is an explicit rate $(\delta \to 1/2+\beta, 2^{-\Omega(m)})$ somewhere condenser $\sgcond: \zo^n \to (\zo^m)^t$, where $t=\poly(1/\delta)$ and $m=n/\poly(1/\delta)$. Moreover, $\sgcond$ is a linear function. 
\end{theorem}

%% file: daext2.tex
\section{Directional Affine Extractor}
\label{sec:daext2}

In this section, we describe our directional affine extractors for linear entropy with exponentially small error. The construction also works for sublinear entropy with a slight loss in the error and output length.

\subsection{Low-Degree Affine Correlation Breaker}
As we introduced in Section~\ref{sec:intro}, keeping the outputs of the directional affine extractor \emph{low-degree} is critical. However, our construction makes use of advice correlation breakers, and all existing correlation breakers have degrees forbiddenly high for our purpose. To handle this, we construct a family of low-degree correlation breakers. We assume that the input random variables to each of the following subroutines are affine. This assumption is valid since in the analysis of Algorithm~\ref{alg:DAExt} where we invoke Theorem~\ref{thm:ldacb} of $\ldACB$, the input random variables are affine.

\paragraph{Substitutes for strong seeded extractors.}
 We will base our construction on a similar framework to the advice correlation breaker in~\cite{ChattopadhyayGL:focs:2021}. To keep the degree low, we substitute the GUV extractors and the condense-then-hash extractors used throughout with the low-degree strong linear seeded extractor from Theorem~\ref{thm:low-deg-lsext}. 

\begin{theorem}[Low-degree strong linear seeded extractors~\cite{Li:CCC:2011}]\label{thm:low-deg-lsext}
    There exists a constant $0<\beta<1$ such that for every $0<\delta<1$ and any $1/\sqrt{n} < \alpha < 1$ there exists a polynomial time computable function $\mathsf{LSExt}:\{0,1\}^n \times \{0,1\}^d \to \{0,1\}^m$ and a constant $0<\beta<1$ such that s.t.
    \begin{itemize}
        \item $d\le\alpha n,m\ge\beta\delta\alpha n$.
        \item For any $(n,\delta n)$-affine source $X$, let $R$ be the uniform distribution on $\{0,1\}^d$ independent of $X$. Then $(\mathsf{LSExt}(X,R),R)$ is $2^{-\Omega(\delta \alpha^2 n)}$-close to uniform.
        \item Each bit of the output is a degree $4$ polynomial of the bits of the two inputs, and for any fixing of $r$ the output is a linear function of $x$.
    \end{itemize}
\end{theorem}

\paragraph{Low-degree look-ahead extractor.} 
The first step is to construct a low-degree look-ahead extractor which is a component of the low-degree advice correlation breaker. The following algorithm is such a construction instantiated with the low-degree strong linear seeded extractor in Theorem~\ref{thm:low-deg-lsext}. Since the low-degree strong linear seeded extractor has shorter output length than the minimal seed length, we cannot directly apply existing lemmas about look-ahead extractors. Instead, we need to tailor a new set of parameters and a new theorem for the low-degree one.
\begin{algorithm}[H]
    \caption{$(k,t,\eps)$-$\laExt(x,y)$}
    \label{alg:laExt}
    \begin{algorithmic}
        \medskip
        \State \textbf{Input:} Bit strings $x$, $y$ of length $n,d$ respectively. Initially, $x$ has entropy $k$.
        \State \textbf{Output:} Bit string $(r_0,r_1)$ of length $2m$.
        \State \textbf{Subroutines and Parameters:} \\
        Let $s=d/(2+2t)$, where $C_0\sqrt{\frac{\log(1/\eps)}{k}}  n\ge s\ge C_1n\sqrt{n}/k$ for some constants $C_0> 0,C_1>1$. \\
        Let $\LSExt_w^1: \bin^{n} \times \bin^s \to \bin^{m_1}$ be the low-degree strong linear seeded extractor from Theorem~\ref{thm:low-deg-lsext} with $\delta_{\ref{thm:low-deg-lsext}}=k/n$, $\alpha_{\ref{thm:low-deg-lsext}} = d/((2t+2)n)$, error $\eps_0=2^{-\Omega(kd^2/((t+1)^2n^2)}$ and output length $m_1=\beta_{\ref{thm:low-deg-lsext}}kd/((2t+2)n)$. \\
        Let $\LSExt_q^1: \bin^{d} \times \bin^{m_1} \to \bin^{m_2}$ be the low-degree strong linear seeded extractor from Theorem~\ref{thm:low-deg-lsext} with $\delta_{\ref{thm:low-deg-lsext}}=1/2$, $\alpha_{\ref{thm:low-deg-lsext}} = \beta_{\ref{thm:low-deg-lsext}}k/((2+2t)n)$, error $\eps_1 = 2^{-\Omega(dk^2/((t+1)^2n^2))}=\eps_0^{\Omega(d/k)}$ and output length $m_2 = \beta^2_{\ref{thm:low-deg-lsext}}kd/(4(t+1)n)$. \\
        Let $\LSExt_w^2: \bin^n \times \bin^{m_2}\to \bin^{m}$ be the low-degree strong linear seeded extractor from Theorem~\ref{thm:low-deg-lsext} with $\delta_{\ref{thm:low-deg-lsext}}=k/(2n)$ and $\alpha_3 = \beta^2_{\ref{thm:low-deg-lsext}}kd/((4+4t)n^2)$, error $\eps_2 = 2^{-\Omega(k^3d^2/((8+8t)^2n^4))}= \eps_0^{\Omega(k^2/n^2)}$ and output length $m=\beta^3_{\ref{thm:low-deg-lsext}}k^2d/((8+8t)n^2)$.
        \\\hrulefill \\
        \begin{enumerate}
            \item Let $s_0 = \Slice(y,s)$
            \item Let $\tilde{r}_0 = \LSExt_w^1(x,s_0)$
            \item Let $s_1 = \LSExt_q^1(y,\tilde{r}_0)$
            \item Let $r_1 = \LSExt_w^2(x, s_1)$
            \item Output $r_0=\Slice(\tilde{r}_0,|r_1|),r_1$
        \end{enumerate}
    \end{algorithmic}
\end{algorithm}

\begin{theorem}[$2$-look-ahead extractor]
    \label{thm:laext}
    
    For every $t\le \sqrt{n}, \;t\in\N$ and $\eps>0$, there exists an explicit function $\laExt:\bin^n\times \bin^d \to \pbra{\bin^m}^2$ which satisfies the following. Let $X,X^{[t]} \in \bin^n$ and $Y,Y^{[t]}\in \bin^d$ be random variables such that $\pbra{X,X^{[t]}}$ is independent of $\pbra{Y,Y^{[t]}}$, $Y=U_d$.
    There exists a large enough constant $C>0$ such that if 
    \begin{align*}
        k & =H(X)\ge C \max\cbra{\pbra{(t+1)^2\log(1/\eps)n^4/d^2}^{1/3},(t+1)\sqrt{n}}; \\
        n & \ge d\ge C(t+1)\max\cbra{n\sqrt{n}/k,\pbra{\log(1/\eps)n^4/k^3}^{1/2}},
    \end{align*}
    then $(R_0,R_1):=\laExt(X,Y)$ and their tamperings $(R_0^{[t]}, R_1^{[t]})$ satisfy
    \begin{align*}
        & (R_0 \approx_{\eps} U_m) \mid (Y,Y^{[t]}); \\
        & (R_1 \approx_{\eps} U_m) \mid (Y,Y^{[t]},R_0,R_0^{[t]}),
    \end{align*}
    where $m=\Omega(k^2d/((1+t)n^2))$. \\
    Moreover, each bit of $r_0'$ is a degree $4$ polynomial of the input bits; each bit of $r_1$ is a degree $40$ polynomial of the input bits.
\end{theorem}
\begin{proof}
    We will show that Algorithm~\ref{alg:laExt} is such a function. First we demonstrate that the choice of the parameters in Algorithm~\ref{alg:laExt} are correct. The first constraint comes from the requirement of the minimal seed length of any $\LSExt$, i.e., we need to guarantee that 
    \begin{align*}
        d/((2t+2)n) & > 1/\sqrt{n}~\tag{seed length requirement of $\LSExt_w^1$}; \\
        \beta_{\ref{thm:low-deg-lsext}}k/((2+2t)n) & > 1/\sqrt{n}~\tag{seed length requirement of $\LSExt_q^1$}; \\
        \beta_{\ref{thm:low-deg-lsext}}^2kd/((4+4t)n^2) & > 1/\sqrt{n}~\tag{seed length requirement of $\LSExt_w^2$},
    \end{align*}
    this puts a lower bound for $d$:
    \begin{align*}
        d > \frac{4}{\beta_{\ref{thm:low-deg-lsext}}^2}\cdot \frac{(t+1)n\sqrt{n}}{k}.
    \end{align*}
    Since $n\ge d$, we have
    \begin{align*}
        k > \frac{4}{\beta_{\ref{thm:low-deg-lsext}}^2}\cdot (t+1)n.
    \end{align*}
    Let $\eps \ge \eps_0+\eps_1+\eps_2$, then there exists a constant $\lambda$ such that $\eps \ge 2^{-\lambda k^3d^2/((t+1)^2n^4)}$. Taking the logarithm on the error and isolating out $k$ and $d$, we have
    \begin{align*}
        & \frac{(t+1)^2\log(1/\eps) n^4}{\lambda} \le k^3 d^2  \iff k\ge C_0\pbra{(t+1)^2\log(1/\eps)n^4/d^2}^{1/3}; \\
        & \frac{(t+1)^2\log(1/\eps) n^4}{\lambda} \le k^3 d^2  \iff d\ge C_0'(t+1)\pbra{\log(1/\eps)n^4/k^3}^{1/2},
    \end{align*}
    for some large enough constants $C_0,C_0'$. Therefore, if $k$ and $d$ satisfy the constraints in Theorem~\ref{thm:laext}, they also works for Algorithm~\ref{alg:laExt}. Next, we prove the extraction properties of the look-ahead extractor.
    \begin{enumerate}
        \item Since $Y=U_d$, $S_0$ is uniform. Since $Y$ is independent of $X$, $X$ is independent of $S_0$. Since $H(X)\ge k$, by the property of strong seeded extractor of $\LSExt_w^1$, $\tilde{R}_0\approx_{\eps_0} U_{m_1} \mid S_0$, which also implies that $R_0\approx_{\eps_0} U_m \mid (Y,Y^{[t]})$ given the independence between $X$ and $(Y,Y^{[t]})$.
        \item Since $S_0,S_0^{[t]}$ are linear functions of $Y$ and $Y^{[t]}$, we have $H(Y\mid S_0,S_0^{[t]}) \ge d - (t+1)\cdot \frac{d}{2+2t} =  d/2$. Since $\tilde{R}_0$ is $\eps_0$ close to uniform, by the property of strong seeded extractor of $\LSExt_q^1$, $S_1\approx_{\eps_0+\eps_1} U_{m_2}$.
        \item Conditioned on the fixings of $(S_0,S_0^{[t]})$, $R_0,R_0^{[t]}$ are linear functions of $X$ and $X^{[t]}$. Therefore, we have $H(X \mid \tilde{R}_0,\tilde{R}_0^{[t]}) \ge k - (t+1)m_1 = k/2$. Since $S_1 \approx_{\eps_0+\eps_1} U_{m_2}$, by the property of strong seeded extractor of $\LSExt_w^2$, $R_1\approx_{\eps_0+\eps_1+\eps_2} U_m\mid (R_0,R_0^{[t]}, Y, Y^{[t]})$.
    \end{enumerate}
    Lastly, the degree of the output follows easily from the degree of the output of $\LSExt$ from Theorem~\ref{thm:low-deg-lsext}. This completes the proof of Theorem~\ref{thm:laext}.
\end{proof}
\paragraph{Low-degree non-malleable independence-preserving merger.}
The second step is to construct a low-degree non-malleable independence-preserving merger. Non-malleable independence-perserving merger was first defined in~\cite{ChattopadhyayL:focs:2016} to merge a somewhere random source while preserving independence among itself and the tampered sources. We start with the definition.
\begin{definition}\label{def:ldnipm}
    An $(t,\ell,\eps)\text{-}\nipm:\bin^n \times \pbra{\bin^m}^\ell  \to \bin^{m_1}$, or $\nipm_\ell$ for short, with error $\eps$ for $\ell \in \N$ is function which satisfies the following property. Suppose
    \begin{itemize}
        \item $V,V^{[t]}$ are random variables, each supported on boolean $\ell \times m$ matrices,  s.t. for any $i \in [\ell]$, $V_i = U_m$; 
        \item for every $j\in[t]$, there exists an $h_j\in [\ell]$ such that $(V_{h_j},V_{h_j}^{j}) = (U_m,V_{h_j}^{j})$;
        \item $X,X^{[t]}$ are random variables independent of $V,V^{[t]}$, each supported on $d$ bits and $X$ has enough entropy, 
    \end{itemize}
    then 
    \begin{align*}
        \nipm_\ell(X,V)\approx_\eps U_{m_1}\mid(\nipm_\ell(X^1,V^1),\cdots,\nipm_\ell(X^t,V^t)).
    \end{align*}
\end{definition}

\begin{algorithm}[H]
    \caption{$\nipm_\ell(x,v)$}
    \label{alg:NIPM_l}
    \begin{algorithmic}
        \medskip
        \State \textbf{Input:} $x$ --- an $n$ bit string, $v$ --- an $\ell\times m$ bit matrix.
        \State \textbf{Output:} $z$ --- an $m\cdot \prod_{i=1}^\ell \alpha_i $ bit string where each $\alpha_i$ is defined below for $i\in[\ell]$. 
        \\\hrulefill
        \State \textbf{Sub-Routines and Parameters: } \\
        Let $\delta_w=k/2n$ be a lower bound on the assumed entropy rate of $x$ in Definition~\ref{def:ldnipm}, $\delta_q=1/2$ a lower bound on the entropy rate of for each $v_{h_i}$ where $i\in [\ell]$. Let $\alpha_1 = m/((3+3t)n)$.\\
        For $i\in [\ell-1]:$
        \begin{itemize}
            \item set $\delta_{\ref{thm:low-deg-lsext}}$ and  $\alpha_{\ref{thm:low-deg-lsext}}$ from Theorem~\ref{thm:low-deg-lsext} to be $\delta_w$ and $\alpha_i$ respectively;
            \item set $\delta_{\ref{thm:low-deg-lsext}}$ and  $\alpha_{\ref{thm:low-deg-lsext}}$ from Theorem~\ref{thm:low-deg-lsext} to be $\delta_q$ and  $\delta_w\beta_{\ref{thm:low-deg-lsext}}\alpha_i$ respectively;
            \item let $\alpha_{i+1}=\delta_q\delta_w\beta^2_{\ref{thm:low-deg-lsext}}\alpha_i$,
        \end{itemize}
        which gives 
        \begin{itemize}
            \item $\alpha_i = \pbra{\delta_w\delta_q}^{i-1}\beta_{\ref{thm:low-deg-lsext}}^{2i-2} \alpha_1$.
        \end{itemize}
        For $i\in [\ell-1]$: 
        \begin{itemize}
            \item $\LSExt_w^i: \bin^n \times \bin^{\alpha_i n} \to \bin^{\delta_w  \beta_{\ref{thm:low-deg-lsext}} \alpha_i n}$ be the extractor from Theorem~\ref{thm:low-deg-lsext} with error $\eps^w_i=2^{-\Omega(\delta_w^{2i-1}\delta_q^{2i-2}\alpha_1^2n)}$. 
            \item $\LSExt_q^i: \bin^{m} \times  \bin^{\delta_w  \beta_{\ref{thm:low-deg-lsext}} \alpha_i n} \to \bin^{\alpha_{i+1} n}$ be the extractor from Theorem~\ref{thm:low-deg-lsext} with error $\eps^q_i=2^{-\Omega(\delta_q^{2i-1}\delta_w^{2i-1}\alpha_1^2m)}$. 
        \end{itemize}
        \hrulefill \\
    Let $s_1 = \Slice(v_1,\alpha_1 n)$. \\
        For $i\in \sbra{\ell-1}$:
        \begin{enumerate}
            \item $r_i = \LSExt^i_w(x,s_i)$
            \item $s_{i+1} = \LSExt^i_q(v_{i+1},r_i)$
        \end{enumerate}
        Let $z=s_\ell$. 
    \end{algorithmic}
\end{algorithm}
    
    \begin{theorem}[$\nipm_\ell$]
        \label{thm:ldnipm}
        For every $\ell\in \N,\eps>0$, if there exists a large enough $C$ such that 
        \begin{itemize}
            \item $k\ge C \max\cbra{\pbra{(t+1)^2\log(1/\eps)n^{2\ell-1}/m^3}^{1/(2\ell -3)},(t+1)\sqrt{n}}$;
            \item $n\ge m\ge C\max\cbra{(t+1)n\sqrt{n}/k,((t+1)\log(1/\eps)n^{2\ell-1}/k^{2\ell-3})^{1/3}}$,
        \end{itemize}
         then there exists an $\nipm_\ell:\bin^n\times\pbra{\bin^m}^\ell \to \bin^{m_1}$ and constants $0<\eta<1$ and $\cbra{\eps_i^w,\eps_i^q}_{i=1}^{\ell-1}$ each larger than $0$ such that 
         \begin{itemize}
             \item $\eps_i^w = \eps^{\Omega((n^{2}/k^{2})^{\ell-i})}$.
             \item $\eps_i^q=\eps^{\Omega((n^2/k^2)^{\ell-i})}$.
             \item $m_1\ge\eta^{\ell}m/(t+1)$.
             \item each output bit of $\nipm_\ell$ is a degree $2^{\Theta(\ell)}$ polynomial of the input bits. 
         \end{itemize}
    \end{theorem}
    In the analysis of $\nipm_\ell$ (and later $\ldACB$), we will be using the lemma below repeatedly. It is adjusted from Lemma~\ref{lemma:ind-merging} for affine sources. See Appendix~\ref{app:missing-proofs} for a proof.
\begin{restatable}[Independence-merging lemma for affine sources]{lemma}{indmergingaffine}
\label{lemma:ind-merging-affine}
    Let $\LExt:\bin^n \times \bin^d \to \bin^m$ be any $(k,\eps)$-strong linear seeded extractor, $X_0 \in \bin^n$ an affine source, $X,X^{[t]}\in \bin^n$, $Y,Y^{[t]}\in\bin^d$ all linear functions of $X_0$, $W=\LExt(X,Y)$ and $W^j=\LExt(X^j,Y^j)$ for every $j\in[t]$.
    Suppose there exists $S,T \subseteq [t]$ such that 
    \begin{itemize}
        \item $(Y,Y^S) \approx_{\delta} (U_d,Y^S)$;
        \item $H(X \mid X^T,Y,Y^{[t]}) \ge k+tm$,
    \end{itemize}
    then 
    \begin{align*}
        W\approx_{\eps+\delta} U_m \mid ( W^{S\cup T}, Y, Y^{[t]}). 
    \end{align*}
\end{restatable}
    
    \begin{proof}[Proof of Theorem~\ref{thm:ldnipm}]
        We will show that Algorithm~\ref{alg:NIPM_l} is such an $\nipm_\ell$. We first argue about the degree of each output bit. Let the degree of $s_i$ be $d_i$ for all $i\in [\ell]$, then they satisfy the following recursive formula
        \begin{align*}
            d_{i} = \begin{cases}
            1 & \text{if $i = 1$} \\
            3(3d_{i-1}+1)+1 = 9d_{i-1}+4 & \text{if $i > 1$}
        \end{cases}
        \end{align*}
        solving which gives us $d_\ell = \frac{9^{\ell -1}-1}{2}$. \\
        We now use induction to show the following claim. We let $R^{[j]}_{h_{[j]}}:=\{R^1_{h_1},R^2_{h_2},\cdots,R^j_{h_j}\}$.
        \begin{claim}
            Without loss of generality, let $1\le h_1\le\cdots\le h_t\le \ell$. For every $j \in [t]$, the following holds after step $h_j$
            \begin{align*}                         
            S_{h_j} 
                & \approx_{\sum_{i\in[h_j-1]}\eps^w_{i}+\sum_{i\in[h_j-1]}\eps^q_i} U_{\alpha_{h_j}n} \mid (S_{[i-1]},S_{[i-1]}^{[t]},R_{[i-1]},R_{[i-1]}^{[t]}), \\            R_{h_j}&\approx_{\sum_{i\in[h_j]}\eps^w_{i}+\sum_{i\in[h_j-1]}\eps^q_i} U_{\delta_w\beta_{h_j}^w \alpha_{h_j} n} \mid (R_{h_{[j]}}^{[j]},S_{[h_j]},S_{[h_j]}^{[t]},R_{[h_j-1]},R^{[t]}_{[h_j-1]}),
            \end{align*}
            which implies that 
            \begin{align*}
                S_\ell                 \approx_{\sum_{i\in[\ell-1]}(\eps^q_{i}+\eps^w_{i})} U_{\alpha_\ell n}\mid (V_{h_{[t]}}^{[t]},S_\ell^{[t]},S_{[\ell-1]},S_{[\ell-1]}^{[t]},R_{[\ell-1]},R^{[t]}_{[\ell-1]}).
            \end{align*}
        \end{claim}
        
        \begin{proof}
            We skip writing errors explicitly below whenever they can be easily seen to follow the claim. \\
            \textbf{Case $i\le h_1-1$.} We prove by induction that 
            \begin{align}
                \label{eqn:nipm1}
               \begin{split} S_i &\approx_{\sum_{j\in[i-1]}\eps^w_j+\sum_{j\in[i-1]}\eps^q_j} U_{\alpha_i n} \mid (S_{[i-1]},S_{[i-1]}^{[t]},R_{[i-1]},R_{[i-1]}^{[t]}),  \\ R_i&\approx_{\sum_{j\in[i]}\eps^w_{j}+\sum_{j\in[i-1]}\eps^q_j} U_{\delta_w\beta^w_i\alpha_i  n} \mid (R_{[i-1]},R_{[i-1]}^{[t]},S_{[i]},S_{[i]}^{[t]}).
               \end{split}
            \end{align}
            In round $1$, since $V_1=U_m$, $S_1=U_{\alpha_1 n}$. Then 
            by Lemma~\ref{lemma:ind-merging-affine}, $(R_1,S_1,S_1^{[t]})\approx_\eps (U_{\delta_w\beta_{1}^w \alpha_{1} n},S_1,S_1^{[t]})$. Then, assume that Eqn.~\eqref{eqn:nipm1} holds $\forall i\in[h_1-2]$. Since 
            
            $H(V_{i+1} \mid R_{[i]},R_{[i]}^{[t]}, S_{[i]},S_{[i+1]}^{[t]})=H(V_{i+1} \mid S_{[i]},S_{[i+1]}^{[t]})\ge m-(t+1)(\sum_{j=1}^{i+1} \alpha_j n)\ge m/2$, by the property of strong seeded extractor, the first part of Eqn.~\eqref{eqn:nipm1} holds. Since 
            
            $H(X \mid S_{[i]},S_{[i]}^{[t]}, R_{[i]},R_{[i+1]}^{[t]})=H(X \mid R_{[i]},R_{[i+1]}^{[t]})\ge k - (t+1)\delta_w(\sum_{j=1}^{i+1} \beta_{\ref{thm:low-deg-lsext}} \alpha_i n) \ge k/2 $, by the property of strong seeded extractor, the second part of Eqn.~\eqref{eqn:nipm1} holds. \\
            \textbf{Case $i=h_j$.} We prove by induction that 
            \begin{align}\label{eqn:nipm2}
            \begin{split}
            S_{h_j}&\approx U_{\delta_w\beta_{h_j}^w\alpha_{h_j}n} \mid (V_{h_{[j]}}^{[j]},S_{h_{[j]}}^{[j]},S_{[h_j-1]},S_{[h_j-1]}^{[t]},R_{[h_j-1]},R^{[t]}_{[h_j-1]}), \\
            R_{h_j}&\approx U_{\alpha_i n}\mid  (R_{h_{[j]}}^{[j]},S_{[h_j]},S_{[h_j]}^{[t]},R_{[h_j-1]},R^{[t]}_{[h_j-1]}),
            \end{split}
            \end{align} 
            then by Lemma~\ref{lemma:ind-merging-affine}, for $i\in[h_{j}+1,h_{j+1}-1]$, $j\in[t]$ (we defined $h_{t+1}-1:=\ell$), it holds that
            \begin{align}\label{eqn:nipm3}
                \begin{split}
                S_{i} & \approx U_{\delta_w\beta_{i}^w\alpha_{i}n} \mid (V_{h_{[j]}}^{[j]},S_i^{[j]},S_{[i-1]},S_{[i-1]}^{[t]},R_{[i-1]},R^{[t]}_{[i-1]}), \\
                R_{i} & \approx U_{\alpha_i n}\mid  (R_{i}^{[j]},S_{[i]},S_{[i]}^{[t]},R_{[i-1]},R^{[t]}_{[i-1]}).    
                \end{split}
            \end{align}
            In round $h_1-1$,  $(V_{h_1},V_{h_1}^1)\approx (U_m,V_{h_1}^1)$. By Lemma~\ref{lemma:ind-merging-affine},
            
            $S_{h_1}\approx U_{\delta_w\beta_{h_1}^w\alpha_{h_1}n}\mid (V_{h_1}^1,S_{h_1}^1,S_{[h_1-1]},S_{[h_j-1]}^{[t]},R_{[h_1-1]},R^{[t]}_{[h_1-1]})$. Then, again by Lemma~\ref{lemma:ind-merging-affine}, $R_{h_1} \approx U_{\alpha_1 n}\mid  (R_{h_1}^{1},S_{[h_1]},S_{[h_1]}^{[t]},R_{[h_1-1]},R^{[t]}_{[h_1-1]})$. Assume that Eqn.~\eqref{eqn:nipm2} holds for $i\in h_{[j-1]}$, $j\in[t]$. 
            Since $(V_{h_j},V_{h_j}^j)\approx (U_m,V_{h_j}^j)$ and that the second equation in Eqn.~\eqref{eqn:nipm3} for $i\in [h_{j}-1]$ holds, then by Lemma~\ref{lemma:ind-merging-affine}, Eqn.~\eqref{eqn:nipm2} holds for $i=h_j$. 
        \end{proof}
        Lastly, the setting of the parameters is similar to that of Theorem~\ref{thm:laext}. This completes the proof of Theorem~\ref{thm:ldnipm}.
    \end{proof}

\paragraph{Low-degree advice correlation breaker}
Now, we are ready to give the construction of the low-degree correlation breaker. We first give the definition of low-degree advice correlation breaker. 

\begin{definition}[$\ldACB$]\label{def:ldacb}
    A function $\ldACB:\bin^n \times \bin^d \times \bin^a \to \bin^m$ is an advice correlation breaker for linearly correlated sources if the following holds. Let
    \begin{itemize}
        \item $A,A^{[t]},B,B^{[t]}$ be random variables on $\bin^n$ and $Y,Y^{[t]}$ be random variables on $\bin^d$ such that $(A,A^{[t]})$ is independent of $(B,B^{[t]},Y,Y^{[t]})$. Moreover, $H(A) \ge k$ and $Y = U_d$;
        \item $X = A+B, X^i = A^i + B^i$ for every $i\in [t]$;
        \item $\alpha,\alpha^1, \cdots, \alpha^t$ be $a$-bit strings s.t. $\alpha \neq \alpha^i$ for every $i \in [t]$;
        \item each bit of the output is a constant degree polynomial of the inputs $X$ ($X^i$) and $Y$ ($Y^i$),
    \end{itemize}
    then 
    \begin{align*}
        \pbra{\ldACB(X,Y,\alpha) \approx_{\eps} U_m} \mid \pbra{\ldACB(X^1,Y^1,\alpha^1),\cdots, \ldACB(X^t,Y^t,\alpha^t)}.
    \end{align*}
    Moreover, if there are random variables $X',A',B'$ and $Y'$ such that $X'=A'+B'$ and $\pbra{Y=U_d} \mid Y'$, then it also holds that 
     \begin{align*}
        \pbra{\ldACB(X,Y,\alpha) \approx_{\eps} U_m} \mid \pbra{\ldACB(X',Y',\alpha),\ldACB(X^1,Y^1,\alpha^1),\cdots, \ldACB(X^t,Y^t,\alpha^t)}.
    \end{align*}
\end{definition}

We remark that Definition~\ref{def:ldacb} differs from standard definitions in that it allows conditioning on an tampered output with the same advice, given that the seed is non-malleable to the tampered seed. We will be using this property in our proof for directional affine extractors.

In our construction, we also need the following function.
\begin{definition}[$\FFAssign$~\cite{ChattopadhyayGL:focs:2021}]\label{def:ffassign}
Let $\FFAssign:\pbra{\bin^n}^2 \times \bin^a \to \pbra{\bin^n}^{2a}$ be defined as follows. Let $r_0,r_1 \in \bin^n$ and $\alpha \in \bin^a$. Let $\alpha_j$ denote the $j$-th bit of $\alpha$. Then $\FFAssign(r_0,r_1,\alpha):=(r_{\alpha_1},r_{1-\alpha_1},\cdots,r_{\alpha_a},r_{1-\alpha_a})$.
\end{definition}
    
\begin{algorithm}[H]
    \caption{$\ldACB(x,y,id)$}
    \label{alg:ldacb}
    \begin{algorithmic}
        \medskip
        \State \textbf{Input:} Bit strings $x=w+z,y, id$ of length $n,d,a$ respectively, where $d<n$.
        \State \textbf{Output:} Bit string $y'$ of length $n_2$.
        \State \textbf{Subroutines and Parameters:} \\
        Let $\LSExt:\bin^n \times \bin^{m_1} \to \bin^{m_2}$ from Theorem~\ref{thm:low-deg-lsext} with $m_1=d/(4+2t)$, output length $m_2=\beta_{\ref{thm:low-deg-lsext}}kd/((8+4t)n)$, entropy $k/2$ and
        error $\eps_1$. \\
        Let $\laExt:\bin^d\times \bin^{m_2} \to \pbra{\bin^{v}}^2$ from Theorem~\ref{thm:laext} where $v=\Omega(m_1/(16+16t))=\Omega(d/(32(1+t)^2))$ with entropy $d/3$ and error $\eps_2$.\\
        Let $\FFAssign:\pbra{\bin^v}^2 \times \bin^a \to \pbra{\bin^v}^{2a}$ from Definition~\ref{def:ffassign}.\\
        Let $\nipm_{2a}:\bin^n\times\pbra{\bin^v}^{2a}\to \bin^{n_2}$ from Theorem~\ref{thm:ldnipm} with entropy $k/2$ and error $\eps_3$. 
        \\\hrulefill \\
        \begin{enumerate}
            \item Let $s =  \Slice(y,m_1)$.
            \item Let $q =\LSExt(x,s)$.
            \item Let $(r_0,r_1) =\laExt(y,q)$.
            \item Let $(v_{1},v_{2},\cdots,v_{(2a-1)},v_{2a}) =\FFAssign((r_0,r_1),\alpha)$.
            \item Output $v^* = \nipm_{2a}(x,v_{1}\circ \cdots \circ v_{2a})$.
        \end{enumerate} 
    \end{algorithmic}
\end{algorithm}

\begin{theorem}[$\ldACB$]
    \label{thm:ldacb}
    For every $0<\eps<1$ and $n\in\N$ and every $k,d,t,a$, there exists a large enough $C$ such that if
    \begin{itemize}
        \item $k\ge C \max\cbra{\pbra{(t+1)^2\log(1/\eps)n^{4a-1}/m^3}^{1/(4a -3)},(t+1)\sqrt{n}}$;
        \item $d\ge C (t+1)^2 \max\cbra{(t+1)n\sqrt{n}/k,((t+1)\log(1/\eps)n^{4a-1}/k^{4a-3})^{1/3}}$,
    \end{itemize}
    then there exists a constant $1>\eta>0$ and an $\ldACB:\bin^n\times \bin^d\times \bin^a \to \bin^m$ which is a low-degree advice correlation breaker for linearly correlated sources s.t.
    \begin{itemize}
        \item $m = \Omega(\eta^{2a}kd/((t+1)^3n))$;
        \item each output bit of $\ldACB$ is a degree $2^{\Theta(2a)}$ polynomial of the input. 
    \end{itemize}
\end{theorem}
\begin{proof}
    We will prove that Algorithm~\ref{alg:ldacb} gives such a function. \\
    First we prove that $\ldACB$ satisfy Definition~\ref{def:ldacb}. 
    \begin{enumerate}
        \item Let $Q_A:= \LSExt(A,S)$, $Q'_A:= \LSExt(A',S')$, $Q_B:= \LSExt(B,S)$, $Q'_B:= \LSExt(B',S')$. Also for all $i\in [t]$, let $Q_A^i:= \LSExt(A^i,S^i)$, let $Q_B^i:= \LSExt(B^i,S^i)$.
        \item Since $Y=U_d \mid Y'$, then $S=U_{m_1}\mid Y'$. Since $H(X\mid Y',S,S^{[t]})\ge H(A)\ge k \ge k/2 + (t+2)m_2$, by Lemma~\ref{lemma:ind-merging-affine} $Q\approx_{\eps_1} U_{m_2}\mid(Q',Y',S,S^{[t]})$.
        \item First note that conditioned on $S$, since $\LSExt$ is a linear function, $Q = Q_A+Q_B$. Moreover, we have that $Y$ is independent of $Q_A$ further conditioned on $Q_B$. Since $H(Y\mid Y',S,S^{[t]},Q_B,Q_B^{[t]},Q')\ge d-(t+1)(m_1+m_2) \ge d/3$ and $Q_A\approx_{\eps_1} U_{m_2}\mid(Q',Y',S,S^{[t]},Q_B,Q_B^{[t]})$, $$R_0\approx_{\eps_1+\eps_2} U_v \mid (S,S^{[t]},Q,Q^{[t]},Q',Y',R_0')$$ and $$R_1\approx_{\eps_1+\eps_2} U_v \mid (S,S^{[t]},Q,Q^{[t]},Q',Y',R_0,R_0',R_0^{[t]},R_1').$$
        \item By the Definition~\ref{def:ffassign}, $V_i\approx_{\eps_1+\eps_2} U_v\mid V',\forall i\in[2a]$. In addition, for every $i\in[t]$, there exists $h_i\in[2a]$ s.t. $V_{h_i}=R_1$ and $V_{h_i}^i=R_0^i$. Therefore for every $i\in[t]$, there exists $h_i\in[2a]$ s.t. $V_{h_i}\approx_{\eps_1+\eps_2} U_v\mid V_{h_i}'$. 
        \item Since $H(X\mid Q,Q',Q^{[t]},Y',S,S^{[t]},V'^*)\ge k-(2+t)m_2-2a v\ge k/2$, by Theorem~\ref{thm:ldnipm}, $V^*\approx_{\eps_1+\eps_2+\eps_3} U_{n_2}\mid (V^{[t]*},V'^*)$.
    \end{enumerate}
    Now since $\LSExt$ and $\laExt$ cause a constant increase in the degree of the output bits, and $\nipm_{2a}$ cause a $2^{\Theta(2a)}$ increase in the degree. each output bit of $\ldACB$ is a degree $2^{\Theta(\ell)}$ polynomial of the input. \\
    Finally, the parameters constraints follows from those of Theorem~\ref{thm:ldnipm}. This completes the proof of Theorem~\ref{thm:ldacb}.
\end{proof}

\subsection{Directional Affine Extractor for Linear Entropy}

Apart from the low-degree correlation breaker, we still need the following extractors as building blocks.

\begin{theorem}[\cite{CG88}]\label{thm:IP}
For every constant $\delta > 0$, there exists a polynomial time algorithm $\IP:(\cbra{0,1}^n)^2 \to \bin^m$ such that if $X$ is an $(n,k_1)$ source, $Y$ is an independent $(n,k_2)$ source and $k_1+k_2\ge (1+\delta)n$, then 
\begin{align*}
    \IP(X,Y) \approx_\eps U_m \mid Y,
\end{align*}
where $\eps = 2^{-\frac{\delta n - m -1}{2}}$.
\end{theorem}

\begin{theorem}[\cite{Li:CCC:2011}]\label{thm:affinesrext}
    For every affine $t\times r$ somewhere random source $X$, there exists a function $\AffineSRExt$ such that $\AffineSRExt(X)$ outputs $m=r/t^{O(\log t)}$ bits that are $2^{-\Omega(r/t^{O(\log t)})}$-close to uniform. Moreover, each bit of the output is a degree $t^{O(1)}$ polynomial of the bits of the input.
\end{theorem}

\begin{theorem}[Seeded non-malleable extractor~\cite{Li:focs:2012}]\label{thm:snmExt}
    For any constant $1>\delta>0$, let $X$ be an $(n,k)$-source with $k=(1/2+\delta)n$ and $Y$ be the uniform distribution on $\bin^{n/2-1}$ independent of $X$. Let $b_1,\cdots, b_{n/2}$ be a basis of $\F_{2^{n/2}}$ regarded as a vector space over $\F_2$. For each $b_i$, let $\overline{Y}_i = \pbra{b_i Y, b_i Y^3}$ where $Y$ is regarded as an element in $\F^{*}_{2^{n/2}}$ and define one bit $Z_i = \IP(X,\overline{Y}_i)$ where $\IP$ is the inner product function over $\F_2^n$. 
    Choose $m=\Omega(n)$ bits from $\cbra{Z_i}$, let $\snmExt(X,Y) = (Z_{i_1},\cdots,Z_{i_m})$. Let $\mathcal A:\bin^n \to \bin^n$ be any function without fixed point, then 
    \begin{align*}
        \abs{\snmExt(X,Y), \snmExt(X, \mathcal A(Y)), Y - U_n,\snmExt(X, \mathcal A(Y)), Y} \le 2^{-\Omega( n)}.
    \end{align*}
\end{theorem}

The following proposition about strong linear seeded extractors and affine sources is useful to us.

\begin{proposition}[\cite{Rao:ccc:09}]\label{prop:uniform}
    Let $\Ext:\bin^n \times \bin^d \to \bin^m$ be a linear strong seeded extractor for min-entropy $k$ with error $\eps < 1/2$. Let $X$ be any affine source with entropy $k$. Then
    \[\Pr_{u \leftarrow_U U_d}[\Ext(X,u) = U_m] \ge 1-\eps.\]
\end{proposition}

We use the following lemma when arguing about strong seeded extractors with deficient seed.

\begin{lemma}[\cite{CGL:stoc:16}]\label{lemma:deficient seed}
    Let $\Ext:\bin^n \times \bin^d \to \bin^m$ be strong seeded extractor for min-entropy $k$, and error $\eps$. Let $X$ be a $(n,k)$-source and let $Y$ be a source on $\bin^d$ with min-entropy $d-\lambda$. Then
    \[\Ext(X,Y) \approx_{2^\lambda \eps} U_m \mid Y.\]
\end{lemma}

We now present our construction of directional affine extractor.

\begin{algorithm}[H]
    \caption{$\DAExt(x)$}
    \label{alg:DAExt}
    \begin{algorithmic}
        \medskip
        \State \textbf{Input:} $x$ --- an $n$ bit string.
        \State \textbf{Output:} $z$ --- an $m$ bit string with $\Omega(n)$.
        \\\hrulefill 
        \State \textbf{Sub-Routines and Parameters: } \\
        Let $\ell_1=\poly(2/\delta),\ell_1'=(n/(2(m'+k+1)))^{\log(2d_{\ref{thm:dimexpand}}+2)}$ where $k$ is defined below, $\ell_2 = \poly(4/\delta)=\ell_3$. Let $\bcond$ be the basic condenser from Algorithm~\ref{alg:bcond}.\\
        Let $\scond_i:\bin^n \to \pbra{\bin^{n/\ell_i^{1/\log (2 d_{\ref{thm:dimexpand}}+2)}}}^{\ell_i}$ for $i\in \{1,3\}$, $\scond_2:\bin^{n/t} \to \pbra{\bin^{n/\pbra{t\ell_2^{1/\log (2 d_{\ref{thm:dimexpand}}+2)}}}}^{\ell_2}$, be linear affine condensers from Theorem~\ref{thm:Lcondmain}. \\ 
        
        Let $\IP:\pbra{\{0,1\}^{n/\pbra{t\ell_2^{1/\log (2 d_{\ref{thm:dimexpand}}+2)}}}}^2\to \{0,1\}^{\Omega(n)}$ be the two-source extractor from Theorem~\ref{thm:IP} with error $\eps_1=2^{-\Omega(n)}$, set up to extract from two independent sources whose entropy rates sum up to more than $1+2 \beta_{\ref{thm:Lcondmain}}$. \\
        Let $\AffineSRExt$ be the extractor for affine somewhere random sources from Theorem~\ref{thm:affinesrext} with error $\eps_2=2^{-\Omega(n)}$. \\
        Let $\LSExt:\{0,1\}^n\times\{0,1\}^d \to \{0,1\}^{m'}$ be the strong linear seeded extractor from Theorem~\ref{thm:low-deg-lsext} set to extract from entropy $\delta n/2$ with error $\eps_3=2^{-\Omega(n)}$. \\
        Let $\Enc:\bin^n \to \bin^{\lambda n}$ be the encoding function of an asymptotically good linear binary code with constant relative rate $1/\lambda$ and constant relative distance $\beta$.\\
        Let $\snmExt:\bin^{2(m'+k+1)}\times \bin^{m'+k}\to\bin^{n_1}$ be the seeded non-malleable extractor from Theorem~\ref{thm:snmExt} with error $\eps_4=2^{-\Omega(n)}$. Choose $m'$ and $k$ such that $\log(n/(m'+k+1))\in \N$ where 
        \begin{itemize}
            \item $m' \le  \beta_{\ref{thm:Lcondmain}}\delta^2 n /(300 t\ell_2 \ell'_3)$, $k=\Omega(n)\le \frac{n_1}{20\log(\lambda n/n_1)}$.
        \end{itemize}
        Let $\ldACB:\bin^n\times\bin^{n_1}\times\bin^{\log\ell_1'}\to \bin^{n_2}$ be the advice correlation breaker from Theorem~\ref{thm:ldacb} with output length $n_2=O(\delta n_1/\poly(\ell_1'))$ and error $\eps_5=2^{-\Omega(n)}$. 
    \end{algorithmic}
\end{algorithm}

\clearpage

\begin{breakalgo}
    \\ \\
    %
    %
        Let $\G$ be the generating matrix of an asymptotically good linear binary code with codeword length
         $m_1$ and constant relative distance $\gamma$. Thus $\G$ is an $\alpha m_1 \times m_1$ matrix for some constant $\alpha>0$. Let $\G_i$ stand for the $i$'th row of the matrix.
        \\\hrulefill \\
    Let $sc_{1} \circ sc_{2} \circ \cdots \circ sc_{\ell_1'}=\bcond^r \circ \scond_1(x)$, where $r=\log(n/(m'+k+1))-1-\ell_1^{1/\log(2d_{\ref{thm:dimexpand}}+2)}$. \\
    Divide $x$ into $t$ blocks $x = x_1\circ \cdots \circ x_t$ where $t = 2^{\lceil\log (10/\delta)\rceil} \ge \delta/10$
        and each block has $n/t$ bits. \\
        For every $i$, $1\le i\le t$ do the following. 
        \begin{enumerate}
            \item Let $y_{i1}\circ \cdots \circ y_{i \ell_2} = \scond_2(x_i)$, where $y_{ij}$ is the $j$'th row of the matrix obtained by applying $\scond_2$ to $x_i$. Note that $\ell_2=O(1)$ and each $y_{ij}$ has $\Omega(n)$ bits.
            \item Apply $\bcond^{\log t}\circ\scond_3$ on $x$. That is, first apply $\scond_3$ on $X$, and then apply $\bcond$ $\log t$ times on the output so that we get $\ell_3'$ blocks $\bcond^{\log t}\circ\scond_3(x)=x_1' \circ \cdots \circ x'_{\ell_3'}$, of equal size with each block having the same number of bits as $y_{ij}$. Note that $\ell_3' = O(1)$. 
            \item Apply $\IP$ to every pair of $x_{j_1}'$ and $y_{ij_2}$, and output $\beta_{\ref{thm:Lcondmain}}\delta^2 n /(300 t\ell_2 \ell'_3)$ bits. Let $sr_i$ be the matrix obtained by concatenating all the outputs $\IP(x_{j_1}',y_{ij_2})$, i.e., each row of $sr_i$ is $\IP(x_{j_1}',y_{ij_2})$ for a pair $(x_{j_1}',y_{ij_2})$.
            \item Let $r_i = \AffineSRExt(sr_i)$.
            \item Let $u_i = \LSExt(x,r_i)$, set up to output $m'$ bits.
            \item Divide $u_i$ into $u_{i1} \circ u_{i2}$ where $u_{i1}$ has $k\log(\lambda n/k)\le n_1/10$ bits and $u_{i2}$ has $\ge m'-n_1/10 $ bits.
            \item Divide $\Enc(x)$ into $k$ blocks of equal size such that $\Enc(x) = \tilde{x}^1 \circ \tilde{x}^2 \circ \cdots \circ \tilde{x}^k$ where each block has $O(1)$ bits. Divide $u_{i1}$ into $k$ equal blocks $u_{i1}^{(1)} \circ \cdots \circ u_{i1}^{(k)}$. Let $h_i={\tilde{x}^1}_{\mid u_{i1}^{(1)}} \circ {\tilde{x}^2}_{\mid u_{i1}^{(2)}} \circ \cdots \circ {\tilde{x}^k_{\mid u_{i1}^{(k)}}}$ and $\tilde{u}_i = u_i \circ h_i$.
            \item Let $sn_{ij_3}$ be $\snmExt$ applied to each $sc_{j_3}$ and $\tilde{u}_i$ and output $n_1\le m'/100$ bits. Let $sn_i$ be the $\ell_1'\times n_1$ matrix obtained by concatenating $sn_{ij_3}$ for $j_3\in [\ell_1']$, i.e., the $j$-th row of $sn_i$ is $sn_{ij}$.
            \item Let $\tilde{y}_i = \bigoplus_{j=1}^{\ell_1'}\ldACB(x,sn_{ij},j)$ and output $n_2\le m'/10000$ bits. 
           \item Let $w_i=\LSExt(x,\tilde{y}_i)$, set up to output $n_3\le m'/1000000$ bits.
           \item Divide the bits of $w_i$ into $s_i = \Omega(n)$ blocks of equal size, with each block having $c_i$ number of bits for some constant $c_i$ to be chosen later. For every $j=1,\cdots, s_i$, compute one bit $v_{ij}$ by taking the product of all the bits in the $j$'th block, i.e., $v_{ij} = \prod_{(j-1)c_i+1}^{jc_i} w_{i\ell}$.
        \end{enumerate}
        Output $m_1 = \Omega(n)$ bits $\{z_j = \bigoplus_{i=1}^t v_{ij}\}$. \\
        \textbf{Disperser to Extractor.} \\
        For each codeword $\G_i$, let $S_i = \cbra{j\in [m_1]: \G_{ij} = 1}$ be the set of indices s.t. the bit of the codeword $\G_i$ at those indices are $1$. Define
        \begin{align*}
            o_i = \bigoplus z_{j:j\in S_i}
        \end{align*}
        to be the bit associated with $\G_i$, i.e., $o_i$ is the $\mathrm{XOR}$ of the $z_j$'s whenever the $j$'th index of the codeword $\G_i$ is $1$. \\
        Take a constant $0<\beta'\le \alpha$, where $\beta'$ is chosen later. Output $o=(o_1,\cdots,o_{\beta' m_1})$. \\
\end{breakalgo}

\begin{theorem}\label{thm:daext}
    For any constant $0<\delta \le 1$, there exists a family of functions $\DAExt:\bin^n \to \bin^m$ where $m=\Omega(n)$, such that for any affine source $X$ of min-entropy at least $\delta n$, any nonzero $a\in\bin^n$, it holds that \[ (\DAExt(X), \DAExt(X+a))  \approx_{\eps} (U_m, \DAExt(X+a)),\] where $\eps = 2^{-\Omega(n)}$.
\end{theorem}
\begin{proof}
    
    In the proof below, we have in mind two $\DAExt$ running in parallel, one with input $X$, the other with input $X+a=:X'$. We use $\{X_i\}_{i\in[t]}$, $\{SR_i\}_{i\in[t]}$, $\{R_i\}_{i\in[t]}$, $\{U_i\}_{i\in[t]}$, $\{\tilde{U}_i\}_{i\in[t]}$, $\{SN_i\}_{i\in [t]}$, $\{SC_i\}_{i\in[\ell_1]}$, $\{\tilde{Y}_i\}_{i\in[t]}$, $\{W_i\}_{i\in[t]}$ to denote the random variables generated in $\DAExt(X)$ and $\{X'_i\}_{i\in[t]}$, $\{SR'_i\}_{i\in[t]}$, $\{R'_i\}_{i\in[t]}$, $\{U'_i\}_{i\in[t]}$, $\{\tilde{U}'_i\}_{i\in[t]}$, $\{SN'_i\}_{i\in [t]}$, $\{SC'_i\}_{i\in[\ell_1]}$, $\{\tilde{Y}'_i\}_{i\in[t]}$, $\{W'_i\}_{i\in[t]}$ to denote the random variables generated in $\DAExt(X+a)$.
    
    Throughout the proof, we maintain a random variable $Z$. We update $Z$ each time a group of random variables has been fixed so that it represents all the random variables that have been fixed. We will make $Z$ explicit each time it is revised. Initially, $Z=0$.
    
    We now show that Algorithm~\ref{alg:DAExt} is an efficient family of such functions. We first argue there exists an iteration $g$ such that conditioned on all the random variables generated in the previous iterations, both $X$ and $X_g$ have $\Omega(\delta)$ entropy rate. 
    
    \begin{lemma}\label{lemma:first fixings}
        There exists $1\le g \le t$ s.t. conditioned on any fixings of
        \begin{align*}
            (X_i,X_i',SR_i,SR_i',R_i,R_i',\tilde{U}_i,\tilde{U}_i',SN_i,SN_i',\tilde{Y}_i,\tilde{Y}_i', W_i,W_i')_{i\in [g-1]}
        \end{align*}
        in order, $X$ is an affine source with $H(X_g)\ge \delta n/(4t)$ and $H(X)\ge 3\delta n/5+ \delta n/(3t)$.
    \end{lemma}
    \begin{proof}
        By Lemma~\ref{lemma:affine entropy}, when dividing $X$ into $t$ blocks, there exist positive integers $k_1,\cdots,k_t$ which sum up to $\delta n$ such that for any $i\in [t]$, conditioned on the fixing of $X_1,\cdots,X_{i-1}$, $H(X_i) = k_i$. Therefore, there must exists an $i$ such that $k_i \ge \delta n/(3t)$. Let $g$ be the minimal index such that $H(X_g)=k_g \ge \delta n/(3t)$. 
        \begin{enumerate}
            \item Consider the affine source $X$ and $X'=X+a$. Once we fix $(X_i = x_i)_{i \in [g-1]}$, $(X'_i:=X_i+a_i=x_i+a_i)_{i \in [g-1]}$ are also fixed. Since $X_i$ is an affine funciton of $X$, after this fixing, $X$ and $X'$ are still affine sources. By Lemma~\ref{lemma:affine entropy}, after this fixing $H(X_g) = H(X_g') \ge k_g \ge \delta n/(3t)$ and $H(X) = H(X') = \sum_{i=g}^t k_i \ge \delta n - (t-1)\cdot \delta n/(3t) \ge 2\delta n/3 + \delta n/(3t)$. Now set $Z= \cbra{X_i,X_i'}_{i\in[g-1]}$.
            \item Note that conditioned on the fixing of $(X_i=x_i)_{i \in [g-1]}$ (thus $(X'_i=x'_i)_{i \in [g-1]}$), both $(SR_i)_{i \in [g-1]}$ and $(SR'_i)_{i \in [g-1]}$ are affine functions of $X$. In general, fixing $(SR_i=sr_i)_{i\in[g-1]}$ does not necessarily fix $(SR_i'=sr'_i)_{i\in[g-1]}$ and in the worst cases $SR_i'=sr_i'$ may be linearly independent with $SR_i=sr_i$. Let $\overline{SR} = SR_1\circ \cdots \circ SR_{g-1}$ and $\overline{SR'}= SR'_1\circ \cdots \circ SR'_{g-1}$. By Lemma~\ref{lemma:affine bound}, since $\overline{SR}\circ \overline{SR'}$ has at most $(\beta_{\ref{thm:Lcondmain}}\delta^2 n/(300t)) \cdot t \cdot 2=\beta_{\ref{thm:Lcondmain}}\delta^2 n/150$ bits, $H(X \mid Z,\overline{SR}\circ\overline{SR'})\ge  2\delta n/3 + \delta n/(3t) - H(\overline{SR}\circ \overline{SR'} (X)) \ge 2\delta n/3 + \delta n/(3t) - \beta_{\ref{thm:Lcondmain}}\delta^2n/150$. \\
            Note that fixing $\overline{SR}\circ\overline{SR'}$ also fixes $\cbra{R_i,R_i'}_{i\in [g-1]}$. Now let $Z = Z \cup \cbra{SR_i,SR_i',R_i,R_i'}_{i\in [g-1]}$.
            \item Let $\overline{\tilde{U}} \circ \overline{\tilde{U}'}=\tilde{U}_1\circ \cdots \circ \tilde{U}_{g-1}\circ \tilde{U}_1'\cdots \circ \tilde{U}'_{g-1}$, then conditioned on any fixing of $Z$, $\overline{\tilde{U}}$ is an affine function of $X$ and it has at most $(\beta_{\ref{thm:Lcondmain}}\delta^2 n/(300t\ell_2\ell_3')) \cdot t \cdot 2=\beta_{\ref{thm:Lcondmain}}\delta^2 n/(150\ell_2\ell_3')$ bits. Therefore, by Lemma~\ref{lemma:affine bound} $H(X \mid Z,\overline{\tilde{U}}\circ \overline{\tilde{U}'})\ge 2\delta n/3 + \delta n/(3t) - \beta_{\ref{thm:Lcondmain}}\delta^2n/150 - \beta_{\ref{thm:Lcondmain}}\delta^2n/(150\ell_2\ell_3')$. Now, $Z = Z \cup \{U_i, \tilde{U}_i\}_{i\in [g-1]}$.
            \item Let $\overline{SN}\circ\overline{SN'}=SN_1\circ \cdots \circ SN_{g-1}\circ SN_1'\cdots \circ SN'_{g-1}$, then conditioned on any fixing of $Z$, $\overline{SN}\circ\overline{SN'}$ is an affine function of $X$ and it has at most $\beta_{\ref{thm:Lcondmain}}\delta^2 n/(15000\ell_2\ell_3')$ bits. 
            Therefore, by Lemma~\ref{lemma:affine bound} 
            $H(X \mid Z, \overline{SN}\circ\overline{SN'})\ge 2\delta n/3 - \beta_{\ref{thm:Lcondmain}}\delta^2n/150 - 1.01\cdot\beta_{\ref{thm:Lcondmain}}\delta^2n/(150\ell_2\ell_3')$. Now, $Z=Z\cup \{\overline{SN},\overline{SN'}\}$.
            \item Let $\overline{\tilde{Y}}\circ\overline{\tilde{Y}'} = \tilde{Y}_1\circ \cdots \circ \tilde{Y}_{g-1}\circ \tilde{Y}_1'\cdots \circ \tilde{Y}'_{g-1}$, then conditioned on any fixing of $Z$, $\overline{\tilde{Y}}\circ\overline{\tilde{Y}'}$ is an affine function of $X$ and it has at most $\beta_{\ref{thm:Lcondmain}}\delta^2 n/(1500000\ell_2\ell_3')$ bits. 
            Therefore, by Lemma~\ref{lemma:affine bound} $H(X\mid Z,\overline{\tilde{Y}}\overline{\tilde{Y}'})\ge 2\delta n/3 + \delta n/(3t) - \beta_{\ref{thm:Lcondmain}}\delta^2n/150 - 1.0101\cdot \beta_{\ref{thm:Lcondmain}}\delta^2n/(150\ell_2\ell_3')$. Now, $Z=Z\cup \{\overline{\tilde{Y}}\overline{\tilde{Y}'}\}$.
            \item Let $\overline{W}\circ \overline{W'} = W_1\circ \cdots \circ W_{g-1}\circ W_1'\cdots \circ W'_{g-1}$, then conditioned on any fixing of $Z$, $\overline{W}\circ\overline{W'}$ is an affine function of $X$ and it has at most $\beta_{\ref{thm:Lcondmain}}\delta^2 n/(150000000\ell_2\ell_3')$ bits. 
            Therefore, by Lemma~\ref{lemma:affine bound}  $H(X\mid Z,W,W')\ge 2\delta n/3 + \delta n/(3t) - \delta^2n/150 - 1.010101\cdot \beta_{\ref{thm:Lcondmain}}\delta^2n/(150\ell_2\ell_3')$. Now, $Z=Z\cup\{W,W'\}$.
            \item Therefore, $H(X_g\mid Z) \ge \delta n/(3t) - \beta_{\ref{thm:Lcondmain}}(\delta^2 n/150 + \delta^2 n/(150\ell_2\ell_3') + \delta^2 n/(15000\ell_2\ell_3') + \delta^2 n/(1500000\ell_2\ell_3') + \delta^2 n/(150000000\ell_2\ell_3')) = \delta n/(3t) - \beta_{\ref{thm:Lcondmain}}\delta^2 n/150 - 1.010101\cdot\beta_{\ref{thm:Lcondmain}}\cdot\delta^2 n/(150\ell_2\ell_3')> \delta n/(4t)$ and $H(X\mid Z)\ge 2\delta n/3 + \delta n/(3t) - \beta_{\ref{thm:Lcondmain}}\delta^2 n/150 - 1.010101\cdot\beta_{\ref{thm:Lcondmain}}\delta^2 n/(150\ell_2\ell_3')> 3\delta n/5+ \delta n/(3t)$.
        \end{enumerate}
    \end{proof}

\begin{lemma}
    \label{lemma:second fixings}
    With probability $1- 2^{-\Omega(n)}$ over the further fixings of $X_g$, $R_g$ is $2^{-\Omega(n)}$-close to uniform. 
\end{lemma} 
\begin{proof}
    We examine the execution of $\DAExt$ on the good block $X_g$ up to step $4$.
    \begin{enumerate}
        \item By Lemma~\ref{lemma:first fixings}, $H(X_g \mid Z)\ge \delta n/(4t) \overset{\text{Theorem~\ref{thm:Lcondmain}}}{\implies} Y_{g1}\circ\cdots\circ Y_{g\ell_2}:=\scond(X_g) =$ somewhere-rate-$(1/2+\beta_{\ref{thm:Lcondmain}})$ source.
        WLOG, assume $Y_{gi}$ has rate $1/2+\beta_{\ref{thm:Lcondmain}}$.
        \item By Lemma~\ref{lemma:affine conditioning}, $\exists A_g,B_g$ s.t. $X = A_g+B_g, X_g(X)=X_g(A_g), H(X_g)=H(A_g)$, and $X_g$ is independent with $B_g$.
        \item After fixing $X_g$, $B_g$ (thus $X$) has min-entropy at least $3\delta n/5 + \delta n/(3t) - \delta n/t \ge \delta n/4$. Now, let $Z=Z\cup \{X_g\}$.
        \item Since the $\ell_3'$ blocks $X = \tilde{X}_1\circ\cdots\circ \tilde{X}_{\ell_3'}$ are obtained by applying $\bcond^{\log t}\circ\scond_3$ on $X$, each $\tilde{X_i}$ is linear in X. Then $\tilde{X_i}(X) = \tilde{X_i}(A_g) + \tilde{X_i}(B_g)$ follows from Lemma~\ref{lemma:affine conditioning}. Let $\tilde{X_i}(A_g)=A_{gi}$ and $\tilde{X_i}(B_g)=B_{gi}$. By Theorem~\ref{thm:Lcondmain}, there exists one block with entropy rate at least $1/2+\beta_{\ref{thm:Lcondmain}}$, let $B_{gj}$ be such a block. 
        \item Note that $\IP(\tilde{X}_j,Y_{gi}) = \IP(A_{gj},Y_{gi}) + \IP(B_{gj},Y_{gi})$. Since $Y_{gi}$ has entropy rate $\ge 1/2+\beta_{\ref{thm:Lcondmain}}$ and $B_{gj}$ has entropy rate $\ge 1/2+\beta_{\ref{thm:Lcondmain}}$, by Lemma~\ref{thm:IP}, with probability $(1-\eps_1)$ over the fixing of $Y_{gi}$ (thus $A_g$ and $X_g$), $\IP(B_{gj},Y_{gi})$ is $\eps_1$ close to uniform. Since $A_g$ (thus $A_{gj}$) is fixed, the random variable $\IP(A_{gj},Y_{gi}) $ is fixed as well. Therefore, $\IP( B_{gj},Y_{gi})$ is $\eps_1$ close to uniform implies that $\IP(Y_{gi},\tilde{X}_j)$ is $\eps_1$ close to uniform. Therefore, with probability $(1-\eps_1-\eps_2)$ over the fixing of $X_g$, $SR_g$ is $(\eps_1+\eps_2)$ close to a somewhere random source.
        \item $SR_g\approx_{\eps_1+\eps_2}$ somewhere random source $\overset{\text{Theorem~\ref{thm:affinesrext}}}{\implies}$ $R_g\approx_{\eps_1+\eps_2}$ uniform.
    \end{enumerate}
\end{proof}
    
    \begin{lemma}
        With probability $1-2^{-\Omega(n)}$ over further fixings of $\pbra{SR_g,SR_g',R_g,R_g'}$,
        $U_g$ is uniform.
    \end{lemma}
    \begin{proof}
        Since $X_g$ is a linear function of $X$, conditioned on any fixing of it, it still holds that $X$ is an affine source. Moreover, conditioned on the fixing of $X_g$ (and thus $X_g'$ as well), $SR_g$ and $SR_g'$ are linear functions of $X$. By Lemma~\ref{lemma:affine conditioning}, there exists independent affine sources $\tilde{A}_g$ and $\tilde{B}_g$ s.t. $X=\tilde{A}_g + \tilde{B}_g$, $SR_g \circ SR'_g(X) = SR_g \circ SR'_g(\tilde{A}_g)$ and $H(SR_g \circ SR'_g) = H(\tilde{A}_g)$. Thus $H(\tilde{B}_g)=H(X) - H(\tilde{A}_g) =  H(X)-H(SR_g \circ SR'_g) \ge \delta n/2$. \\
        Next note $R_g$ is a deterministic function of $SR_g$ thus independent of $\tilde{B}_g$. In addition, $R_g$ is $2^{-\Omega(n)}$-close to uniform by Lemma~\ref{lemma:second fixings}.  Now, by Theorem~\ref{thm:low-deg-lsext} and Proposition~\ref{prop:uniform}, with probability $(1-\eps_3)$ over the fixings of $R_g$ (and thus with probability $(1-\eps_3)$ over the fixings of $SR_g$), $\LSExt(\tilde{B}_g,R_g)$ is uniform. Since $\LSExt$ is a linear function and $\LSExt(\tilde{A}_g,R_g)$ is fixed, with probability $(1-\eps_3)$ over the fixings of $R_g$, $\LSExt(X,R_g)$ is uniform.
    \end{proof}
 At this point, set $Z=Z\cup\{SR_g,SR_g',R_g,R_g'\}$. We have already shown that with high probability over the fixing of $Z$, $U_g$ is uniform. Now, we want to establish that $\tilde{U}_g$ and $\tilde{U}_g'$, which are $U_g$ and $U_g'$ appended with advice are linearly correlated, i.e., there exists an affine function $\mathcal A$ without fixed points s.t. $\mathcal A(\tilde{U}_g)=\tilde{U}_g'$. We achieve this in the following two Lemmas.  
    
\begin{lemma}
\label{lemma:affine correlation}
    Conditioned on $Z$, there exists $\tilde{A}_g,\tilde{B}_g$ s.t. $X=\tilde{A}_g+\tilde{B}_g$, $U_g(X)=U_g(\tilde{A}_g)$, and $U_g(\tilde{B}_g)=0$. Moreover, 
    $U_g'$ is linearly correlated with $U_g$ conditioned on any fixing of $U'_g(\tilde{B}_g)$.
\end{lemma}
\begin{proof}
    By Lemma~\ref{lemma:affine conditioning}, there exists $\tilde{A}_g,\tilde{B}_g$ s.t. $X=\tilde{A}_g+\tilde{B}_g$, $U_g(X)=\tilde{A}_g$, and $U_g(\tilde{B}_g)=0$. Moreover, there exists an affine function $L$ such that $\tilde{A}_g = L(U_g)$. Now conditioned on the fixing of $U'_g(\tilde{B}_g)$, $U'_g$ is an affine function of $\tilde{A}_g$, and thus an affine function of $U_g$.
\end{proof}
As a reminder, at this stage, we have
    \begin{align*}
        Z = \{X_i,X_i',SR_i,SR_i',R_i,R_i',\tilde{U}_i,\tilde{U}_i',SN_i,SN_i',\tilde{Y}_i,\tilde{Y}_i', W_i,W_i\}_{i\in [g-1]} \cup \{X_g, X_g', SR_g, SR_g', U'_g(\tilde{B}_g)\}.
    \end{align*}
\begin{lemma}
    Conditioned on the event that $U_g$ is uniform, with probability at least $1-2^{-{\Omega(n)}}$ over the fixing of $(Z,U_{g_1},U_{g_1}',H_g, H'_g)$, there exists an affine map $\mathcal A:\bin^{|\tilde{U}_g|} \to \bin^{|\tilde{U}_g|}$ without fixed point such that $\tilde{U}_g' = \mathcal{A}(\tilde{U}_g)$.
\end{lemma}
\begin{proof}
    We first show that $\tilde{U}_{g} \neq \tilde{U}'_{g}$ with high probability.
    \begin{itemize}
        \item If $U_{g1}\neq U_{g1}'$, then $\tilde{U}_g(X)\neq \tilde{U}_g'(X)$ always holds. 
        \item If $U_{g1}= U_{g1}'$, we show that $H_g \neq H_g'$ with probability $1-2^{-\Omega(n)}$. Note that this is equivalent to $H_g \oplus H_g' \neq 0 \iff (\tilde{X}^1\oplus \tilde{X}'^1)_{\mid U_{g1}^{(1)}} \circ (\tilde{X}^2\oplus \tilde{X}'^2)_{\mid U_{g1}^{(2)}} \circ \cdots \circ (\tilde{X}^k\oplus \tilde{X}'_k)_{\mid U_{g1}^{(k)}} \iff \tilde{a}^1_{\mid U_{g1}^{(1)}} \circ \tilde{a}^2_{\mid U_{g1}^{(2)}} \circ \cdots \circ \tilde{a}^k_{\mid U_{g1}^{(k)}}$ with probability $1-2^{-\Omega(n)}$ where $\tilde{a}^j$ for $j\in [k]$ are obtained by divide the $\Enc(a)$ (i.e. the encoded shift between the source $X$ and its tampering $X' =X+a$) into $k$ equal blocks such that $\Enc(a) = \tilde{a}^1 \circ \tilde{a}^2 \circ \cdots \circ \tilde{a}^k$. Let $\ell_1,\cdots,\ell_k$ be the number of bits in each block that are non-zero. Since at least $\beta$ fraction of bits in $\Enc(a)$ differs from the bits of the codeword $0=\Enc(0)$, $\sum_{i=1}^k \ell_i \ge \beta \lambda n$. Therefore we have  
        \begin{align*}
            \Pr\sbra{H_g \neq H_g' \mid U_{g1} = U'_{g1}} 
            & 
            = 1 - \prod_{i=1}^k \pbra{1-\frac{\ell_i}{\lambda n/k}} \\
            & 
            \ge 1- \pbra{\frac{\sum_{i=1}^k \pbra{1-\frac{\ell_i}{\lambda n/k}}}{k}}^k \\
            & 
            \ge 1 - \pbra{1-\beta}^k\\
            & 
            \ge 1 - 2^{-\Omega(n)}
            . 
            \tag{$k = \Omega(n)$}
        \end{align*}
        \end{itemize}
        Therefore, in total, with probability $1-2^{-\Omega(n)}$, $\tilde{U}_g\neq \tilde{U}'_g$.
        By Lemma~\ref{lemma:affine correlation}, $U_g'$ is linearly correlated with $U_g$ conditioned on $Z$. Now, note that $\tilde{U}_g$ is a composition of $U_g$ with $H_g$ with $U_{g1},H_{g}$ fixed. And the same holds for $\tilde{U}'_g$. Therefore there exists some affine map $\mathcal A$ such that $\mathcal A(\tilde{U}_g) = \tilde{U}'_g$. 
\end{proof}
    
    \begin{lemma}
        Conditioned on the further fixings of $(\tilde{U}_g,\tilde{U}_g')$, there exists a constant $\beta>0$ such that $SC$ is a $(1/2+\beta)$ affine somewhere random source.
    \end{lemma}
    \begin{proof}
        Let $Z=Z\cup\{\tilde{U}_g,\tilde{U}_g'\}$. It is easy to see that the bound $H(X \mid Z)\ge \delta n/2$ holds. Therefore, by Theorem~\ref{thm:Lcondmain}, $SC$ is a $(1/2+\beta_{\ref{thm:Lcondmain}})$ affine somewhere random source.
    \end{proof}
    \begin{lemma}
        Conditioned on $U_g$ is uniform as well as $\tilde{U}_g$ and $\tilde{U}_g'$ are linearly correlated, $SN_g=\snmExt(SC_1,\tilde{U}_g)\circ\cdots \circ \snmExt(SC_{\ell_1'},\tilde{U}_g)$ is $2^{-\Omega(n)}$ close to an affine somewhere random source. Moreover, there exists $h\in [\ell_1']$ such that
        \begin{align*}
            SN_{gh}\approx_{2^{-\Omega(n)}}U_{n_1}) \mid SN'_{gh}. 
        \end{align*}
    \end{lemma} 
    \begin{proof}
        Since $SC$ is a $(1/2+\beta)$ affine somewhere random source, there exists an $h\in[\ell_3]$ such that $H(SC_h)\ge 1/2+\beta$. For the seeds, the conditioning of $(U_{g1},U_{g1}')$ cause a deficiency of at most $2^{0.22n_1}$ to $\tilde{U}_g$ from being uniform. Then by Theorem~\ref{thm:snmExt} and Lemma~\ref{lemma:deficient seed}, conditioned on $(\tilde{U}_g,\tilde{U}_g')$, we have
        \begin{align*}
            (\snmExt(SC_h,\tilde{U}_g)\approx_{2^{0.22n_1}\eps_4}U_{n_1}) \mid \snmExt(SC_h',\tilde{U}_g'). 
        \end{align*}
        Since $2^{0.22n_1}\eps_4=2^{-\Omega(n)}$, $\snmExt$ is a linear function of $SC$ conditioned on $Z$, it holds that 
        \[(SN_{gh}\approx_{2^{-\Omega(n)}}U_{n_1}) \mid SN'_{gh},\]
        and $SN_g$ is $2^{-\Omega(n)}$ close to an affine somewhere random source.
    \end{proof}
    \begin{lemma}
        With probability $1-2^{-\Omega(n)}$ over further fixings of $SN'_{gh}= \snmExt(SC_h',\tilde{U}_g')$, $\tilde{Y}_g \oplus \tilde{Y}_g'$ is uniform.
    \end{lemma}
    \begin{proof}
        First note that both $X$ and $SN_g$ (similarly $X'$ and $SN_g'$) are affine sources.
        By Theorem~\ref{thm:ldacb}, we have
        \begin{align*}
            \ldACB(X,SN_{gh},h) \approx_{\eps_5} U_{n_2}\mid\underbrace{\ldACB(X',SN_{gh}',h)}_{\substack{\text{same advice but} \\ \text{$SN_{gh}\approx_{2^{-\Omega(n)}} U_{n_1}\mid SN'_{gh}$}}},\underbrace{\cbra{\ldACB(X,SN_{gj},j),\ldACB(X',SN_{gj}',j)}}_{\substack{\text{the set contains the output }\forall j\in[\ell_1']\setminus\cbra{h}\\\text{different advice}}}.
        \end{align*}
        Therefore, it holds that
        \begin{align*}
            \tilde{Y}_g \oplus \tilde{Y}_g' &=\bigoplus_{j\in[\ell_1']} \ldACB(X,SN_{gj},j)\oplus \bigoplus_{j\in[\ell_1']}\ldACB(X',SN'_{gj},j) \\
            &= \ldACB(X,SN_{gh},h) \oplus \pbra{\bigoplus_{j\in[\ell_1']\setminus \cbra{h}}\ldACB(X,SN_{gj},j) \oplus \bigoplus_{j\in[\ell_1']}\ldACB(X',SN'_{gj},j)} \\
            &\approx_{2^{-\Omega(n)}} U_{n_2}.
        \end{align*}
        
    \end{proof}
    Let $Z=Z\cup\{SN'_{gh}\}.$
    \begin{lemma}
        With probability $1-2^{-\Omega(n)}$ over the fixing of $(Z,\tilde{Y}_g,\tilde{Y}'_g)$, $W_g$ is uniform conditioned on $W_g'$. Moreover, there exists a random variable $\hat{B}_g$ conditioned on which  $X$ is an affine function of $W_g$.
    \end{lemma}
    \begin{proof}
        By Lemma~\ref{lemma:ind-merging-affine}, $W_g \approx_{2^{-{\Omega(n)}}} U_{n_3} \mid W_g'$. Since conditioned on $(\tilde{Y}_g,\tilde{Y}'_g)$, $W_g$ is a linear function of $X$, by Lemma~\ref{lemma:affine conditioning}, there exists $\hat{A}_g$ and $\hat{B}_g$ such that $W_g(X)=W_g(\hat{A}_g)$ and $W_g(\hat{B}_g)=0$. Moreover, conditioned on the fixing of $\hat{B}_g$, $X$ is an affine function of $W_g$. Now, let $Z=Z\cup\{\tilde{Y}_g,\tilde{Y}'_g,W_g',\hat{B}_g\}$.
    \end{proof}

    \begin{lemma}
        For all $1\le j \le m_1$, $z_j$ is a constant degree polynomial of the bits of $x$.
    \end{lemma}
    \begin{proof}
        \hfill \break
        We consider how the degree of the polynomial accumulates inside the for-loop of Algorithm~\ref{alg:DAExt}.
        \begin{enumerate}
            \item According a similar argument to~\cite{Li:CCC:2011}, each bit of $u_i$ is a $O(1)$ degree polynomial of $x$.
            \item In step $7$, it suffices to consider each bit of $h_i$. For each $j\in[k]$, $\tilde{x}_j$ has $O(1)$ bits. Since to sample a bit from $\tilde{x}_j$, each $u_{i1}^{(j)}$ only needs to be $\log \abs{\tilde{x}_j}=O(1)$ long to encode all the indices of $\tilde{x}_j$. Therefore, the $j$-th bit of $h_i$ is a linear function of $\tilde{x}_j$ and a $\log \abs{\tilde{x}_j}=O(1)$ degree polynomial of $u_{i1}^{(j)}$.
            \item In step $8$, each bit of $sn_{i}$ is bilinear map on $sc$ and $\tilde{u}_i$. Therefore, each bit of $sn$ is a degree $O(1)$ polynomial of the bits of $x$.
            \item In step $9$, each bit of $\tilde{y}_i$ is a degree $O(1)$ polynomial of bits of $x$ and $sc$ according to Theorem~\ref{thm:ldacb}. 
            \item In step $10$, each bit of $w_i$ is a constant degree polynomial of the inputs by Theorem~\ref{thm:low-deg-lsext}. 
            \item In step $11$, since each $c_i$ is a constant, the degree of resulting monomials by taking products of $c_i$ bits is constant. Therefore, each bit of $v_{ij}$ for all $j\in[s_i]$ is a degree $O(1)$ polynomial of $x$.
        \end{enumerate}
        Finally, it is direct that each bit of $z_j$ is a constant degree polynomial of $x$ for each $j\in[m_1]$.
    \end{proof}
    
    \begin{lemma}
        For any integer $s\in [m_1]$, let $\tilde{Z} = (\tilde{Z}_1, \cdots, \tilde{Z}_{s})$ where $\tilde{Z}_j = \bigoplus^t_{i=g} V_{ij}$. Then conditioned on any fixing of $W_g'$ and $\hat{B}_g=\hat{b}$, there exists some $b\in \bin^{s}$ such that \[ \Big|\Supp(\DAExt(X) \mid \DAExt(X')=b)\Big| = 2^{s}.\]
    \end{lemma}
\begin{proof}
First note that for all $i\in [t]$, each bit of $W_i$ is a constant degree polynomial of the bit of $X$. Therefore, 
conditioned on the fixing of $Z$, for every $i \ge g+1$, each bit of $\tilde{W}_i$ and $\tilde{W}'_i$ is a degree $\le c(\delta)$ polynomial of the bits of $W_g$. Thus, for every $i\ge g+1$, the degree of the bit in $W_i$ and $W'_i$ is $c(\delta)$ multiple of the degree of the bit in $W_g$. Therefore, if 
\begin{align*}
    c_i > c(\delta) c_{i+1},\;\forall i
\end{align*}
then the degree of the polynomials $\bigoplus_{i=g+1}^t V_{ij}$ and $\bigoplus_{i=g+1}^t V'_{ij}$ is less than the degree of $V_{gj}$. Therefore, there exists a fixing of $\{V_{ij},V'_{ij}\}_{i \in [g+1-t],j\in [m]}$ such that $\tilde{Z}_j \oplus \tilde{Z}'_j$ can take both values in $\bin$. Since $V'_{gj}$ is fixed, $\tilde{Z}'_j$ is fixed as well. This ensures there exists $\{\tilde{z}_j'\}_{j\in [m]}$ such that 
that $\tilde{Z}_j \mid (\tilde{Z}_j'=\tilde{z}_j')$ can take both values in $\bin$. \\

Next we show that $\tilde{Z}_j$ take both values in $\bin$ conditioned on $\tilde{Z}_S$ where $S\subseteq [m]\setminus{\{j\}}$ where $\tilde{Z}_S$ denotes $(\tilde{Z}_i)_{i\in S}$. Assume that for some $(z_i)_{i\in S}$ such that when $(\tilde{Z}_i=z_i)_{i\in S}$, $\tilde{Z}_j$ is fixed to $z_j$, then it holds that
\begin{align*}
    P_g  = \prod_{i\in S}\pbra{\tilde{Z}_i + z_i + 1}\pbra{\tilde{Z}_j + z_j} \equiv 0.
\end{align*}
However, this cannot be true since $P_g$ has a monomial $V_{gj}$ of bits from $W_g$ that are different from the monomials of the same degree from $\tilde{Z}_i$ (if $z_j=1$). Since $W_g$ is uniform, $V_{gj}$ is nonzero with any fixings of $\tilde{Z}_S$. Therefore $P_g$ cannot always be $0$.
\end{proof}

The techniques to bootstrap an extractor from a disperser follow essentially the same line as~\cite{Li:CCC:2011}. We restate them here for the completeness of the proof.
    \begin{lemma}
        The random variables $O_1,\cdots,O_{\alpha m_1}$ form an $\eps$-biased space. 
    \end{lemma}
    \begin{proof}
        Let $\varnothing \neq T \subseteq [\alpha m_1]$, $S_i = \cbra{j \in [m_1]: \G_{ij} = 1}$, $S_T = \cbra{j \in [m_1]: \oplus_{i\in T} \G_{ij} = 1}$. Then 
        \begin{align*}
            \bigoplus_{i\in T} O_i = \bigoplus Z_{j:j\in S_T}.
        \end{align*}
        Since any non-zero linear combination of codewords is again a codeword. The set $S_T$ has cardinality at least $\gamma m_1$. Now note that conditioned on the fixing of $Z$, each $O_i$ is a degree $c_g$ polynomial of $W_g$. Moreover, $\forall i \in S_T$, 
        \begin{align*}
            O_i  = V_{gi} \oplus \bigoplus_{j=g+1}^{t} V_{ji}.
        \end{align*}
        Since for all $i\in S_T$, $V_{gi}$ is the product of some disjoint set (w.r.t. $V_{g\ell}$'s, $\forall \ell \in S_T\setminus \cbra{i}$) of the bits of $W_g$, it is easy to see that $\{V_{gi}:i\in S_T\}$ is a set of $|S_T|\ge \gamma m_1 $ independent copies of the same function which we simply refer to as $f$. Since $P:=\bigoplus_{i\in T}\pbra{\bigoplus_{j=g+1}^{t} V_{ji}}$ has degree less than $c_g-1$ and $f$ has degree $c_g$, 
        \begin{align*}
            \Cor(f,P) \le 1-2^{-c_g}.
        \end{align*}
        Then by Theorem~\ref{thm:xor lemma}, 
        \begin{align*}
            \Cor(f^{\oplus |S_T|},P) \le \exp(-\Omega(|S_T|/(4^{c_g-1}\cdot (c_g-1)))) \le 2^{-\Omega(\gamma m_1)}.
        \end{align*}
        Since our choice of $T$ is arbitrary, and there are 
        at most $2^{\alpha m_1}-1$ such choices, there exists an absolute constant $c_0$ s.t.
        \begin{align*}
            \Cor(f^{\oplus |S_T|},P) \le 2^{-c_0\gamma m_1)}
        \end{align*}
        for any $\varnothing \neq T \subseteq [\alpha m_1]$. \\
        Since $f^{\oplus |S_T|}$ is uniform, it holds that
        \begin{align*}
            \Delta(\bigoplus_{i\in T} O_i,U) \le 2^{-c_0 \gamma m_1}.
        \end{align*}
        we conclude that $O_1,\cdots,O_{\alpha m_1}$ form an $\eps$-biased space.
    \end{proof}
    Now by Lemma~\ref{lemma:eps biased}, 
    \begin{align*}
        \Delta(O - U_{\beta' m_1}) \le 2^{\beta' m_1/2} \cdot 2^{-c_0 \gamma m_1}.
    \end{align*}
    Choose $0 < \beta' \le \alpha$ s.t. $\beta' \le c_0\gamma$. Then
    \begin{align*}
        \Delta(O - U_{m}) \le 2^{-c_0 m/2}.
    \end{align*}
    Therefore, the output of Algorithm~\ref{alg:DAExt} are $m$ bits that are $2^{-\Omega(m)}$-close to uniform.
\end{proof}

\subsection{Directional Affine Disperser and Extractor for Sublinear Entropy Sources}
In this subsection, we demonstrate how to push the entropy requirement of Algorithm~\ref{alg:DAExt} to sublinear. We first examine how we can do this for the disperser.

\begin{theorem}
There exists a constant $c>1$ and an efficient family of functions $\DADisp:\bin^n\to\bin^m$ such that $m=n^{\Omega(1)}$ and for every affine source $X$ with entropy $cn(\log\log n)^2/\log n$, there exists some $b \in \bin^m$ such that $|\Supp(\DAExt(X) \mid \DAExt(X+a)=b)|=2^m$. 
\end{theorem}
\begin{proof}[Proof Sketch]
The disperser construction is Algorithm~\ref{alg:DAExt} up to the phase ``Disperser to Extractor", except for now, we do not make assumptions about the entropy of the input. When pushing down the entropy requirement, we are mainly interested in $2$ binding quantities --- the bit length of $W_g$ and the degree of each bit of $\{W_i:i \in [g+1,t]\}$. \\
We first examine the length of $W_g$ in a step-by-step manner. \\
\textbf{Step 1.} In step $1$ of Algorithm~\ref{alg:DAExt}, by Theorem~\ref{thm:Lcondmain}, we get that $\ell_2=\poly(1/\delta)$, and each $Y_{gj}$ has $n/\poly(1/\delta)$ bits. \\
\textbf{Step 2.} We get $\ell_3'=\poly(1/\delta)$. \\
\textbf{Step 3.} The total number of rows in the matrix $SR_g$ is $\ell_2\ell_3'=\poly(1/\delta)$, with each row having $\delta^2n/(300t\ell_2\ell_3')=n/(\poly(1/\delta))$ bits. By Theorem~\ref{thm:IP}, the error is $2^{-n/\poly(1/\delta)}$. \\
\textbf{Step 4.} We apply $\AffineSRExt$. By Theorem~\ref{thm:affinesrext}, we get each $R_g$ has $n/\poly(1/\delta)^{O(\log (1/\delta))}$ bits, with error $2^{-n/\poly(1/\delta)^{O(\log (1/\delta))}}$. \\
\textbf{Step 5.} By Theorem~\ref{thm:low-deg-lsext}, after applying $\LSExt$, $U_g$ has $n/(1/\delta)^{O(\log (1/\delta))}$ bits with error $2^{-n/(1/\delta)^{O(\log (1/\delta))})}$. \\
\textbf{Step 8.} First note that by Theorem~\ref{thm:Lcondmain}, $SC_g$ has $n/(1/\delta)^{O(\log (1/\delta))}$ bits and $SC$ has $\ell_1'=(1/\delta)^{O(\log (1/\delta))}$ rows. By Theorem~\ref{thm:snmExt}, each row of $SN_g$ has bits $n/(1/\delta)^{O(\log (1/\delta))}$ with error $2^{-n/(1/\delta)^{O(\log (1/\delta))}}$. \\
\textbf{Step 9.} By Theorem~\ref{thm:ldacb}, $\tilde{Y}_g$ has $n/(1/\delta)^{O(\log (1/\delta))}$ bits with error $2^{-n/(1/\delta)^{O(\log (1/\delta))}}$. \\
\textbf{Step 10.} By Theorem~\ref{thm:low-deg-lsext}, after applying $\LSExt$, $W_g$ has $n/(1/\delta)^{O(\log (1/\delta))}$ bits with error $2^{-n/(1/\delta)^{O(\log (1/\delta))}}$. \\
We now check if the degrees of the polynomials produced in \textbf{Step 11} satisfy the requirements as in the analysis of Theorem~\ref{thm:daext}, which adds constraints $c_i>c(\delta)c_{i+1},\forall i$. First note that up to a sequence of fixings of r.v.s, $X$ is an affine function of $W_g$. Now by Theorem~\ref{thm:Lcondmain}, each bit of $\scond_2(X_i)$ and $\bcond^{\log t}\circ\scond_3(x)$ is a linear function of the input bits. The function $\IP$ is a degree $2$ polynomial. Therefore each bit of $SR_i$ is a degree $2$ polynomial of the input bits. Since each bit of the output of $\AffineSRExt$ is a degree $\poly(1/\delta)$ polynomial of the input bits. Therefore each bit of $R_i$ is a degree $\poly(1/\delta)$ of the bits of $W_g$. By Theorem~\ref{thm:low-deg-lsext}, each bit of $U_i$ is a constant degree polynomial of the input bits. Since $\Enc$ is linear and each bit of $\tilde{X}_i$ has $O(1)$ bits, each bit of $h_i$ is a constant degree polynomial of the inputs. By Theorem~\ref{thm:Lcondmain} and Theorem~\ref{thm:snmExt}, each bit of $SN_{ij}$ is a constant degree polynomial of the input bits. By Theorem~\ref{thm:ldacb}, each bit of $\tilde{Y}_i$ is a degree $2^{\log(\ell_1')}=2^{\log((1/\delta)^{O(\log (1/\delta))})}=(1/\delta)^{O(\log (1/\delta))}$ degree polynomial of the input bits. By Theorem~\ref{thm:low-deg-lsext}, each bit of $W_i$ is a constant degree polynomial of the input bits. 
    Thus, we conclude that for every $i\ge g+1$, each bit of $W_i$ is a degree $(1/\delta)^{O(\log (1/\delta))}$ polynomial of the bits of $W_g$. Therefore, we have 
    \begin{align*}
        c(\delta) = (1/\delta)^{O(\log (1/\delta))}.
    \end{align*}
    Since we need $c_i > c(\delta)c_{i+1}$ for every $1 \le i \le 10/\delta$, we have the following upper bound for all the $c_i$'s.
    \begin{align*}
        c(\delta)^{10/\delta} = ((1/\delta)^{O(\log (1/\delta))})^{O(1/\delta)} = (1/\delta)^{O((1/\delta)\log(1/\delta))}.
    \end{align*}
    Since each $W_i$ has $n/(1/\delta)^{O(\log(1/\delta))}$ bits, it suffices to have 
    \begin{align*}
        n/(1/\delta)^{O(\log(1/\delta))}>(1/\delta)^{O((1/\delta)\log(1/\delta))}.
    \end{align*}
    It suffices to take $\delta = c(\log\log n)^2/\log n$ for some constant $c$.
\end{proof}

We now discuss the case for the extractor.
\begin{theorem}
There exists a constant $c>1$ and an efficient family of functions $\DAExt:\bin^n\to\bin^m$ such that $m=n^{\Omega(1)}$ and for every affine source $X$ with entropy $cn(\log\log\log n)^2/\log\log n$, \[ (\DAExt(X), \DAExt(X+a))  \approx_{\eps} (U_m, \DAExt(X+a))~,\] where $\eps = 2^{-n^{\Omega(1)}}$. 
\end{theorem}
\begin{proof}[Proof Sketch]
    We follow up on the discussion for sublinear entropy disperser. Assume that the entropy is set such that we indeed obtain a disperser. Note that the disperser outputs 
    \begin{align*}
        n/\pbra{(\log(1/\delta))^{O(1/\delta)} \cdot (1/\delta)^{O((1/\delta)\log(1/\delta))}} = n/(1/\delta)^{O((1/\delta)\log(1/\delta))}
    \end{align*}
    bits. From this point, there is and only is one more constraint to consider which is on the degree of the polynomials. For the extractor, we need to guarantee that the XOR lemma from Theorem~\ref{thm:xor lemma} yields subexponential error.
    In other words, we need to guarantee
    \begin{align*}
        \frac{n/(1/\delta)^{O((1/\delta)\log(1/\delta))}}{(1/\delta)^{O((1/\delta)\log(1/\delta))} \cdot 2^{(1/\delta)^{O((1/\delta)\log(1/\delta))}}} = n^{\Omega(1)}.
    \end{align*}
    It suffices to take $\delta = c(\log\log\log n)^2/\log\log n$ for some constant $c$.
\end{proof}

%% file: ac0.tex
\section{Average-case \texorpdfstring{$\ac^0$}{Lg} Hardness for Read-Once Branching Programs}
\label{sec:ac0}

In this section, we build an $\ac^0$-computable extractor that are capable of extracting randomness from the preimage of any output of any read-once branching program of suitable size. 

We use the following two constructions of extractors in $\ac^0$ from previous works.
\begin{theorem}[\cite{ChengLi18}]\label{thm:acbfext}
For any constants $c \in \mathbb{N}$, $\delta \in (0,1]$, there exists an explicit deterministic $(k  = \delta n, \eps = 2^{-\log^c n})$-extractor $\acbfExt:\{0,1\}^{n} \rightarrow \{0,1\}^{\Omega(k)}$ that can be computed by $\AC^0$ circuits of depth $O(c)$, for any $(n, k)$-bit-fixing source. 
\end{theorem}

\begin{theorem}[\cite{papakonstantinou2016true}]\label{thm:aclext}
For any constants $c \in \mathbb{N}$, $\delta \in (0,1]$, there exists an explicit strong linear seeded $(k  = \delta n, \eps =  2^{-\log^c n})$-extractor $\acLExt:\{0,1\}^{n} \times \{0,1\}^{d} \rightarrow \{0,1\}^{\Omega(k)}$ that can be computed by $\AC^0$ circuits of depth $O(c)$, with seed length $d=O(\log^{c+1} n)$.
\end{theorem}

\subsection{\texorpdfstring{$\ac^0$}{Lg}-Computable \texorpdfstring{$t$}{Lg}-Affine Correlation Breaker}

Our construction of $\ac^0$-computable $t$-affine correlation breaker builds on the skeleton of the $t$-affine correlation breaker in~\cite{ChattopadhyayL22}, which in turn applies the standard correlation breaker in~\cite{Li:stoc:17}. Towards this, we first give an $\ac^0$-computable $\flip$, which is used as a subroutine in the standard correlation breaker. We then replace the strong seeded extractors in the standard correlation breaker and the $t$-affine correlation breaker with $\acLExt$. 

\begin{algorithm}[H]
    \caption{$\ac^0$-$\flip(x,y,b)$}
    \label{alg:flip-flop-ac0}
    \begin{algorithmic}
        \medskip
        \State \textbf{Input:} Uniform bit strings $x,y$ of length $n_1,n_1$ respectively, a bit $b$ and a circuit depth parameter $c\in \N$.
        \State \textbf{Output:} Bit string $\hat{x}$ of length $n_2$.
        \State \textbf{Parameters and Subroutines:} Let $n_2 = \Omega(n_1) \le n_1/20$ and $d = \Omega(n_1) \le n_2/10$. Let $\acLExt_1: \{0,1\}^{n_1} \times \{0,1\}^d \to \{0,1\}^d$ be $(k_1=n_1/10,\eps_1=2^{-\log^c n})$-strong linear seeded extractor from Theorem~\ref{thm:aclext}, $\acLExt_2: \{0,1\}^{n_2} \times \{0,1\}^d \to \{0,1\}^d$ be $(k_2=n_2,\eps_2 = 2^{-\log^c n})$-strong linear seeded extractors from Theorem~\ref{thm:aclext}. Let $\acLExt_3:\{0,1\}^{n_1} \times \{0,1\}^d \to \{0,1\}^{n_2}$ be a $(k_3=n_1/2,\eps_3=2^{-\log^c n})$-strong linear seeded extractor from Theorem~\ref{thm:aclext}. \\
        Let $\laExt: \{0,1\}^{n_1} \times \{0,1\}^{n_2+d}\to \{0,1\}^{2d}$ be a look-ahead extractor for an alternating extraction protocol run for $2$ rounds using $\acLExt_1,\acLExt_2$ as the seeded extractors.
        \\\hrulefill \\
        \begin{enumerate}
            \item Let $\tilde y=\Slice(y,n_2)$, $s_0 = \Slice(\tilde y,d)$, $\laExt(x,(\tilde y,s_0)) = r_0,r_1$
            \item Let $\overline{y} = \acLExt_3(y,r_b)$
            \item Let $\overline{s_0} = \Slice(\overline{y},d)$, $\laExt(x,(\overline{y}, \overline{s_0})) = \overline{r_0}, \overline{r_1}$
            \item Let $\hat y= \acLExt_3(y,\overline{r_{1-b}})$
            \item Let $y_0 = \Slice(\hat{y},d)$ 
            \item Output $\hat x = \acLExt_3(x,y_0)$
        \end{enumerate}
    \end{algorithmic}
\end{algorithm}

\begin{theorem}[$\ac^0$-$\flip$]\label{thm:ac0-flip}
    For any integer $c,n_1>0$ and any $\eps>0$, there exists an explicit function $\ac^0\text{-}\flip:\bin^{n_1} \times \bin^{n_1} \times \bin \to \bin^m$, satisfying the following: let $X$ be an independent uniform source on $n_1$ bits, and $X'$ be a random variable on $n_1$ bits arbitrarily correlated with $X$. Let $Y$ be an independent uniform source on $n_1$ bits, and $Y$ be a random variable on $n_1$ bits arbitrarily correlated with $Y$. Suppose $(X,X')$ is independent of $(Y,Y')$. If $k=\Omega(n_1)$, then $\ac^0$-$\flip$ can be computed by $\ac^0$ circuits of depth $O(c)$ and for any bit $b$, it holds that
    \begin{align*}
        \ac^0\text{-}\flip(X,Y,b) \approx_\eps U_m \mid (Y,Y').
    \end{align*}
    Furthermore, for any bits $b,b'$ with $b\neq b'$, we have
    \begin{align*}
        \ac^0\text{-}\flip(X,Y,b) \approx_\eps U_m \mid (\ac^0\text{-}\flip(X',Y',b'), Y, Y').
    \end{align*}
    where $m \ge \Omega(k)$ and $\eps = 6\cdot 2^{-\log^c n_1}$.
\end{theorem}
\begin{proof} We show that Algorithm~\ref{alg:flip-flop-ac0} is a construction of such functions.
\begin{enumerate}
    \item Since $S_0 = U_d$ and $\avgH(X)\ge k_1$, by the property of $\acLExt_1$, conditioned on the fixings of $S_0$, $R_0\approx_{\eps_1} U_d$ is a linear function of $X$, thus of $(Y,Y')$. Since $\avgH(\tilde Y \mid S_0,S_0')\ge n_2-2d \ge k_2$ and $R_0\approx_{\eps_1} U_d$, by the property of $\acLExt_2$, conditioned on the fixings of $R_0$, $S_1\approx_{\eps_1+\eps_2} U_d$ is a linear function of $Y'$, thus independent of $X$. Since $\avgH(X \mid R_0,R_0')\ge n_1k-2d\ge k_1$ and $S_1\approx_{\eps_1+\eps_2} U_d$, by the property of $\acLExt_1$, conditioned on $S_1$, $R_1\approx_{\eps_1+\eps_2+\eps_1} U_d$ is a linear function of $X$, thus independent of $(Y,Y')$.
    \item Since $R_b\approx_{\eps_1+b(\eps_1+\eps_2)} U_d$ and $R_b$ ($R_b'$) is independent of $Y$ ($Y'$) conditioned on $\cbra{S_0,S_1}$ ($\cbra{S_0',S_1'}$). Fix $(R_b,R_b')$ and $\overline{Y'}$, $\avgH(Y \mid S_0,S_0',S_1,S_1',\overline{Y'}) \ge n_1-4d-n_2 \ge k_3$, then by the property of $\acLExt_3$, $\overline{Y}\approx_{\eps_3+\eps_1+b(\eps_1+\eps_2)} U_{n_2}$.
    \item Now that $\overline{Y'}$ is fixed, we can fix $(\overline{R_0'},\overline{S_1'},\overline{R_1'})$. This only cause at most $2d$ entropy loss to $X$.
    Since $\overline{S_0}$ is a slice of $\overline{Y}$, then $\overline{S_0}\approx_{\eps_3+\eps_1+b(\eps_1+\eps_2)} U_d$. Since also $\avgH(X\mid R_0,R_0',R_1,R_1',\overline{R_0'},\overline{R_1'})$ $\ge n_1-6d \ge k_1$, by the property of $\acLExt_1$, conditioned on the fixings of $\overline{S_0}$, $\overline{R_0}\approx_{\eps_3+2\eps_1+b(\eps_1+\eps_2)} U_d$ is a linear function of $X$, thus independent of $(Y,Y')$. Since $\avgH(\overline{Y} \mid \overline{Y'},S_0,S_0',S_1,S_1,\overline{S_0},\overline{S_0'})$
    $\ge n_2 - 6d \ge k_2$ and $\overline{R_0}\approx_{\eps_3+2\eps_1+b(\eps_1+\eps_2)} U_d$, by the property of $\acLExt_2$, conditioned on the fixings of $\overline{R_0}$, $\overline{S_1}\approx_{\eps_3+(2+b)\eps_1+(b+1)\eps_2} U_d$ is a linear function of $\overline{Y}$, thus independent of $X$. Since $\avgH(X \mid R_0,R_0',R_1,R_1',\overline{R_0},\overline{R_0'},\overline{R_1'})\ge n_1 - 7d \ge k_1$ and $\overline{S_1}\approx_{\eps_3+(2+b)\eps_1+(b+1)\eps_2} U_d$, by the property of $\acLExt_1$, conditioned on the fixings of $\overline{S_1}$, $\overline{R_1}\approx_{\eps_3+(3+b)\eps_1+(b+1)\eps_2} U_d$ is a linear function of $X$, thus independent of $(Y,Y')$.
    \item From the above analysis, for all $b'\in\bin$, $\overline{R_{1-b'}} \approx_{\eps_3+3\eps_1+\eps_2} U_d \mid (Y,Y')$ and $\overline{R_{1-b}} \approx_{\eps_3+3\eps_1+\eps_2} U_d \mid (\overline{R'_{1-b}},Y,Y')$. Therefore, conditioned on the fixing of $(\overline{R_{1-b}},\overline{R'_{1-b}})$, $\hat Y$ ($\hat Y'$) is a linear function of $Y$ ($Y'$) and is therefore independent of $X$ ($X'$). By Lemma~\ref{lemma:ind-merging}, $\hat Y \approx_{2\eps_3+3\eps_1+\eps_2} U_{n_2} \mid (\hat Y',\overline{R_{1-b}},\overline{R'_{1-b}})$.
    \item Now it is easy to see that $Y_0 \approx_{2\eps_3+3\eps_1+\eps_2} U_{n_2} \mid ( Y_0,\overline{R_{1-b}},\overline{R'_{1-b}})$, and $Y_0$ $(Y_0')$ is independent of $X$ $(X')$. We also have $\avgH(X \mid R_0,R_0',R_1,R_1',\overline{R_0},\overline{R_1},\overline{R_0'},\overline{R_1'}) \ge n_1-8d \ge k_3$. By Lemma~\ref{lemma:ind-merging}, it holds that $\hat X \approx_{3\eps_3+3\eps_1+\eps_2} U_{n_1} \mid (\hat X',Y,Y')$. 
\end{enumerate}
This completes the proof of Theorem~\ref{thm:ac0-flip}.
\end{proof}

\paragraph{$\ac^0$-computable standard correlation breaker.} 
The following algorithm is a modification of the correlation breaker in~\cite{Li:stoc:17} so that it is computable by $\ac^0$ circuits.

\begin{definition}[$\acCB$]\label{def:accb}
    A function $\acCB:\bin^n \times \bin^d \times \bin^a \to \bin^m$ is a correlation breaker for entropy $k$ with error $\eps$ that can be computed by $\ac^0$ circuit of depth $c$ (or a $(k,\eps,c)$-affine correlation breaker for short) if for every $X,X'\in \bin^n$, $Y,Y' \in \bin^d$, $\alpha,\alpha' \in \bin^a$ s.t.
    \begin{itemize}
        \item $X$ is an $(n,k)$ source and $Y$ is uniform
        \item $(X,X')$ is independent of $(Y,Y')$
        \item $\alpha \neq \alpha'$
    \end{itemize}
    $\acCB$ can be computed by $\ac^0$ circuits of depth $O(c)$ and
    \begin{align*}
        \acCB(X,Y,\alpha) \approx_{\eps} U_m \mid \acaffCB(X',Y',\alpha').
    \end{align*}
    We say $\acCB$ is strong if
     \begin{align*}
        \acCB(X,Y,\alpha) \approx_{\eps} U_m \mid (\acaffCB(X',Y',\alpha'),Y',Y).
    \end{align*}
\end{definition}

\begin{algorithm}[H]
    \caption{$\acCB(x,y,id)$}
    \label{alg:CB}
    \begin{algorithmic}
        \medskip
        \State \textbf{Input:} Bit strings $x, y,id$ of length $n,d,a$ respectively.
        \State \textbf{Output:} Bit string $\hat v$ of length $m$.
        \State \textbf{Subroutines and Parameters:} \\
        Fix a constant $c$. Let $\ell = \log(a)$, $s=d/(1000(\ell+1))$, $r=s/a$, $m=\Omega(d)$. \\
        Let $\acLExt_0:\bin^n \times \bin^{s_0} \to \bin^{d_0}$ be the $\ac^0$-computable strong seeded extractor from~\ref{thm:aclext} set to extract from a $(n,d)$ source where $d=\Omega(n)$, seed $s_0=O(\log^{c+1}n)$, output $ d_0=\Omega(d)\le 0.3d$, $d_0\ge 200\ell s$ and error $\eps_n = 2^{-\log^c n}$. 
        \\
        Let $\IP:\bin^{d_0}\times \bin^{d_0}\to \bin^{d_0/6}$ be the two source extractor from Theorem~\ref{thm:IP} with error $\eps_{\IP} = 2^{-\Omega(d)}$. \\
        Let $\ac^0$-$\laExt_{2\ell+1}:\bin^d \times \bin^{d_0/6} \to \pbra{\bin^{3s}}^{2\ell+1}$ be the look-ahead extractor from Lemma~\ref{lemma:look-ahead} with the following extractors for Quentin and Wendy:
        \begin{itemize}
            \item $\acLExt_q:\bin^{d_0/6}\times \bin^{3s}\to \bin^{3s}$ be the $\ac^0$-computable strong seeded extractor from Theorem~\ref{thm:aclext} set to extract for $(d_0/6,d_0/12)$ sources with error $\eps_d = 2^{-\log^c d_0/6}$.
            \item $\acLExt_w:\bin^{d}\times \bin^{3s}\to \bin^{3s}$ be the $\ac^0$-computable strong seeded extractor from Theorem~\ref{thm:aclext} set to extract from $(d,d/4)$ sources with error $\le \eps_d = 2^{-\log^c d_0/6}$.
        \end{itemize}
        Let $\ac^0$-$\laExt_{\ell+1}:\bin^n \times \bin^{d_0/6} \to \pbra{\bin^{3s}}^{\ell+1}$ be the look-ahead extractor from Lemma~\ref{lemma:look-ahead}.
        \begin{itemize}
            \item $\acLExt_q:\bin^{d_0/6}\times \bin^{3s}\to \bin^{3s}$ be the $\ac^0$-computable strong seeded extractor from Theorem~\ref{thm:aclext} set to extract from $(d_0/6,d_0/12)$ sources with error $\eps_d = 2^{-\log^c d_0/6}$.
            \item $\acLExt_w':\bin^{n}\times \bin^{3s}\to \bin^{3s}$ be the $\ac^0$-computable strong seeded extractor from Theorem~\ref{thm:aclext} set to extract from $(n,d/4)$ sources with error $\eps_n = 2^{-\log^c n}$.
          
        \end{itemize}
        Let $\ac^0$-$\flip: \bin^{3s} \times \bin^{3s} \times \bin^a \to \bin^{r}$ be the $\ac^0$-computable $\flip$ from Theorem~\ref{thm:ac0-flip} with error $6 \cdot \eps_n$. \\
        Let $\acLExt: \bin^{3s} \times \bin^{r} \to \bin^r$ be the $\ac^0$-computable strong seeded extractor from Theorem~\ref{thm:aclext} set to extract from uniform sources with error $\eps_s= 2^{-\log^c(3s)}$. \\
        Let $\acLExt': \bin^{3s} \times \bin^{r} \to \bin^{r}$ be the $\ac^0$-computable strong seeded extractor from Theorem~\ref{thm:aclext} set to extract form uniform sources with error $\eps_s$.\\
        Let $\acLExt'': \bin^d \times \bin^{r}\to \bin^{s}$ be the $\ac^0$-computable strong seeded extractor from Theorem~\ref{thm:aclext} set to extract from a $(d,d/4)$ source with error $\eps_d$. \\
        Let $\acLExt''': \bin^n \times \bin^{s}\to \bin^{m}$ be the $\ac^0$-computable strong seeded extractor from Theorem~\ref{thm:aclext} set to extract from a $(n,d/4)$ source with error $\eps_n$.
    \end{algorithmic}
\end{algorithm}
\clearpage
\begin{breakalgo}
  \\ \\
    Let $\ac^0$-$\nipm_2$ construction be $2$-alternating extraction $\bin^{r}\times \bin^{r} \times \bin^{3s} \to \bin^r$ from Definition~\ref{def:l-look-ahead} with the following extractors for Quentin and Wendy:
        \begin{itemize}
            \item $\acLExt_q':\bin^r \times \bin^r \to \bin^{r/2}$ be the $\ac^0$-computable strong seeded extractor from Theorem~\ref{thm:aclext} set to extract from $(r,r)$ sources with error $\eps_r = 2^{-\log^c r}$.
            \item $\acLExt_w'':\bin^{3s} \times \bin^{r/2} \to \bin^{r}$ be the $\ac^0$-computable strong seeded extractor from Theorem~\ref{thm:aclext} set to extract from $(3s,s)$ sources with error $\eps_s$.
        \end{itemize}
        \hrulefill 
        \begin{enumerate}    
            \item Let $y_0\circ y_1=\Slice(y,s_0+0.3d)$ where $y_0$ has length $s_0$ and $x_1 = \acLExt_0(x,y_0)$.
            \item Compute $z = \IP(x_1,y_1)$. 
            \item Let $r_0,r_1,\cdots,r_{2\ell} = \ac^0\text{-}\laExt_{2\ell+1}(y,z)$.
            \item Let $s_0,s_1,\cdots,s_{\ell} = \ac^0\text{-}\laExt_{\ell+1}(x,z)$.
            \item Let $V^0$ be an $a\times r$ matrix whose $i$'th row is $V^{0}_i = \ac^0\text{-}\flip(s_0,r_0,\alpha_i)$ and has $r$ bits. 
            \item For $j = 1,\cdots,\ell$ do the following. Merge the matrix $v^{j-1}$ two rows by two rows: Note that $v^{j-1}$ has $a/2^{j-1}$ rows, for $i = 1,\cdots,a/2^j$, compute $\overline{v^{j-1}_i}  = \ac^0\text{-}\nipm(v^{j-1}_{2i-1}, v^{j-1}_{2i} ,r_{2j-1})$  which outputs $r$ bits, and $\tilde{v}_i^{j-1}= \acLExt(r_{2j},\overline{v^{j-1}_i})$ which has $r$ bits. Finally compute $v^j_i = \acLExt'(s_j, \tilde{v}^{j-1}_i$) which has $r$ bits.
            \item Compute $\hat{v} = \acLExt'''(x,\acLExt''(y,v^\ell))$.
            \end{enumerate}
\end{breakalgo}

\begin{theorem}[$\acCB$]
    \label{thm:acCB}
    For every constant $c$, there exists an explicit strong correlation breaker $\bin^n \times \bin^d \times \bin^a \to \bin^{\Omega(d)}$ for entropy $d$ with error $\eps = O(a \cdot 2^{-\log^c n})$, where $d= \Omega(n)$ and $a = O(\frac{d}{\log^c d})$. Moreover, the correlation breaker is computable by $\ac^0$ circuits of depth $O(c)$.
\end{theorem}
\begin{proof}
    We show that Algorithm~\ref{alg:CB} gives such a correlation breaker. We shall analyze the algorithm step by step. \\
    \textbf{Step 1.} Fix $(Y_0,Y_0')$, conditioned on this fixing, $X_1$ is a linear function of $X$ and is independent of $(Y_1,Y_1')$ and $(Y,Y')$. \\
    \textbf{Step 2.} Since $X_1 \approx_{\eps_n} U_{0.3d}$ and $Y_1 = U_{0.3d}$, by the definition of $\IP$, $Z\approx_{\eps_{\IP}} U_{d_0/6}$. \\
    \textbf{Step 3.} 
        \begin{itemize}
            \item Further fix $(Y_1,Y_1')$, conditioned on this fixing, it holds $(Z,Z')$ is a deterministic function of $(X_1,X_1')$, and thus $(Z,Z')$ is independent of $(Y,Y')$.
            \item $Z\approx_{\eps_{\IP}} U_{d_0/6}$ and $\avgH(Y \mid Y_0,Y_0',Y_1,Y_1') \ge d - 2s_0 - 2\cdot d_0 \ge 0.3d$.
            \item By Lemma~\ref{lemma:look-ahead}, since $\avgH(Y \mid Y_1,Y_1')\ge 0.3d\ge d/4+2(2\ell+1)(3s)+2\log(1/\eps_d)$ and $d_0/6\ge d_0/12 + 2(2\ell+1)(3s) + 2\log(1/\eps_d)$, we have that for any $0 \le j \le 2\ell - 1$, it holds that
            \begin{align*}
                R_{j+1} \approx_{O(\ell \eps_d)} U_{3s} \mid (Z, Z',R_0,R_0',\cdots,R_j,R_j').
            \end{align*}
            By a hybrid argument and the triangle inequality, we have that
            \begin{align}\label{eq:r-look-ahead}
                (Z,Z',R_0,R_0',\cdots,R_{2\ell},R_{2\ell}') \approx_{O(\ell^2 \eps_d)} (Z,Z',U_{3s},R_0',\cdots,U_{3s},R_{2\ell}').
            \end{align}
            where each $U_{3s}$ is independent of all the previous random variables (but may depend on later random variables). 
            \item Conditioned on the fixing of $(Z,Z')$, we have $\cbra{(R_i,R_i')}_{i\in[0,2\ell]}$ is a deterministic function of $(Y_1,Y_1')$, thus independent of $(X,X')$. 
        \end{itemize}
        \textbf{Step 4.} 
        \begin{itemize}
            \item Fix $(X_1,X_1')$, conditioned on this fixing, it holds $(Z,Z')$ is a deterministic function of $(Y_1,Y_1')$, and thus $(Z,Z')$ is independent of $(X,X')$.
            \item $Z\approx_{\eps_{\IP}} U_{d_0/6}$ and $\avgH(X \mid X_1,X_1') \ge d-2\cdot d_0 \ge 0.4d$.
            \item By Lemma~\ref{lemma:look-ahead}, since $\avgH(X \mid X_1,X_1')\ge 0.4d\ge d/4+2(\ell+1)(3s)+2\log(1/\eps_n)$ and $d_0/6\ge d_0/12 + 2(\ell+1)(3s)+2\log(1/\eps_d)$, we have that for any $0 \le j \le \ell -1$, it holds that
            \begin{align*}
                S_{j+1} \approx_{O(\ell (\eps_n+\eps_d)/2)} U_{3s} \mid (Z,Z',\cbra{S_0,S_0',\cdots,S_j,S_j'}).
            \end{align*}
            By a hybrid argument and the triangle inequality, we have that
            \begin{align}\label{s-look-ahead}
                (Z,Z',S_0,S_0',\cdots,S_{\ell},S_{\ell}') \approx_{O(\ell^2 (\eps_n+\eps_d)/2)} (Z,Z',U_{3s},S_0',\cdots,U_{3s},S_{\ell}').
            \end{align}
            where each $U_{3s}$ is independent of all the previous random variables (but may depend on later random variables).
            \item Conditioned on the fixing of $(Z,Z')$, we have $\cbra{(S_i,S_i')}_{i\in[0,\ell]}$ is a deterministic function of $(X_1,X_1')$, thus independent of $(Y,Y')$. 
        \end{itemize}
        Therefore, we conclude that conditioned on the fixing of $(X_1,X_1',Y_1,Y_1',Z,Z')$, we have $\cbra{(R_i,R_i')}_{i\in[0,2\ell]}$ is a deterministic function of $(Y,Y')$, and $\cbra{(S_i,S_i')}_{i\in[0,\ell]}$ is a deterministic function of $(X,X')$, thus they are independent. Moreover each $R_i$ and $S_i$ is close to uniform given the previous random variables. From now on, we will assume that each $R_i$ and $S_i$ are uniform (*) and add back an error of $O(\ell^2(\eps_n+\eps_d))$ in the end. Since in the algorithm and the analysis below, each $R_i$ and $S_i$ are used at most twice either as source of seed, this is sufficient. \\
        \textbf{Step 5.} By Theorem~\ref{thm:ac0-flip}, for all $i\in [a]$, $V_i^0 \approx_{O(\eps_n)} U_s$. Moreover, since $\alpha \neq \alpha'$, there exists an $i\in [a]$ such that $V_i^0 \approx_{O(\eps_n)} U_s \mid (V_i'^0,R_0,R'_0)$. Now that conditioned on the fixing of $(R_0,R_0')$, $(V^0,V'^0)$ is a deterministic function of $(S_0,S'_0)$, and thus independent of $\cbra{(R_i,R'_i)}_{i\in[2\ell]}$. \\
        \textbf{Step 6.}
        First note that the followings:
        \begin{enumerate}
            \item conditioned on the fixing of $(R_0,R'_0)$, $(V^0,V'^0)$ is a linear function of $(S_0,S_0')$.
            \item Each row of $V^0$ is close to uniform and there exists a row in $V^0$ that is close to uniform even conditioned on the corresponding row in $V'^0$.
        \end{enumerate} Along the analysis below, we prove by induction that for any $j\in[0,\ell]$, 
        \begin{enumerate}
            \item[(a)] each row of $V^j$ is close to uniform, and there exists a row in $V^j$ that is close to uniform even conditioned on the corresponding row in $V_j'$.
        \end{enumerate}
        For any $j\in[\ell]$, it holds that
        \begin{enumerate}
            \item[(b)] conditioned on the fixing of $(R_0,R_0',\cdots,R_{2j-2},R'_{2j-2})$, $(V^{j-1},V'^{j-1})$ is a linear functions of $(S_0,S_0',\cdots,S_{j-1},S'_{j-1})$.
            \item[(c)] each row of $\overline{V^{j-1}}$ ($\tilde{V}^{j-1}$) is close to uniform, and there exists a row in $\overline{V^{j-1}}$ ($\tilde{V}^{j-1}$) that is close to uniform even conditioned on the corresponding row in $\overline{V'^{j-1}}$ ($\tilde{V}'^{j-1}$).
        \end{enumerate}
        For each iteration $j\in[\ell]$, Step $6$ generates $3$ new somewhere random matrices: $\overline{V^{j-1}}$, $\tilde{V}^{j-1}$, and $V^j$ of size $(a/2^j)\times r$,$(a/2^j)\times r$, and $(a/2^j)\times r$ respectively. Each one of them has some properties: \\
        \textbf{Matrix $\overline{V^{j-1}}$.} Conditioned on the fixings of $(R_0,R_0',\cdots,R_{2j-2},R_{2j-2}')$, by our assumption (*), $R_{2j-1} = U_{3s}$. Now, condition on $(R_0,R_0',\cdots,R_{2j-1},R_{2j-1}')$, by Lemma~\ref{lemma:look-ahead}, each row of $\overline{V^{j-1}}$ is $O(2^{j-1}(\eps_s+\eps_n))$ close to uniform. Since there exists one row in $V^{j-1}$ that is close to uniform even given the corresponding row in $V'^{j-1}$, by Lemma~\ref{lemma:ind-merging}, there is one row in $\overline{V^{j-1}}$ that is close to uniform even conditioned on the same row in $\overline{V'^{j-1}}$. Moreover, conditioned on the fixing of $(R_0,R'_0,\cdots,R_{2j-1},R'_{2j-1})$, $(\overline{V^{j-1}},\overline{V'^{j-1}})$ is a linear function of $(V^{j-1},V'^{j-1})$, which, by induction hypothesis, is a linear function $(S_0,S'_0,\cdots,S_{j-1},S'_{j-1})$, and thus independent of $R_{2j}$.\\
        \textbf{Matrix $\tilde{V}^{j-1}$.} 
        First note that the $i$-th row of the matrix $\tilde{V}^{j-1}$ is obtained by using the $i$-th row of matrix $\overline{V^{j-1}}$ to extract from $S_j$, for each $i\in [a/2^j]$. In addition, conditioned on $\overline{V^{j-1}}$, $\tilde{V}^{j-1}$ is a deterministic function of $R_{2j}$. Since $R_{2j} = U_{3s} \mid (R_0,R'_0,\cdots, R_{2j-1}$ 
        $,R'_{2j-1})$ and $\avgH(R_{2j} \mid \overline{{V}^{j-1}}, \tilde{V}^{j-1}_{[u]}) \ge 3s - ur \ge 3s - ar/2^{j-1} \ge 3s - ar \ge s + \log(1/\eps_s)$ where $u\in [a/2^j-1]$, each row of $\tilde{V}^{j-1}$ is uniform by the definition of $\acLExt$. Since there is one row in $\overline{{V}^{j-1}}$ that is $O(2^{j-1}(\eps_s+\eps_n))$ close to uniform conditioned on the corresponding row in $\overline{V'^{j-1}}$, by Lemma~\ref{lemma:ind-merging}, there is also one row in $\tilde{V}^{j-1}$ that is close to uniform even conditioned on the corresponding row in $\tilde{V}'^{j-1}$. \\
        \textbf{Matrix $V^j$.} First note that the $i$-th row of the matrix $V^j$ is obtained by using the $i$-th row of matrix $\tilde{V}^{j-1}$ to extract from $S_j$, for each $i\in [a/2^j]$. In addition, conditioned on $\tilde{V}^{j-1}$, $V^{j}$ is a deterministic function of $S_{j}$. Since $S_j = U_{3s} \mid (S_0,S'_0,\cdots, S_{j-1},S'_{j-1})$ and $\avgH(S_j \mid \tilde{V}^{j-1}, V^j_{[u]}) \ge 3s - ur \ge 3s - ar/2^{j-1} \ge 3s - ar \ge s + \log(1/\eps_s)$ where $u\in [a/2^j-1]$, each row of $V^j$ is $O(2^{j}(\eps_s+\eps_n))$ close to uniform by the definition of $\acLExt'$. Since there is one row in $\tilde{V}^{j-1}$ that is close to uniform conditioned on the corresponding row in $\tilde{V}'^{j-1}$, by Lemma~\ref{lemma:ind-merging}, there is also one row in $V^j$ that is close to uniform even conditioned on the corresponding row in $V'^j$.  \\
        Setting $j=\ell$, we get  ins$V^\ell \approx_{O(2^{\ell}(\eps_s+\eps_n))} U_r \mid V'^\ell$. \\
        \textbf{Step 7.} Note that $H(Y \mid \cbra{R_i,R'_i}_{i\in[0,2\ell]}) \ge d/4+2\log(1/\eps_d)$ and $H(X \mid \cbra{S_i,S'_i}_{i\in [0,\ell]})\ge d+2\log(1/\eps_n)$, since $V^\ell \approx_{O(a(\eps_s+\eps_n))} U_{r} \mid V'^\ell$, by $2$ iterative use of Lemma~\ref{lemma:ind-merging}, it follows that $\hat V \approx_{O(\eps_n+\ell^2(\eps_n+\eps_d)+a(\eps_s+\eps_n)+\eps_d+\eps_n)} U_m \mid \hat V' \iff \hat V \approx_{O(a\eps_n)} U_m \mid \hat V'$. Since $\hat V$ is a deterministic function of $X$ conditioned on $(Y,Y',\cbra{S_i,S'_i}_{i\in [0,\ell]})$ and $\acLExt'''$ is strong, it also holds that $\hat V \approx_{O(a\eps_n)} U_m \mid (\hat V',Y,Y')$. This completes the proof of Theorem~\ref{thm:acCB}.
\end{proof}

\paragraph{$\ac^0$-computable $t$-affine correlation breaker.}
The following definition is a modification of $t$-affine correlation breaker~\cite{ChattopadhyayL22} into the $\ac^0$-computable setting.
\begin{definition}[$\acaffCB$]\label{def:acaffcb}
    A function $\acaffCB:\bin^n \times \bin^d \times \bin^a \to \bin^m$ is a $t$-affine correlation breaker for entropy $k$ with error $\eps$ that can be computed by $\ac^0$ circuit of depth $c$ (or a $(t,k,\eps,c)$-affine correlation breaker for short) if for every $X,A,B \in \bin^n$, $Y,Y^{[t]} \in \bin^d$, $Z$ and string $\alpha,\alpha^{[t]} \in \bin^a$ s.t.
    \begin{itemize}
        \item $X = A+B$
        \item $\avgH(A \mid Z) \ge k$
        \item $(Y,Z) = (U_d,Z)$
        \item $A$ is independent of $(B,Y,Y^{[t]})$ given $Z$
        \item $\alpha,\alpha^1, \cdots, \alpha^t$ be $a$-bit strings s.t. $\alpha \neq \alpha^i$ for every $i \in [t]$
    \end{itemize}
    $\acaffCB$ can be computed by $\ac^0$ circuits of depth $O(c)$ and
    \begin{align*}
        \acaffCB(X,Y,\alpha) \approx_{\eps} U_m \mid \cbra{\acaffCB(X^i,Y^i,\alpha^i)}_{i\in[t]}.
    \end{align*}
    We say $\acaffCB$ is strong if
     \begin{align*}
        \acaffCB(X,Y,\alpha) \approx_{\eps} U_m \mid (\cbra{\acaffCB(X^i,Y^i,\alpha^i)}_{i\in[t]},Y^{[t]},Y).
    \end{align*}
\end{definition}
Algorithm~\ref{alg:acaffcb} below is a construction of  strong $(t,k,\eps,c)$-affine correlation breaker.   
\begin{algorithm}[H]
    \caption{$\acaffCB(x,y,id)$}
    \label{alg:acaffcb}
    \begin{algorithmic}
        \medskip
        \State \textbf{Input:} Bit strings $x=w+z,y, id$ of length $n,d=\Omega(n),a$ respectively.
        \State \textbf{Output:} Bit string $q_{\lceil \log t\rceil}$ of length $r$.
        \State \textbf{Subroutines and Parameters:} \\
        Fix a constant $c$. Let $d_0' = O(\log^{c+1} n)$, $d_0 \le \min\cbra{k,d}/(10t+10)$, $d_x\le d_0/(2\log t)$, $r=k/(10+10t)$, $d_y= \frac{r}{4t\log t}$. \\
        Let $\acLExt:\bin^n \times \bin^{d_0'} \to \bin^{d_0}$ be the $\ac^0$-computable strong seeded extractor from Theorem~\ref{thm:aclext} with error $\eps_n = 2^{-\log^c n}$. \\
        Let $\acCB:\bin^d \times \bin^{d_0} \times \bin^a \to \bin^{d_x}$ be the $\ac^0$-computable correlation breaker from Theorem~\ref{thm:acCB} with error $\eps' = O(a\cdot 2^{-\log^c d})$. \\
        
        Let $\acLExt':\bin^n \times \bin^{d_x} \to \bin^{r}$ be the $\ac^0$-computable strong seeded extractor from Theorem~\ref{thm:aclext} with error $\eps_n = 2^{-\log^c n}$. \\
        Let $\acLExt_w:\bin^d \times \bin^{d_y} \to \bin^{d_x}$ be the $\ac^0$-computable strong seeded extractor from Theorem~\ref{thm:aclext} with error $\eps_d = 2^{-\log^c d}$. \\
        Let $\acLExt_q:\bin^r \times \bin^{d_x} \to \bin^{d_y}$ be the $\ac^0$-computable strong seeded extractor from Theorem~\ref{thm:aclext} with error $\eps_r = 2^{-\log^c r}$.
        \\\hrulefill \\
        Let $y_0 = \Slice(y,d_0')$ \\
        Let $x_0 = \acLExt(x,y_0)$ \\
        Let $y_1 = \acCB(y,x_0,\alpha)$ \\
        Let $q_0 = \acLExt'(x,y_1)$ \\
        For every $i$, $1\le i\le \lceil \log t \rceil$ do the following
        \begin{enumerate}
            \item Let $s_{i-1} = \Slice(q_{i-1},d_y)$ 
            \item Let $r_{i-1} = \acLExt_w(y,s_{i-1})$ 
            \item Let $\overline{s_i} = \acLExt_q(q_{i-1},r_{i-1})$ 
            \item Let $\overline{r_i} = \acLExt_w(y, \overline{s_i})$ 
            \item Let $q_i = \acLExt'(x, \overline{r_i})$
        \end{enumerate}
    \end{algorithmic}
\end{algorithm}

\begin{theorem}[$\acaffCB$]\label{thm:acaffcb}
    For every $c \in \N$, constant $0<\delta<1$ and $n\in\N$ and every $k,d,t,a$, there exists a constant $C$ such that if
    \begin{itemize}
        \item $k\ge \delta n$
        \item $d= \Omega(n)$ and $d\le n$
        \item $t = O(1)$
        \item $a \le C\frac{n}{\log^c (n)}$
    \end{itemize}
    then there exists a strong $\acaffCB:\bin^n\times \bin^d\times \bin^a \to \bin^m$ which is computable by depth $O(c)$ $\ac^0$ circuits
    \begin{itemize}
        \item $m = \Omega(k)$
        \item $\eps = O(2^{-\log^{c-1} k})$

    \end{itemize}
\end{theorem}
\begin{proof}
    We will prove that Algorithm~\ref{alg:acaffcb} gives such a function. \\
    First we prove that $\acaffCB$ satisfy Definition~\ref{def:acaffcb}. 
    \begin{enumerate}
        \item For all $i\in [t]$, let $X_{0,A}^i:= \acLExt(A,Y_0^i)$, $X_{0,B}^i:= \acLExt(B,Y_0^i)$, $Q_{0,A}^i:=\acLExt'$ 
        $(A,Y_1^i)$, $Q_{0,B}^i:= \acLExt'(B,Y_1^i)$. Let $Z$ be $Z_{\ref{def:acaffcb}}$ from Definition~\ref{def:acaffcb}.
        \item By definition of $\acLExt$, 
        \[X_{0,A}\approx_{\eps_n} U_{d_0} \mid (Z,Y_0,Y_0^{[t]},X_{0,B},X_{0,B}^{[t]}).\]
        \item Since $\avgH(Y \mid Z,Y_0,Y_0^{[t]},X_{0,B},X_{0,B}^{[t]}) \ge d - (t+1)d_0' \ge 9d/10$, $R_{1,A},R_{1,A}^{[t]}$ are independent of $Y,Y^{[t]}$ given $Z,Y_0,Y_0^{[t]},X_{0,B},X_{0,B}^{[t]}$, and 
        $\acCB$ is a strong correlation breaker, it holds $\forall i\in [t]$ that 
        \[Y_1 \approx_{\eps_n+\eps'} U_{d_x} \mid (Y_1^i,Z,Y_0,Y_0^{[t]}, X_{0,B},X_{0,B}^{[t]},X_0,X_0^i).\]
        \item Since conditioned on the fixing of $X_0,X_{0,B}^{[t]}$, $Y_1$ is a deterministic function of $Y$ and is independent of $X_0^{[t]}$, 
        \[Y_1 \approx_{\eps_n+\eps'} U_{d_x} \mid (Y_1^i,Z,Y_0,Y_0^{[t]}, X_{0,B},X_{0,B}^{[t]},X_0,X_0^{[t]}).\]
        \item By Lemma~\ref{lemma:ind-merging}, it holds $\forall i\in [t]$ that 
        \[Q_{0,A} \approx_{2\eps_n+\eps'} U_r \mid (Q_{0,A}^i,Z,Y_0,Y_0^{[t]}, X_{0,B},X_{0,B}^{[t]},X_0,X_0^{[t]},Y_1,Y_1^{[t]}).\]
        Since $(Q_{0,B},Q_{0,B}^{[t]})$ is independent of $Q_{0,A}$, let \[Z_0:= (Z,Y_0,Y_0^{[t]}, X_{0,B},X_{0,B}^{[t]},X_0,X_0^{[t]},Y_1,Y_1^{[t]},Q_{0,B},Q_{0,B}^{[t]}),\] it also holds that 
        \[Q_{0,A} \approx_{2\eps_n+\eps'} U_r \mid (Q_{0,A}^i,Z_0),\]
        which is equivalent to 
        \[Q_{0} \approx_{2\eps_n+\eps'} U_r \mid (Q_{0}^i,Z_0).\]
    \end{enumerate}
    \begin{claim}
     Each one of $\overline{S_i}$, $\overline{R_i}$, $Q_i$, $R_i$ is close to uniform and independent of every $\min\cbra{2^i,t}$ tampered r.v.'s.
    \end{claim}
    \begin{proof}
        For each $i\in [\lceil \log t \rceil]$, let 
        $$Z_{i,1,B} := (Z_{i-1}, S_{i-1,B},S_{i-1,B}^{[t]});\;Z_{i,2} := (Z_{i,1,B},S_{i-1},S_{i-1}^{[t]});\;Z_{i,3} := (Z_{i,2},R_{i-1},R_{i-1}^{[t]})$$
        $$Z_{i,3,B} := (Z_{i,3},\overline{S_{i,B}},\overline{S_{i,B}}^{[t]});\;Z_{i,4} := (Z_{i,3},\overline{S_{i}},\overline{S_{i}}^{[t]});\;Z_{i} := (Z_{i,4},\overline{R_{i}},\overline{R_{i}}^{[t]}),$$
        let $T_i$ be any subset of $[t]$ of size $2^i$ if $2^i\le t$, otherwise, let it be $[t]$.
        Now we define an ordering for the claims $\mathcal C$ according to which we prove by induction. The first claim is Sub-step $5$ with $i=0$. Then the claims follow the order of round $i$, Sub-step $1$; round $i$, Sub-step $2$, ..., round $i$, Sub-step $5$, round $i+1$, Sub-step $1$; round $i+1$, Sub-step $2$, ..., round $\lceil \log t \rceil$, Sub-step $5$. First note that by the above arguments, the claim in Sub-step 5 below holds for $i=0$. It is clear that claims in $\mathcal C$ of order $\le k$ implies that of order $k+1$. \\
        \textbf{Sub-step 1:} $S_{i-1} \approx_{2\eps_n+\eps'+(i-1)(2\eps_d+\eps_n+\eps_r)} U_{d_y} \mid (S_{i-1}^{T_{i-1}}, Z_{i-1})$; $S_{i-1,A} \approx_{2\eps_n+\eps'} U_{d_y} \mid (S_{i-1,A}^{T_{i-1}}, Z_{i,1,B})$ as long as the statement in Sub-step 5 holds for $i-1$. \\
        \textbf{Sub-step 2:} It holds by Lemma~\ref{lemma:ind-merging} that $R_{i-1} \approx_{2\eps_n+\eps'+\eps_d+(i-1)(2\eps_d+\eps_n+\eps_r)} U_{d_x} \mid (R_{i-1}^{T_{i-1}}, Z_{i,2})$ as long as the statement in Sub-step 1 holds and $\avgH(Y \mid Z_{i,2}) \ge 9d/10 - 2(i-1)(t+1)d_x \ge d/2$. \\
        \textbf{Sub-step 3:} It holds by Lemma~\ref{lemma:ind-merging} that $\overline{S_i} \approx_{2\eps_n+\eps'+\eps_d+\eps_r+(i-1)(2\eps_d+\eps_n+\eps_r)} U_{d_x} \mid (\overline{S_i}^{T_{i}},Z_{i,3})$; $\overline{S_{i,A}} \approx_{2\eps_n+\eps'+\eps_d+\eps_r+(i-1)(2\eps_d+\eps_n+\eps_r)} U_{d_x} \mid (\overline{S_{i,A}}^{T_{i}},Z_{i,3,B})$ as long as the statement in Sub-step 2 holds and $\avgH(Q_{i-1} \mid Z_{i,3}) \ge r- (2i-1)(t+1)d_y \ge r/2$.\\
        \textbf{Sub-step 4:} It holds by Lemma~\ref{lemma:ind-merging} that $\overline{R_i} \approx_{2\eps_n+\eps'+2\eps_d+\eps_r+(i-1)(2\eps_d+\eps_n+\eps_r)} U_{d_y} \mid (\overline{R_i}^{T_{i}},Z_{i,4})$ as long as the statement in Sub-step 3 holds and $\avgH(Y \mid Z_{i,4}) \ge 9d/10 -(2i-1)(t+1)d_x \ge d/2$. \\
        \textbf{Sub-step 5:} It holds by Lemma~\ref{lemma:ind-merging} that $Q_{i} \approx_{2\eps_n+\eps'+i(2\eps_d+\eps_n+\eps_r)} U_r \mid (Q_{i}^{T_{i}}, Z_{i})$ as long as the statement in Sub-step 4 holds and $\avgH(X \mid Z_i) \ge k - (d_0+r)(t+1)-2(i-1)(t+1)d_x \ge k/2$.

    \end{proof}
    Now, note that conditioned on $Z_{\lceil \log t \rceil}$, which contains $(\overline{R_{\lceil \log t \rceil}},\overline{R_{\lceil \log t \rceil}}^{[t]})$, $Q_{\lceil \log t \rceil}$ 
    $\approx_{O(\eps'+(\log t)\eps_n)} U_r \mid Q_{\lceil \log t \rceil}^{[t]}$. Moreover, $Q_{\lceil \log t \rceil},Q_{\lceil \log t \rceil}^{[t]}$ are deterministic functions of $X$ and are independent of $Y,Y^{[t]}$. Therefore, we have
    \[Q_{\lceil \log t \rceil} \approx_{O((a+\log t)\cdot 2^{-\log^c n})} U_r \mid (Q^{[t]}_{\lceil \log t \rceil},Y,Y^{[t]}).\]
    This completes the proof of Theorem~\ref{thm:acaffcb}.
\end{proof}

\subsection{\texorpdfstring{$\ac^0$}{Lg}-Computable Extractor for Read-Once Branching Program Sources}
\begin{algorithm}[H]
    \caption{$\acExt(x)$}
    \label{alg:acExt}
    \begin{algorithmic}
        \medskip
        \State \textbf{Input:} $x$ --- an $n$ bit string.
        \State \textbf{Output:} $z$ --- an $m$ bit string with $m = \Omega(n)$.
        \\\hrulefill
        \State \textbf{Sub-Routines and Parameters}: \\
        Let $\acbfExt:\{0,1\}^n \to \{0,1\}^{n_1}$ be a linear seeded extractor from Theorem~\ref{thm:acbfext} set to extract from min-entropy $k_1=\delta n$ with error $\eps_1 = 2^{-\log^c (n/t)}$. \\
        Let $\acaffCB:\bin^{n} \times \bin^{n_1} \times \bin^a \to \bin^m$, $a=\log(t)$, be the $t$-affine correlation breaker from Theorem~\ref{thm:acaffcb} with error $\eps_2=O(2^{-\log^{c-1}n})$.
        \\\hrulefill
        \begin{enumerate}
            \item Divide $x$ into $t=2/\delta$ blocks such that $x = x_1 \circ \cdots \circ x_t$.
            \item Let $y_1\circ \cdots \circ y_t = \acbfExt(x_1)\circ\cdots\circ\acbfExt(x_t)$ such that each $y_i$ is of length $n_1<\delta^2/100 n$ bits.
            \item Let $s$ be a $t\times m$ matrix whose $i$'th row $s_i$, is $\acaffCB(x,y_i,i)$.
            \item Output $z=\bigoplus^{t}_{j=1} s_i$.
        \end{enumerate}
    \end{algorithmic}
\end{algorithm}
\begin{theorem} \label{thm:ac0-ext}
    For any constant $0<\delta \le 1$, there exists a family of functions $\acExt:\bin^n\to\bin^m$ computable in $\ac^0$, such that for any sources $X=A+B$ where $A$ and $B$ are independent and have disjoint spans, $A$ has entropy $\delta n$ and $B$ is an almost bit-fixing source of entropy $\delta n$, $\acExt(X)\approx_\eps U_m \mid B$ for $m=\Omega(n)$ and $\eps = O(2^{-\log^{c-1} n})$.
\end{theorem}
\begin{proof}

\begin{lemma}\label{lemma:ac0-division-fixing}
    There exists $g\in [t]$ such that conditioned on the fixing of $\{B_1,\cdots,B_{g-1},A_g,B_{g+1},$
    $\cdots,B_t\}$, the followings are true.
    \begin{itemize}
        \item $X_g$ is an almost bit-fixing source of entropy rate $\delta$ and $Y_g \approx_{2^{-\log^c (n/t)}} U_{n_1}$.
        \item $Y_g$ is a deterministic function of $B$.
        \item $\cbra{Y_1,\cdots,Y_{g-1},Y_{g+1},\cdots,Y_t}$ are deterministic functions of $A$.
    \end{itemize}
\end{lemma}
\begin{proof}
Since $X = A + B$, then $X_i = A_i + B_i$ for all $i \in [t]$. Since $H_\infty(B)\ge \delta n$, and each $B_i$ is of block length $n/t$, there exists $g\in [t]$ such that $H_\infty(B_g) \ge \delta n/t = \delta^2 n/2$. Now by the extraction property of $\acbfExt$, we have $Y_g \approx_{\eps_1} U_{n_1}$.
\end{proof}

\begin{lemma}\label{lemma:ac0-entropy}
    Conditioned on the additional fixing of $Y=\cbra{Y_i}_{i\in [t]\setminus \cbra{g}}$, $H_\infty(A)\ge \delta n/4$.
\end{lemma}
\begin{proof}
Since $X_g$ has entropy at most $\delta n/2$, $H_\infty(B_g) \ge \delta^2 n/2$. Now as $A_g$ is independent of $B_g$, $H_\infty(A)\le (1-\delta)\delta n/2$. Since $|Y| \le tn_1 \le \delta n/50$, we have $\avgH(A \mid A_g, Y) \ge \delta n - (1-\delta)\delta n/2 - \delta n/50 \ge \delta n/4$.
\end{proof}

\begin{lemma}
    With probability $1-\eps_2$ over the fixings of $(A_g,\{S_i,Y_i\}_{i\in [t]},B)$, $Z\approx_{2\eps_1+\eps_2} U_m$.
\end{lemma}
\begin{proof}
    Let $Z_{\ref{def:acaffcb}} = \cbra{Y_i,B_i}_{i\in [t]\setminus \cbra{g}} \cup \cbra{A_g}$.
    By Lemma~\ref{lemma:ac0-division-fixing}, with probability $1-2^{-\log^c (n/t)}$, $Y$ is a somewhere random source. Moreover, since $A$ and $B$ are independent, we have $Y_g=U_{n_1}\mid Z$. By Lemma~\ref{lemma:ac0-entropy}, $\avgH(A \mid Z_{\ref{def:acaffcb}})\ge 4\delta/n$. By Theorem~\ref{thm:acaffcb}, $S_g \approx_{\eps_1+\eps_2} U_m \mid (\cbra{S_i}_{i\in [t]\setminus \cbra{g}},Y^{[t]})$. Since $Y_g$ is a deterministic function of $B_g$, and conditioned on $Y_g$ and $Z_{\ref{def:acaffcb}}$, $S_g$ is a deterministic function of $A$, it holds that $S_g \approx_{\eps_1+\eps_2+\eps_1} U_m \mid (\cbra{S_i,Y_i}_{i\in [t]\setminus \cbra{g}},B)$, which implies $Z\approx_{2\eps_1+\eps_2} U_m \mid B$.
\end{proof}
\end{proof}

\begin{theorem}
    For any constant $\delta>0$, let $\acExt$ be a function from Theorem~\ref{thm:ac0-ext} for $\delta_{\ref{thm:ac0-ext}} = \delta/3$ with error $\eps = 2^{-\Omega(\log^{c-1} n)}$, then
    \begin{align*}
        \mathsf{ROBP}_{2\eps}(\acExt) > 2^{(1-\delta)n}.
    \end{align*}
\end{theorem}
We prove the above theorem in two steps. First, we recall a lemma in~\cite{ChattopadhyayL:ccc:2023} and show that there exists a sum of two sources $X=A+B$ with the following $3$ properties, (1) $A$ and $B$ are supported on disjoint subsets of input bits; (2) $A$ has min-entropy $(1-\delta)n-\log s$ and $B$ has min-entropy at least $\delta n$; and (3) $B$ is an oblivious bit-fixing source. Then we show that the output of our extractor is close to uniform conditioned on the output of $\ROBP$.
\begin{lemma}[A special case of Lemma 3.1 from~\cite{ChattopadhyayL:ccc:2023}]
    \label{lemma:robp-source}
    Let $X$ be a uniform random variable over $\F_2^n$. For every read-once BP $f:\F_2^n \to \bin$ of size $s$ and every $d \in [n]$, there exists a random variable $E$, and random variables $A,B\in\F_2^n$ s.t.
    \begin{itemize}
        \item $E$ has support size at most $2s$.
        \item $X=A+B$.
        \item For every $e\in \Supp(E)$, define $A_e=A\mid_{E=e}$, $B_e = B\mid_{E=e}$, 
        Then we have 
        \begin{itemize}
            \item $A^e$ and $B^e$ are independent.
            \item $B^e$ is uniform over a subset of coordinates $V_e^{B}$ of dimension $d$.
            \item There exists a complemented subspace $V_e^{A}$ of $V_e^{B}$ such that $A_e \in V_e^{A}$. 
        \end{itemize}
        \item There exists a deterministic function $g$ s.t. $g(E,B)=f(X)$.
    \end{itemize}
\end{lemma}
Then we prove the claim below, which implies the average-case lower bound of $\ROBP$.

\begin{claim}
For any constant $\delta>0$, let $\acExt$ be a function from Theorem~\ref{thm:ac0-ext} with $\delta_{\ref{thm:ac0-ext}} = \delta/3$ outputting $1$ bit with error $\eps$, and $f:\bin^n \to \bin$ be any $\ROBP$ of size $s=2^{(1-\delta)n}$. Let $X$ be a uniform random variable over $\F_2^n$. Then 
\[(\acExt(X),B,E,f(X)) \approx_\eps (U,B,E,f(X)).\]
\end{claim}
\begin{proof}
Note that $\acExt$ is a strong $(\delta n/3, \eps)$ extractor, then by Lemma~\ref{lemma:avg-ext}, it is a $(\delta n/3 + \poly\log n,2\eps)$ average case extractor. Since
$\avgH(A \mid E) = 2^{n-\delta n/3-\log(2s)} =2\delta/3-1\ge \delta n/3 + \poly\log n$, we have
\[(\acExt(X),B,E) \approx_{2\eps} (U,B,E).\]
Since $f(X)=g(E,B)$ is a deterministic function of $E$ and $B$, we can conclude that
\[(\acExt(X),B,E,f(X)) \approx_{2\eps} (U,B,E,f(X)).\]
\end{proof}

%% file: open.tex
\section{Open Problems}\label{sec:open}
Our work leaves several natural open problems. The most obvious is to further improve the constructions of directional affine extractors and the average-case hardness for $\SROLBP$s. It would also be quite interesting to show any hardness of explicit functions for $\WROLBP$s, which appears to require new ideas. Finally, it is an interesting question to see if there exist functions in $\ac^0$ that achieve optimal hardness for $\ROBP$s, or strong hardness for $\SROLBP$s.

%% file: acknowledgement.tex
\section*{Acknowledgement}
We thank anonymous reviewers for their helpful comments and a reviewer for pointing us to~\cite{glinskih_et_al:LIPIcs.MFCS.2017.26}.

%% file: app-ac0-xor-srolbp.tex
\section{Depth \texorpdfstring{$3\;\ac^0[\oplus]$}{Lg} Circuits Can Compute Optimal Directional Affine Extractors}
\label{app:ac0}

In this section, we extend the results in~\cite{CohenT:random:2015} and prove depth $3$ $\ac^0[\oplus]$ circuits can compute optimal directional affine extractors given by the probabilistic method. 
\paragraph{Existence of Directional Affine Extractors.} We first display the optimal directional affine extractor.

\begin{claim}
    \label{claim:opt-daext}
    There exist universal constants $n_0,c$ such that the following holds. For every $\eps>0$ and $n > n_0$ there exists a directional affine extractor for dimension $k$ with bias $\eps$, $F:\F^n_2\to\F_2$, where $k = \log \frac{n}{\eps^2}+\log \log \frac{n}{\eps^2}+c$.
\end{claim}
\begin{proof}
    For the purpose of this proof, it is more convenient to work with the definition of $\DAExt$ in~\cite{GryaznovPT:CCC:2022}.
    \begin{definition}
        A boolean function $f:\F_2^n \to \F_2$ is a directional affine extractor for dimension $d$ with bias $\eps$ if for every affine subspace $X$, every non-zero $a$, it holds that 
        \[\DAExt(X)+\DAExt(X+a)\approx_\eps U_1.\]
    \end{definition} 
    This definition is equivalent to Definition~\ref{def:daext} up to a quadratic blow-up in the error. Check Appendix $B$ in~\cite{ChattopadhyayL:ccc:2023} for a proof. 
    Let $F:\F^n_2\to\F_2$ be a random function, namely, $\cbra{F(x), x \in \F^n_2}$ are fresh random bits. Fix an affine subspace $U \subseteq \F^n_2$ of dimension $k$, a non-zero $a\in \F^n_2$. Depending on whether $U+a$ coincides with $U$, there are two cases to consider. \\ 
    \textbf{Case 1.} $U+a \neq U$. For any $x_1, x_0 \in U,\;x_1\neq x_0$, since $x_1+x_0 \in U$, it holds that $x_1+a \not\in\cbra{x_0,x_0+a}$. Therefore, $\cbra{F(x)+F(x+a),x\in U}$ are independent random bits and it holds that 
    \begin{align*}
        \Pr\sbra{\frac{1}{2^k}\abs{\sum_{x\in U}(-1)^{F(x)+F(x+a)}}\ge \eps} \le 2\cdot e^{-\frac{2^k \eps^2}{2}}. \tag{Hoeffding Inequality}
    \end{align*}
    \textbf{Case 2.} $U+a = U$. For any $x_1,x_0 \in U\;x_1\neq x_0$, $x_1 \in\cbra{x_0,x_0+a} \iff x_1 = x_0+a$. If this is the case, then $F(x)+F(x+a)=F(x+a)+F((x)$. Therefore, $\{\pbra{(-1)^{F(x)+F(x+a)}+(-1)^{F(x+a)+F(x)}}/2,$
    $x\in \F_2^n\}$ are independent random variables supported on $\{-1,1\}$ and it holds that
    \begin{align*}
        \Pr\sbra{\frac{1}{2^{k}}\abs{\sum_{x\in U}(-1)^{F(x)+F(x+a)}}\ge \eps} \le 2\cdot e^{-\frac{2^{k-1} \eps^2}{2}}. \tag{Hoeffding Inequality}
    \end{align*}
    The number of pairs of affine subspaces of the same underlining linear subspace is bounded by $\binom{2^n}{2}\binom{2^n}{k}\le 2^{(k+2)n}$. Hence by Union Bound over all pairs of affine subspaces of the same underlining linear subspace, if $2^{(k+2)n}\cdot 2\cdot e^{ - \frac{2^{k} \eps^2}{2}}=2^{(k+2)n+1 - \frac{2^{k}\eps^2}{\ln 4}}< 1$ and $2^{(k+2)n}\cdot 2\cdot e^{ - \frac{2^{k-1} \eps^2}{2}}=2^{(k+2)n+1 - \frac{2^{k-1}\eps^2}{\ln 4}}< 1$ then there exists a directional affine extractor of dimension $k$ with error $\eps$. It is verified that the same choice of $k = \log \frac{n}{\eps^2}+\log \log \frac{n}{\eps^2}+c$ for some fixed constant $c$ as in~\cite{CohenT:random:2015} suffices for the above inequalities to hold.

\paragraph{Existence of Sumset Linear Injectors.}
    The following definition of sumset linear injectors slightly generalize the notion of injector in~\cite{CohenT:random:2015}. 
    They will be applied in the construction of a more ``structured" random function which is a $\DAExt$.
    \begin{definition}\label{def:injector}
        An $(n, k_1, k_2, d)$ sumset linear injector with size $m$ is a family of $d\times n$ matrices $\cbra{A_1,\cdots,A_m}$ over $\F_2$ with the following property: for every pairs of subspaces $U,V\subseteq \F_2^n$ of dimension $k_1,k_2$ respectively where $\dim(U\cap V)\le 1$, there exists an $i\in [m]$ such that $\ker(A_i)\cap (U+V) = \cbra{0}$.
    \end{definition}
    \begin{lemma}
        \label{lemma:injector}
        For every $n,k_1,k_2$ such that $2\le k_1,k_2\le n$, there exists an $(n,k_1,k_2,k_1+k_2+1)$ linear injector with size $m=n(k_1+k_2)$.
    \end{lemma}
    \begin{proof}
        Fix a pair of subspaces $U,V\subseteq \F_2^n$ of dimension $k_1,k_2$ respectively where $U\cap V = \cbra{0}$. Let $A$ be a $d\times n$ matrix such that every entry of $A$ is sampled from $\F_2$ uniformly and independently at random. For every $u+v\in (U+V)\setminus \cbra{0}$ it holds that $\Pr[A(u+v)=0]=2^{-d}$. By taking the union bound over all pairs of elements in $U\setminus \cbra{0}$ and $V\setminus \cbra{0}$, we get that 
        \begin{align*}
            \Pr\sbra{\ker(A)\cap (U+V)\neq\cbra{0}}\le 2^{k_1+k_2 - d}.
        \end{align*}
    Let $A_1,\cdots,A_m$ be $d\times n$ matrices such that the entry of each of the matrices is sampled from $\F_2$ uniformly and independently at random. By the above equation, it holds that 
    \begin{align*}
        \Pr\sbra{\forall i\in[m]\;\ker(A_i)\cap (U+V)\neq\cbra{0}}\le 2^{m(k_1+k_2 - d)}.
    \end{align*}
    The number of sum of two linear subspaces of dimension $k_1$ and $k_2$ is bounded by $\binom{2^n}{k_1}\binom{2^n}{k_2}$, which is bounded above by $2^{n(k_1+k_2)-2}$ for $k\ge 2$. Thus if $ 2^{n(k_1+k_2)-2}\cdot 2^{m(k_1+k_2-d)}<1$ there exists an $(n,k_1,k_2,d)$ linear injector with size $m$. The latter equation holds for $d=k_1+k_2+1$ and $m=n(k_1+k_2)$.
    \end{proof}
    
    \paragraph{More Structured Random Functions.} Now we apply the sumset injector to reduce the randomness used in Claim~\ref{claim:opt-daext}.
    
    \begin{lemma}
    Let $n_0,c$ be the constants from Claim~\ref{claim:opt-daext}. Let $n>n_0$ and let $k,\eps$ be such that $k = \log \frac{n}{\eps^2}+\log \log \frac{n}{\eps^2}+c$. Let $\cbra{A_1,\cdots,A_m}$ be an $(n,k,2,d)$ linear injector with size $m$. Then, there exists functions $f_1,\cdots,f_m:\F_2^d\to\F_2$ such that the function $f:\F_2^n\to\F_2$ defined by 
    \begin{equation}
        \label{eqn:daext-decom}
        f(x)=\bigoplus^m_{i=1}f_i(A_ix)
    \end{equation}
    is a directional affine extractor for dimension $k$ with bias $\eps$.
    \end{lemma}
    \begin{proof}
    The proof idea is that ``$(U+V)$-wise'' independence, where $U$ is any affine subspace and $V=\{0,a\}=\spn\{0,a\}$ for any $a\neq 0\in \F_2$ suffices for the proof of Claim~\ref{claim:opt-daext}. In other words, we only need $\cbra{f(x)}_{x\in U\cup (U+a)}$ to be independent random bits, instead of full independence over the truthtable of $f$. We now construct such a random function, and by replacing the random function in the proof of Claim~\ref{claim:opt-daext} with this newly constructed function, we find optimal directional affine extractors in a restricted class of random functions. This will enable us to argue about its complexity. 
    Let $F_1,\cdots,F_m:\F_2^d\to \F_2$ be independent random functions, that is, the random bits $\cbra{F_i(x):{i\in [m],x\in\F_2^d}}$ are independent. Define the random function $F:\F_2^n\to \F_2$ as follows
    \begin{align*}
        F(x) = \bigoplus_{i=1}^m F_i(A_ix).
    \end{align*}
    Let $(U,U+a)$ be any pair of affine subspaces of the same underlining linear subspace $U'$, let $V=\spn\{0,a\}$. By Definition~\ref{def:injector}, there exists an $i\in [m]$ such that $\ker(A_i) \cap \pbra{U'+V} = \cbra{0}$ . This implies that for every two distinct elements $u,v\in U$ it holds that $A_i(u),A_i(v),A_i(u+a),A_i(v+a)$ are pairwise distinct. 
    Otherwise we would reach the contradiction that $A_i(u+v)=0$ or $A_i(a)=0$ or $A_i(u+v+a)=0$ and thus $u + v$ or $a$ or $u+v+a$, a non-zero vector in $U'+V$, lies in $\ker(A_i)$. Since $F_i$ is a random function, and $A_i$ is injective on $ U\cup (U+a)$, the random bits $\cbra{F_i(u)}_{u\in U\cup (U+a)}$ are independent. Since for all $x\in (U \cup (U+a))$, the fresh random coin $F_i(A_ix)$ is used and only used to generate $F(x)$, it holds that 
    $F(x)$ is independent and random in $U \cup (U+a)$.
    \end{proof}
    \begin{theorem}
    \label{thm:computation-daext}
    Let $f$ be the function from Eqn.~\eqref{eqn:daext-decom}, where $\cbra{A_1,\cdots,A_{n(k+2)}}$ is the $(n,k,2,k+3)$ sumset linear injector from Lemma~\ref{lemma:injector}. Then, $f$ is a directional affine extractor of dimension $k$ and error $\eps$, where $k=\log \frac{n}{\eps^2}+\log \log \frac{n}{\eps^2}+O(1)$. Moreover,
    \begin{enumerate}
        \item $\deg(f)=\log \frac{n}{\eps^2}+\log \log \frac{n}{\eps^2}+O(1)$.
        \item   $f$ can be realized by a $\mathsf{XOR}\text{-}\mathsf{AND}\text{-}\mathsf{XOR}$ circuit of size $O((n/\eps)^2\cdot \log^3(n/\eps))$.
        \item $f$ can be realized by a De Morgan formula of size $O((n^5/\eps^2)\cdot \log^3(n/\eps))$.
    \end{enumerate}
    \end{theorem}
    \begin{proof}
        Exactly the same as~\cite{CohenT:random:2015}. 
    \end{proof}
\end{proof}

%% file: app-missing-proofs.tex
\section{Missing Proofs}
\label{app:missing-proofs}

\subsection{Proof of Lemma~\ref{lemma:affine conditioning}}

We recall Lemma~\ref{lemma:affine conditioning}.
\affineconditioning*
\begin{proof}
We prove the second and fourth bullet points.\\
\textbf{$2$nd bullet point.} Let $L = \overline{L}+ c'$ where $\overline{L}:\{0,1\}^n \to \{0,1\}^m$ is a linear function. Consider the set $\Supp(X) \cap \Ker(\overline{L})$ which is a linear subspace, let $B$ be this linear subspace with an arbitrary affine shift $c''$, then it holds that $L(B) = L(c'') = \overline{L}(c'')+c':=c$. Let $A=X-B$. Then $L(A) = L(X) - L(B) = \overline{L}(X) +c' - c = \overline{L}(X-(B-c''))-\overline{L}(c'')=\Supp(X)\cap \Span(\overline{L}) - \overline{L}(c'')$. \\
\textbf{$4$th bullet point.} For any $\ell \in \Supp(L(X))$, conditioned on the fixing of $L(X) =  \ell$, by the second bullet it holds that $L(X) = L(A)+ L(B) = L(A) + c = \ell$.  By the third bullet, this implies $A = L^{-1}(L(A)) = L^{-1}(\ell -c)$. Therefore, $H(X \mid_{L(X)=\ell})=H(L^{-1}(\ell -c)+B)=H(B)$, thus independent of $\ell$.
\end{proof}

\subsection{Proof of Lemma~\ref{lemma:ind-merging-affine}}
We recall Lemma~\ref{lemma:ind-merging-affine}.
\indmergingaffine*
\begin{proof}
    First note that since $X,X^T,Y,Y^{[t]}$ are linear functions of $X_0$, by Lemma~\ref{lemma:affine conditioning}, the entropy of $X$ given $X^T,Y,Y^{[t]}$ is constant. Therefore, it suffices to use shannon entropy $H$ instead of average case min-entropy $\widetilde{H}_\infty$. \\
    Conditioned on the fixings of $Y^S$, it holds that $W^S$ are linear functions of $X^S$ and therefore linear functions of $X_0$. By Lemma~\ref{lemma:affine conditioning}, there exists affine sources $A=W^S(X)$ and $B$ such that $X=A+B$. By Lemma~\ref{lemma:affine bound}, $H(B)=H(X\mid Y^S,W^S)\ge H(X)-|S|\cdot m$.
    Now further condition on $(X^T,Y,Y^{[t]\setminus S})$, we have that $H(B\mid X^T,Y,Y^{[t]\setminus S})\ge H(X\mid X^T,Y,Y^T)-|S|\cdot m\ge k$. By Proposition~\ref{prop:uniform}, it follows that with probability $1-\eps$, $\LExt(X,Y)=\LExt(B,Y)+\LExt(A,Y)=\LExt(B,Y)+\mathrm{const}= U_m$. Since $W^{T}$ is a deterministic function of $X^T$ and $Y^T$, what we have shown implies
    \begin{align*}
        W \approx_{\eps+\delta} U_m \mid (W^{S\cup T}, Y, Y^{[t]}).
    \end{align*}
\end{proof}